\documentclass[]{amsart}
\usepackage[utf8]{inputenc}
\usepackage[a4paper, margin=1in]{geometry}
\usepackage{bbm,amsmath,amsfonts,amssymb,amsthm,amscd,amsbsy,mathabx,epsfig,array,mathrsfs,tikz,tikz-cd,url, graphicx,subcaption,csquotes,pgfplots}
\usepackage[colorlinks=false,citecolor=green,linktoc=page]{hyperref}
\usetikzlibrary{patterns}
\pgfplotsset{compat=1.10}
\usepgfplotslibrary{fillbetween}
\usetikzlibrary{trees,calc,decorations.markings,decorations.pathmorphing,3d}

\theoremstyle{plain}
\newtheorem{thm}{Theorem}[section]
\newtheorem*{thm*}{\bf Theorem }
\newtheorem{lem}[thm]{Lemma}
\newtheorem{prop}[thm]{Proposition}
\newtheorem*{prop*}{\bf Proposition}
\newtheorem{cor}[thm]{Corollary}

\theoremstyle{definition}
\newtheorem{defn}[thm]{Definition}

\newtheorem{example}[thm]{Example}

\theoremstyle{remark}

\numberwithin{equation}{section}

\title[On the Homology of Unions of Certain Non-Degenerate Quadrics in General Position]{On the Homology of Unions of Certain Non-Degenerate Quadrics in General Position}

\author{Maximilian M{\"u}hlbauer}

\begin{document}
\tikzset{
    photon/.style={decorate, decoration={snake}, draw=black},
    electron/.style={draw=black, postaction={decorate},
        decoration={markings,mark=at position .55 with {\arrow[draw=black]{>}}}},
    gluon/.style={decorate, draw=black,
        decoration={coil,amplitude=4pt, segment length=5pt}}
}

\begin{abstract}
    Motivated by problems arising in the complex analysis of perturbative quantum field theory, we investigate the homology of finite unions of certain non-degenerate quadratic affine hypersurfaces of complex dimension $n$ in general position. The homology of such unions and the homology of $\mathbb{C}^{n+1}$ relative to such unions is decomposed into a direct sum of homology groups of the various possible intersections of these hypersurfaces. This allows us to compute the homology groups up to isomorphism. Furthermore, we compute an explicit set of generators for a specific arrangement of these hypersurfaces by constructing a sufficiently large CW-subcomplex of the union in question. The intersection indices of all generators with the Borel-Moore homology class of $(i\cdot\mathbb{R})^{n+1}$ in the complement of the union of surfaces are computed.
\end{abstract}

\maketitle
\tableofcontents

\section{Motivation and Introduction}\label{sec:introduction}
The motivation for this article is a physical one. In \cite{cutkosky1}, a conjecture regarding the analytical structure of holomorphic functions given by certain parameter-dependent integrals, known as Cutkosky's Theorem \cite{cutkosky}, is proven in a special case up to the computation of the involved intersection index. These functions arise in perturbative quantum field theory as so called Feynman integrals and their analytical structure is highly relevant in contemporary research (see for example \cite{analytic1}, \cite{analytic2}, \cite{analytic3}, \cite{analytic4} to name just a few). The integrals discussed in \cite{cutkosky1} are essentially of the form
\begin{equation}\label{eq:feynman_integral}
    \int_\Gamma\frac{u^{2\lambda-n-1}\cdot\Omega_{n+1}}{\prod_{i=1}^N((z+u\cdot a_i)^2-u^2\cdot r_i^2)^{\lambda_i}}.
\end{equation}
This is to be understood as follows: The natural number $n$ is assumed to be even. We denote the homogeneous coordinates of the complex projective space $\mathbb{C}\mathbb{P}^{n+1}$ by $[u:z_1:\cdots:z_{n+1}]=[u:z]$.\footnote{The reason for this notation is that the integral \eqref{eq:feynman_integral} arises from a compactification of an integral over $(i\cdot\mathbb{R})^{n+1}$ embedded in $\mathbb{C}\mathbb{P}^{n+1}$ via $z\mapsto[1:z]$. In this picture, the coordinate $u$ plays a special role and we think of the hyperplane with $u=0$ as the hyperplane at \enquote{infinity}.} The $\lambda_i\in\mathbb{C}$ are complex numbers with $\text{Re}\,\lambda_i>0$ for all $i\in\{1,\ldots,N\}$ and we set $\lambda:=\sum_{i=1}^N\lambda_i$. Furthermore, we denote
\begin{equation}
    \Omega_{n+1}:=u\cdot dz_1\wedge\cdots\wedge dz_{n+1}+\sum_{i=1}^{n+1}(-1)^i\cdot z_i\cdot du\wedge dz_1\wedge\cdots\wedge\widehat{z}_i\wedge\cdots\wedge dz_{n+1}.
\end{equation}
and the integration cycle $\Gamma$ is the closure of $(i\cdot\mathbb{R})^{n+1}$ embedded in $\mathbb{C}\mathbb{P}^{n+1}$ via the inclusion
\begin{equation}
    i:\mathbb{C}^{n+1}\hookrightarrow\mathbb{C}\mathbb{P}^{n+1}\qquad z\mapsto[1:z].
\end{equation}
The integral \eqref{eq:feynman_integral} also depends on complex parameters $a=(a_1,\ldots,a_N)\in(\mathbb{C}^{n+1})^N$ and $r=(r_1,\ldots,r_N)\in(\mathbb{C}^\times)^N$ and we seek to understand the integral as a holomorphic function in these parameters. Generally, this function is not single-valued and its analytic structure can be understood by isotopy techniques for singular integrals (see \cite{app-iso} or \cite{pham}). The holomorphic $(n+1)$-form which is integrated here is (for $\text{Re}\,\lambda$ sufficiently large) defined everywhere on $\mathbb{C}\mathbb{P}^{n+1}$, except on the set
\begin{equation}
    S(a,r):=\{[u:z]\in\mathbb{C}\mathbb{P}^{n+1} \;|\; \exists i\in\{1,\ldots,N\}:(z+u\cdot a_i)^2=u^2\cdot r_i^2\}.
\end{equation}
This set is the union of $N$ non-degenerate complex projective quadrics. Let us write the parameter space as $T:=(\mathbb{C}^{n+1})^N\times(\mathbb{C}^\times)^N$ and let $(a_\text{in},r_\text{in})\in T$ such that $S(a_\text{in},r_\text{in})\cap\mathbb{R}\mathbb{P}^{n+1}=\emptyset$. Then the integral \eqref{eq:feynman_integral} defines a holomorphic function in a neighborhood of $(a_\text{in},r_\text{in})$. It is well-known (see \cite{cutkosky1} for a proof with slightly different terminology) that the integral \eqref{eq:feynman_integral} can be analytically continued from this neighborhood to $(a,r)\in T$ as long as $\det((a_i-a_1)(a_j-a_1))_{2\leq i,j\leq N}\neq0$ and the \enquote{finite} parts
\begin{equation}\label{eq:quadrics}
    S_i(a_i,r_i):=\{z\in\mathbb{C}^{n+1} \;|\; (z-a_i)^2=r_i^2\}, \qquad i=1,\ldots,N
\end{equation}
of the quadrics are in general position as submanifolds of $\mathbb{C}^{n+1}$. The subset $L$ of $T$ at which this is not the case is called the \textit{Landau surface} of the integral \eqref{eq:feynman_integral}. To understand this analytic continuation, we can investigate the action of the fundamental group of $T\backslash L$ on the homology groups $H_k(\mathbb{C}\mathbb{P}^{n+1},S(a,r))$ and $H_k(\mathbb{C}\mathbb{P}^{n+1}-S(a,r))$: A (homotopy class of) a loop $\gamma$ in $T\backslash L$ based at $(a,r)$ induces an ambient isotopy of $S(a,r)$ in $\mathbb{C}\mathbb{P}^{n+1}$. Applying this isotopy to a homology class $\Gamma\in H_k(\mathbb{C}\mathbb{P}^{n+1}-S(a,r))$ yields a new class $\gamma_\ast\Gamma\in H_k(\mathbb{C}\mathbb{P}^{n+1}-S(a,r))$ and the result of analytically continuing along $\gamma$ can be written as the integral of the same form as in \eqref{eq:feynman_integral} over $\gamma_\ast\Gamma$. Since $\pi_1(T\backslash L)$ is generally non-trivial, this can give rise to a rich analytical structure. If the geometric situation at a point $(a_p,r_p)\in L$ is relatively tame (as is the case for unions of non-degenerate quadrics), i.e. if $S(a_p,r_p)$ only exhibits simple pinches (which appear at isolated points), the action of a small loop around $(a_p,r_p)$ can be described as follows: There exists $(a',r')\in T\backslash L$ near $(a_p,r_p)$ and and an open set $U$ near the pinch point such that $U\cap S(a',r')$ looks like a (complex) $(n+1-m)$-sphere (with $m$ the number of hypersurfaces pinching at $(a_p,r_p)$) which vanishes as $(a',r')\to(a_p,r_p)$. In this situation, the homology group $H_{n+1}(U-S(a',r'))$ (resp. $H_{n+1}(U,S(a',r'))$) is generated by a single element $[\tilde{e}]$ (resp. $[\mathbf{e}]$). It can be shown that the problem can be localized in the sense that a small loop around $(a_p,r_p)$ takes $\Gamma\in H_{n+1}(\mathbb{C}\mathbb{P}^{n+1}-S(a,r))$ to
\begin{equation}
    \Gamma+x\cdot\tilde{e},
\end{equation}
where $x\in\mathbb{Z}$ is an integer which can be computed by the Picard-Lefschetz formula \cite{app-iso}
\begin{equation}
    x=(-1)^\frac{(n-m)(n-m+1)}{2}\cdot\langle\Gamma|\mathbf{e}\rangle.
\end{equation}
While the local geometry of $S$ in $U$ might generally not give us much information about the global behaviour, we shall see that in our case we can actually understand the homology class of $\mathbf{e}$ as a class in the relative homology $H_{n+1}(\mathbb{C}^{n+1},\bigcup_{i=1}^{N}S_i(a_i',r_i'))$ of the \enquote{finite} parts of the hypersurfaces. The goal of this article is to establish a sufficient description of this homology group to compute the intersection index and thus the result of analytically continuing \eqref{eq:feynman_integral} along the loop $\gamma$.\\
Denote $S_\text{fin}(a,r):=\bigcup_{i=1}^NS_i(a_r,r_i)\subset\mathbb{C}^{n+1}$. The two main results of this paper are the following. First, we show that if the $N$ affine quadratic hypersurfaces $S_1(a,r),\ldots,S_N(a,r)$ from \eqref{eq:quadrics} are in general position, the homology decomposes as
\begin{equation}\label{eq:homology_decomposition_intro}
    H_k(\mathbb{C}^{n+1},S_\text{fin}(a,r))\simeq\widetilde{H}_k(S_\text{fin}(a,r))\simeq\bigoplus_{\substack{I\subset\{1,\ldots,N\} \\ I\neq\emptyset}}\widetilde{H}_{k-|I|+1}(\bigcap_{i\in I}S_i(a_i,r_i)).
\end{equation}
We also give an explicit isomorphism realizing this decomposition.\footnote{It should be remarked that the isomorphism is non-canonical and requires some choices.} Second, we determine explicit generators $[\mathbf{e}_I]$ (resp. $[e_I]$) for $H_k(\mathbb{C}^{n+1},S_\text{fin}(a,r))$ (resp. $\widetilde{H}_k(S_\text{fin}(a,r))$) adapted to this decomposition for a special choice of $a$ and $r$. We do this by constructing a sufficiently large CW-complex to encompass all generators. Our choice of generators makes it possible to easily compute the intersection index of them with $(i\cdot\mathbb{R})^{n+1}$ (viewed as a Borel-Moore homology class in the complement of $S_\text{fin}(a,r)$). In fact, we shall prove that our choice of generators satisfies
\begin{equation}
    \langle(i\cdot\mathbb{R})^{n+1}|\mathbf{e}_I\rangle=\begin{cases} 1 & \text{if }|I|=1 \\ 0 & \text{if }|I|\geq2. \end{cases}
\end{equation}
This article has the following structure: First, we recap some basics, in particular with regards to the notion of the intersection index, and fix some notation in Section \ref{sec:preliminaries}. In the next Section \ref{sec:complex_spheres}, we discuss some basic topological properties of the affine quadratic hypersurfaces investigated here. Subsequently, in Subsection \ref{subsec:homology}, we establish the decomposition \eqref{eq:homology_decomposition_intro} and give an explicit description of the involved isomorphisms. Finally, in Section \ref{sec:generators}, we construct a CW-subcomplex of $\mathbb{C}^{n+1}$ which is big enough to contain all generators. The sought generators are then computed explicitly and the relevant intersection indices evaluated. In the Conclusion \ref{sec:conclusion}, we give some outlook on the analytic implications of this article and what remains to be done.

\section{Preliminaries}\label{sec:preliminaries}
We assume here that the reader is sufficiently familiar with the basic concepts of algebraic topology. Most of the results employed here can be found in one of the standard textbooks on the subject, like \cite{dold} or \cite{hatcher}. Some of the background that is more specific to the application to singular integrals is covered in \cite{pham}. For more information about the motivation from complex analysis in perturbative quantum field theory, the reader is referred to \cite{cutkosky1}. In this section, we take some time to recap some facts from algebraic topology and review the notion of the intersection index in some detail.
\subsection{Some Reminders on Algebraic Topology}
As mentioned above, the reader is assumed to be sufficiently familiar with the basics of algebraic topology. Some notions employed in this article however, are not part of the standard repertoire and therefore we recall them briefly. We also take some time to fix some notation to avoid confusion.\\
When we talk about the $k$th homology group of a topological space $X$, we generally mean its singular homology group (with coefficients in $\mathbb{Z}$), denoted by $H_k(X)$. The $k$th reduced homology group is denoted by $\widetilde{H}_k(X)$. If $c$ is a chain in $X$, we denote its homology class by $[c]$. To compute the generators in Section \ref{sec:generators} however, we use a CW-structure of the given space and we identify singular and cellular homology in the standard fashion.\\
We assume the reader to know the Mayer-Vietoris sequence. In the literature, the existence of this long exact sequence is often times only stated for a decomposition of $X$ into two subspaces $A,B\subset X$ such that the interior of $A$ and $B$ cover $X$. In practice, this condition is often not met and it is necessary to find open neighborhoods $U,V\subset X$ of $A$ and $B$ such that $U$ deformation retracts to $A$, $V$ deformation retracts to $B$ and $U\cap V$ deformation retracts to $A\cap B$. This is often tricky to do. If $X$ is a CW-complex such that $A$ and $B$ are subcomplexes however, it is always possible to find such $U$ and $V$.
\begin{thm}[Mayer-Vietoris Sequence for Subcomplexes, \cite{hatcher}]\label{thm:mayer-vietoris}
    Let $X$ be a CW-complex and let $A,B\subset X$ be two subcomplexes such that $X=A\cup B$. Then there exists a long exact sequence
    \begin{equation}
        \cdots\to H_k(A\cap B)\to H_k(A)\oplus H_k(B)\to H_k(X)\overset{\delta}\to H_{k-1}(A\cap B)\to\cdots,
    \end{equation}
    where the \textit{Mayer-Vietoris homomorphism} $\delta$ acts as follows:\footnote{The standard notation for this homomorphism in the literature is $\partial$. To distinguish this map from the boundary, we choose the symbol $\delta$ instead.} Any chain $x\in C_k(X)$ can be written as $x=u+v$ with $u\in C_k(A)$ and $v\in C_k(B)$ such that $\partial u=-\partial v$. Then $\delta(x):=\partial u$. Furthermore, there is an analogous sequence for the reduced homology groups.
\end{thm}
We also recap the notion of (co)homology with support in a family of supports. First, recall the following
\begin{defn}[Family of Supports, \cite{pham}]\label{defn:family_of_supports}
    Let $X$ be a topological space. A \textit{family of supports} in $X$ is a collection $\Phi$ of closed subsets of $X$ with the following three properties:
    \begin{enumerate}
        \item For all $A,B\in\Phi$, we have $A\cup B\in\Phi$.
        \item For all $A\in\Phi$ and $B\subset A$ closed, we have $B\in\Phi$.
        \item For all $A\in\Phi$, there exists a neighborhood of $A$ in $\Phi$.
    \end{enumerate}
\end{defn}
Of particular importance to us are the families of supports $F$ of all closed subsets and $c$ of all compact subsets.\footnote{The latter is not necessarily a family of supports due to condition (3). But for a sufficiently nice space, for example a locally compact space, this is the case. All spaces relevant in this text are sufficiently nice.}\\
Let $X$ be a topological space and let us consider formal linear combinations of singular simplices in $X$ which are locally finite, but not necessarily finite. I.e. we consider $n$-chains
\begin{equation}\label{eq:locally_finite_chain}
    c=\sum_{\sigma}n_\sigma\cdot\sigma,
\end{equation}
with $n_\sigma\in\mathbb{Z}$ and the sum is over all continuous maps $\sigma:\Delta^n\to X$. Locally finite means that for all compact sets $K\subset X$, the number of $\sigma$ in $c$ with $n_\sigma\neq0$ and $\sigma(\Delta^n)\cap K\neq\emptyset$ is finite. The support of $c$ is
\begin{equation}
    |c|:=\bigcup_{\substack{\sigma \\ n_\sigma\neq0}}\sigma(\Delta^n).
\end{equation}
Let $\Phi$ be a family of supports in $X$. It can be shown that chains as in \eqref{eq:locally_finite_chain} with support in $\Phi$ define abelian groups of chains $C_n^\Phi(X)$ and corresponding homology groups $H_n^\Phi(X)$. For $\Phi=c$, these are the regular singular homology groups and for $\Phi=F$, we obtain Borel-Moore homology $H_n^\text{BM}(X)$. A similar construction can be done for cohomology. For this article, it suffices to cover the case where $X$ is a smooth manifold. By the De Rham's Theorem, we can identify the cohomology groups $H^n(X)$ of $X$ with the De Rham cohomology groups $H^n_\mathrm{dR}(X)$ modulo torsion. We shall do this without explicit mention. As before, we can define cochain groups $C_\Phi^n(X)$ of differential $n$-forms with support (in the usual sense of support of a differential form) in $\Phi$ and therefore cohomology groups $H^n_\Phi(X)$ with support in $\Phi$. For cohomology, $\Phi=F$ recovers the regular De Rham cohomology groups, while $\Phi=c$ is cohomology of forms with compact support. Poincaré-duality generalizes to homology and cohomology with support in an arbitrary family of supports.
\begin{thm}[Generalized Poincaré Duality, \cite{pham}]
    Let $X$ be a smooth $n$-manifold and $\Phi$ a family of supports in $X$. For all $m\in\mathbb{N}$, there is an isomorphism
    \begin{equation}
        H_m^\Phi(X)\overset{\sim}\to H_\Phi^{n-m}(X).
    \end{equation}
\end{thm}
Furthermore, we have the following duality \cite{iversen}: If $X$ is an oriented $m$-manifold and $S\subset X$ a subset such that $X-S$ is locally compact with a closed embedding into $X$,\footnote{This is the case for example if $S$ is a finite union of closed submanifolds of $X$.} then there is an isomorphism
\begin{equation}\label{eq:borel_moore_duality}
    \varphi:H_m^\mathrm{BM}(X-S)\overset{\sim}\to H^{n-m}(X,S).
\end{equation}

\subsection{The Intersection Index}
The primary goal of this paper is the computation of intersection indices. Therefore we take some time to recap the notion of an intersection index (sufficiently general for our purposes) and state some of the results we employ in the course of this work for the convenience of the reader. In the literature, one can find various versions of this concept with varying degrees of generality. Our presentation is adapted to the needs of the problem we seek to solve.\\
Let $X$ be a topological space and $A,B\subset X$ two subspaces. Recall that $(X;A,B)$ is called an \textit{excisive triad} if the inclusion $i:A\hookrightarrow A\cup B$ induces an isomorphism $H_\bullet(A,A\cap B)\simeq H_\bullet(A\cup B,B)$ \cite{dold}.\footnote{In the reference, five other equivalent definitions are given.}
\begin{prop}[\cite{dold}]\label{prop:excisive_couple}
    If $A$ and $B$ are open subsets or if $X$ is a cell complex such that $A$ and $B$ are subcomplexes, then $(X;A,B)$ is an excisive triad.
\end{prop}
For an excisive triad $(X;A,B)$, we have a relative cap product (see \cite{dold}, Definition 12.1)
\begin{equation}
    \cap:H^m(X,A)\times H_n(X,A\cup B)\to H_{n-m}(X,B), \qquad m,n\in\mathbb{N}.
\end{equation}
In particular, if $X$ is connected and $B=\emptyset$, the case $m=n$ yields a cap product
\begin{equation}
    \cap:H^m(X,A)\times H_m(X,A)\to H_0(X)\simeq\mathbb{Z}, \qquad m\in\mathbb{N}.
\end{equation}
We always identify $H_0(X)$ with $\mathbb{Z}$ without explicit mention. By the duality \eqref{eq:borel_moore_duality}, we thus obtain a pairing
\begin{equation}\label{eq:intersection_pairing}
    \langle|\rangle:H_{n-m}^\mathrm{BM}(X-A)\times H_m(X,A)\to\mathbb{Z}, \qquad (x,y)\mapsto\langle x|y\rangle:=\varphi(x)\cap y.
\end{equation}
The map \eqref{eq:intersection_pairing} is called the \textit{intersection pairing} and for $x\in H_{n-m}^\mathrm{BM}(X-A)$ and $y\in H_m(X,A)$, the integer $\langle x|y\rangle$ is called the \textit{intersection index} of $x$ and $y$. The intersection index is generally hard to compute. We cover the following case in which concrete calculations can often be done.\\
Let $X$ be an oriented $n$-manifold and $S_1,S_2\subset X$ two oriented submanifolds of dimension $i$ and $n-i$ respectively. Let $x\in S_1\cap S_2$ and let $v_1,\ldots,v_i$ be a positively oriented basis of the tangent space $T_xS_1$ at $x$ and $w_1,\ldots,w_{n-i}$ a positively oriented basis of the tangent space $T_xS_2$ at $x$. If $S_1$ and $S_2$ intersect in general position at $x$ then $v_1,\ldots,v_i,w_1,\ldots,w_{n-i}$ can be viewed as a basis of $T_xX$. We say that the orientation of $S_1$ and $S_2$ \textit{match} if $v_1,\ldots,v_i,w_1,\ldots,w_{n-i}$ is positively oriented with respect to the orientation of $X$. Due to the orientation, we can view $S_1$ and $S_2$ as (Borel-Moore) homology classes.
\begin{prop}\cite{pham}\label{prop:intersection_index_disjoint_support}
    Let $X$ be an oriented manifold and suppose that $S_1,S_2\subset X$ are submanifolds (possibly with boundary) of complementary dimensions in general position. If $S_1$ is compact,\footnote{In \cite{pham}, the statement is given for $S_1$ \textit{or} $S_2$ compact. We slightly weakened this statement here in accordance with our definition of the intersection index.} then $S_1\cap S_2$ consists of a finite number of points and
    \begin{equation}
        \langle S_2|S_1\rangle=N_+-N_-,
    \end{equation}
    where $N_+$ (resp. $N_-$) is the number of points in $S_1\cap S_2$ where the orientation matches (resp. does not match). In particular, the intersection index vanishes if $S_1$ and $S_2$ are disjoint.\footnote{The last part of this statement is valid in much more generality. But for our purposes, this is not necessary.}
\end{prop}
There is a useful identity with respect to the intersection index and the boundary operator in the case of differentiable manifolds. We briefly recall the notion of the \textit{Leray coboundary} (see \cite{pham} or the original work \cite{leray} for details): Let $X$ be a differentiable manifold and $S\subset X$ a closed submanifold of codimension $r$. Then choosing a closed tubular neighborhood $\nu:E\to S$ of $S$, we obtain a fiber bundle with fibers $(r-1)$-spheres. The coboundary $\delta$ fibers $(q-r)$-cycles in $S$ with these spheres. On the level of homology, this induces a homomorphism
\begin{equation}
    \delta^\ast:H_k(S)\to H_{k+r-1}(X-S).
\end{equation}
For our purposes, it is convenient to define $\delta_\ast:=(-1)^{k-1}\cdot\delta^\ast$, which is the map that we call the Leray coboundary here. This convention is adopted from \cite{pham} and simplifies the residue formulas for which the Leray coboundary is meant. As such, it agrees with the convention in the original work \cite{leray} where this notion was introduced, but differs by a factor $(-1)^{k-1}$ from the convention in \cite{app-iso}.
\begin{prop}\cite{pham}
    For all $\alpha\in H_k(S)$ and $\beta\in H_{2n-k-1}(X,S)$, we have
    \begin{equation}\label{eq:intersection_dualty_1}
        \langle\alpha|\partial_\ast\beta\rangle=\langle\delta^\ast\alpha|\beta\rangle=(-1)^{k+1}\cdot\langle\delta_\ast\alpha|\beta\rangle.
    \end{equation}
\end{prop}
It is sometimes (and in particular in this article) necessary to consider an iterated construction to generalize to the case where $S$ is a finite union of closed submanifolds. This requires us to distinguish boundaries and coboundaries for different manifolds. Therefore, we drop the $\ast$ from our notation and write $\delta:=\delta_\ast$. Let $S=\bigcup_{i=1}^NS_i$ with each $S_i$ being a closed submanifold of codimension 2 of $X$. If the $S_i$ are in general position, the intersections $\bigcap_{i\in I}S_i$ are closed manifolds of codimension $2|I|$ for any non-empty $I\subset\{1,\ldots,N\}$. Thus, we obtain a sequence of maps
\begin{equation}\label{eq:iterated_boundary}
    H_{k-N}(\bigcap_{i=1}^NS_i)\overset{\delta_N}\to H_{k-N+1}(\bigcap_{i=1}^{N-1}S_i-S_N)\overset{\delta_{N-1}}\to\cdots\overset{\delta_2}\to H_{k-1}(S_1-\bigcup_{i=2}^NS_i)\overset{\delta_1}\to H_k(X-\bigcup_{i=1}^NS_i).
\end{equation}
Similarly, we can consider the sequence of boundary maps
\begin{equation}\label{eq:iterated_coboundary}
    H_{k+N}(X,\bigcup_{i=1}^NS_i)\overset{\partial_1}\to H_{k+N-1}(S_1,\bigcup_{i=2}^NS_i)\overset{\partial_2}\to\cdots\overset{\partial_{N-1}}\to H_{k+1}(\bigcap_{i=1}^{N-1}S_i,S_N)\overset{\partial_N}\to H_k(\bigcap_{i=1}^NS_i),
\end{equation}
where each $\partial_i$ takes the part of the boundary that lies in $S_i$. Let
\begin{equation}
    \alpha\in H_k(\bigcap_{i=1}^NS_i) \qquad\text{and}\qquad \beta\in H_{2n-k-N}(X,\bigcup_{i=1}^NS_i).
\end{equation}
From an obvious inductive argument, we obtain the formula
\begin{equation}\label{eq:intersection_dualty_2}
    \langle\alpha|(\partial_N\circ\cdots\circ\partial_1)(\beta)\rangle=(-1)^{N(k+\frac{(N+1)}{2})}\cdot\langle(\delta_1\circ\cdots\circ\delta_N)(\alpha)|\beta\rangle
\end{equation}
by using equation \eqref{eq:intersection_dualty_1} $N$ times.

\section{Arrangements of Complex Spheres in General Position}\label{sec:complex_spheres}
As explained in the Introduction \ref{sec:introduction}, the goal of this article is to better understand the relative homology groups $H_{n+1}(\mathbb{C}^{n+1},S(a,r))$, where $S(a,r)\subset\mathbb{C}^{n+1}$ is the subspace
\begin{equation}
    S(a,r):=\{z\in\mathbb{C}^{n+1} \;|\; \prod_{i=1}^N((z-a_i)^2-r_i^2)=0\}
\end{equation}
for some given $N\in\mathbb{N}^\ast$ and $(a,r)=(a_1,\ldots,a_N,r_1,\ldots,r_N)\in(\mathbb{C}^{n+1})^N\times(\mathbb{C}^\times)^N$ (with $a_i\in\mathbb{C}^{n+1}$ and $r_i\in\mathbb{C}^\times$ for all $1\leq i\leq N$). It is the union of $N$ non-degenerate affine quadratic hypersurfaces
\begin{equation}
    S_i(a_i,r_i):=\{z\in\mathbb{C}^{n+1} \;|\; (z-a_i)^2=r_i^2\}, \qquad\qquad i=1,\ldots,N
\end{equation}
of dimension $n$. In this section, we investigate some of the elementary topological properties of such unions. In analogy with the equation for a real sphere, we adapt the terminology from \cite{pham} and call these hypersurfaces \textit{complex spheres}. More precisely, we denote the \textit{complex $n$-sphere} around $a\in\mathbb{C}^{n+1}$ with radius $r\in\mathbb{C}^\times$ by
\begin{equation}
    S_\mathbb{C}^n(a,r):=\{z\in\mathbb{C}^{n+1} \;|\; (z-a)^2=r^2\}.
\end{equation}
Since
\begin{equation}
    \partial_z((z-a)^2-r^2)=2\cdot(z-a)\neq0 \qquad \forall z\in S_\mathbb{C}^n(a,r),
\end{equation}
the complex sphere $S_\mathbb{C}^n(a,r)$ also carries the structure of a complex analytic $n$-manifold (by the holomorphic version of the Implicit Function Theorem, see for example \cite{grauert}). In particular, $S_\mathbb{C}^n(a,r)$ is a (non-finite) CW-complex (see for example \cite{munkres}). The \textit{complex unit $n$-sphere} is denoted by $S_\mathbb{C}^n:=S_\mathbb{C}^n(0,1)$. Note hat $S_\mathbb{C}^n(a,r)$ is path-connected (unless $n=0$, in which case $S_\mathbb{C}^n(a,r)$ consists of two points). Clearly, any two complex $n$-spheres $S_\mathbb{C}^n(a_1,r_1)$ and $S_\mathbb{C}^n(a_2,r_2)$ are homeomorphic via the homeomorphism
\begin{equation}\label{eq:unit_homeo}
    S_\mathbb{C}^n(a_1,r_1)\overset{\sim}\to S_\mathbb{C}^n(a_2,r_2), \qquad z\mapsto \frac{r_2}{r_1}(z-a_1)+a_2.
\end{equation}
In particular, any complex $n$-sphere $S_\mathbb{C}^n(a,r)$ is homeomorphic to the unit sphere $S_\mathbb{C}^n$. Furthermore, it is well-known that the restriction of $S_\mathbb{C}^n(a,r)$ to the hyperplane $\text{Im}\,z=\text{Im}\,a$ is a real $n$-sphere to which $S_\mathbb{C}^n(a,r)$ deformation retracts. While this is common knowledge, the proof is not often spelled out explicitly, so we repeat it here.
\begin{prop}\label{prop:deformation_retract}
    Let $S_\mathbb{C}^n(a,r)$ be a complex $n$-sphere and write $r=e^{i\phi}\cdot|r|$. Then $S_\mathbb{C}^n(a,r)$ (strongly) deformation retracts to
    \begin{equation}\label{eq:deformation_retract}
        \{z\in\mathbb{C}^{n+1} \;|\; \mathrm{Im}(e^{-i\phi}\cdot z)=\mathrm{Im}(e^{-i\phi}\cdot a),\; (\mathrm{Re}(e^{-i\phi}\cdot z)-\mathrm{Re}(e^{-i\phi}\cdot a))^2=|r|^2\}\simeq S^n.
    \end{equation}
\end{prop}
\begin{proof}
    By applying the homeomorphism \eqref{eq:unit_homeo}, we may assume $r=1$ and $a=0$, i.e. $S_\mathbb{C}^n(a,r)=S_\mathbb{C}^n$, without loss of generality. This homeomorphism sends the set \eqref{eq:deformation_retract} to the real unit $n$-sphere $S^n$ embedded in $\mathbb{C}^{n+1}$ via the inclusion $\mathbb{R}^{n+1}\hookrightarrow\mathbb{C}^{n+1},\;x\mapsto x+i\cdot0$. By a slight abuse of notation, we identify $i(S^n)$ with $S^n$.\\
    We decompose the defining equation of $S_\mathbb{C}^n$ into its real and imaginary part:
    \begin{equation}
        z^2=1 \qquad\Leftrightarrow\qquad (\text{Re}\,z)^2-(\text{Im}\,z)^2=1 \quad\land\quad \text{Re}\,z\cdot\text{Im}\,z=0.
    \end{equation}
    Now let us define the continuous map
    \begin{equation}
        f:S_\mathbb{C}^n\times[0,1]\to\mathbb{R},\qquad (z,t)\mapsto\sqrt{\frac{1+(1-t)^2\cdot(\text{Im}\,z)^2}{(\text{Re}\,z)^2}}.
    \end{equation}
    Note that for all $z\in S_\mathbb{C}^n$, we have $(\text{Re}\,z)^2,1+(1-t)^2\cdot(\text{Im}\,z)^2\geq1$, in particular $(\text{Re}\,z)^2>0$ and $1+(1-t)^2\cdot(\text{Im}\,z)^2\geq0$, so that $f$ is well-defined. Also note that for all $z\in S_\mathbb{C}^n$, we have 
    \begin{equation}
        f(z,0)=\sqrt{\frac{1+(\text{Im}\,z)^2}{(\text{Re}\,z)^2}}=\sqrt{\frac{(\text{Re}\,z)^2}{(\text{Re}\,z)^2}}=1.
    \end{equation}
    This allows us to define the continuous map
    \begin{equation}
        g:S_\mathbb{C}^n\times[0,1]\to S_\mathbb{C}^n, \qquad (z,t)\mapsto f(z,t)\cdot\text{Re}\,z+i\cdot(1-t)\cdot\text{Im}\,z.
    \end{equation}
    First, we need to verify that the image of $g$ is indeed contained in $S_\mathbb{C}^n$. For $(z,t)\in S_\mathbb{C}^n\times[0,1]$, we compute
    \begin{equation}
        \text{Re}\,g(z,t)\cdot\text{Im}\,g(z,t)=(1-t)\cdot f(z,t)\cdot\text{Re}\,z\cdot\text{Im}\,z\overset{\text{Re}\,z\cdot\text{Im}\,z=0}=0
    \end{equation}
    and
    \begin{equation}
        \begin{split}
            (\text{Re}\,g(z,t))^2-(\text{Im}\,g(z,t))^2&=(f(z,t))^2\cdot(\text{Re}\,z)^2-(1-t)^2\cdot(\text{Im}\,z)^2\\
            &=\frac{1+(1-t)^2\cdot(\text{Im}\,z)^2}{(\text{Re}\,z)^2}\cdot(\text{Re}\,z)^2-(1-t)^2\cdot(\text{Im}\,z)^2=1.
        \end{split}
    \end{equation}
    Now we show that $g$ is indeed the sought (strong) deformation retraction of $S_\mathbb{C}^n$ to $S^n$. For any $z\in S_\mathbb{C}^n$, we have
    \begin{equation}
        g(z,0)=\underbrace{f(z,0)}_{=1}\cdot\,\text{Re}\,z+i\cdot(1-0)\cdot\text{Im}\,z=\text{Re}\,z+i\cdot\text{Im}\,z=z
    \end{equation}
    and
    \begin{equation}
        g(z,1)=f(z,1)\cdot\text{Re}\,z+(1-1)\cdot\text{Im}\,z=f(z,1)\cdot\text{Re}\,z \quad\in\quad S^n.
    \end{equation}
    Here, we used $(f(z,1))^2=\frac{1}{(\text{Re}\,z)^2}$ in the last step. Furthermore, for any $(z,t)\in S^n\times[0,1]$, we have $f(z,t)=\frac{1}{|\text{Re}\,z|}$ and thus
    \begin{equation}
        g(z,t)=f(z,t)\cdot\text{Re}\,z \quad\in\quad S^n.
    \end{equation}
    This completes the proof.
\end{proof}
Using the previous Proposition \ref{prop:deformation_retract}, it is easy to compute the homology groups of any complex sphere. It is convenient for our purposes to express this result in terms of the reduced homology groups.
\begin{cor}\label{cor:homology_complex_sphere}
    For any $n\in\mathbb{N}$ and $k\in\mathbb{N}$, we have
    \begin{equation}
        \widetilde{H}_k(S_\mathbb{C}^n)\simeq\widetilde{H}_k(S^n)\simeq\begin{cases} \mathbb{Z} & \mathrm{if}\text{ }k=n \\ 0 & \mathrm{otherwise}. \end{cases}
    \end{equation}
    Furthermore, $\widetilde{H}_n(S_\mathbb{C}^n)$ is generated by a homology class which can be represented by the real $n$-sphere contained within $S_\mathbb{C}^n$.
\end{cor}
\begin{proof}
    The existence of the first isomorphism follows immediately from Proposition \ref{prop:deformation_retract}, the second is just the famous computation of the homology of the (real) sphere.
\end{proof}
For the calculation of intersection indices, we also recall the following well-known result due to É. Cartan on the self-intersection of spheres:
\begin{prop}[\cite{pham}]\label{prop:intersection_index_sphere}
    For all $k\in\mathbb{N}$, the intersection index of the real $k$-sphere $S^k$ in the complex $k$-sphere $S_\mathbb{C}^k$ with itself is
    \begin{equation}
        \langle S^k|S^k\rangle=\begin{cases} 2\cdot(-1)^\frac{k}{2} & \text{if }k\text{ even} \\ 0 & \text{if }k\text{ odd}.\end{cases}
    \end{equation}
\end{prop}
Now, we focus on finite unions of complex spheres. The next proposition shows that, as long as the involved complex spheres are in general position, all such unions of a fixed number of complex spheres are topologically essentially the same (in a strong sense).
\begin{prop}\label{prop:ambient_isotopy_of_complex_spheres}
    Let $(a,r),(a',r')\in(\mathbb{C}^{n+1})^N\times(\mathbb{C}^\times)^N$ such that
    \begin{equation}
        S_\mathbb{C}^n(a_1,r_1),\ldots,S_\mathbb{C}^n(a_N,r_N) \qquad\text{and}\qquad S_\mathbb{C}^n(a_1',r_1'),\ldots,S_\mathbb{C}^n(a_N',r_N')
    \end{equation}
    are two collections of complex $n$-spheres in general position and such that
    \begin{equation}
        \det((a_i-a_1)(a_j-a_1))_{2\leq i,j\leq N},\;\det((a_i'-a_1')(a_j'-a_1'))_{2\leq i,j\leq N}\neq0.
    \end{equation}
    Then $\bigcup_{i=1}^NS_\mathbb{C}^n(a_i,r_i)$ and $\bigcup_{i=1}^NS_\mathbb{C}^n(a_i',r_i')$ are ambient isotopic.
\end{prop}
\begin{proof}
    This is one of the main results from \cite{cutkosky1} (with slightly different terminology).
\end{proof}
While from a homological point of view, a single complex sphere is no different than a real sphere, there clearly are some topological differences. For example, the complex sphere $S_\mathbb{C}^n$ is not bounded (and in particular not compact). This allows for the following phenomenon, which make unions of complex spheres particularly nice to handle.
\begin{prop}\label{prop:intersection_of_complex_spheres}
    Let $S_\mathbb{C}^n(a_1,r_1),\ldots,S_\mathbb{C}^n(a_N,r_N)$ be complex $n$-spheres in general position such that $\det((a_i-a_1)(a_j-a_1))_{2\leq i,j\leq N}\neq0$. Then the intersection $\bigcap_{i=1}^NS_\mathbb{C}^n(a_i,r_i)$ is homeomorphic to a complex $(n-N+1)$-sphere.
\end{prop}
\begin{proof}
    The defining equations for $\bigcap_{i=1}^NS_\mathbb{C}^n(a_i,r_i)$ read
    \begin{equation}
        (z-a_1)^2=r_1^2\quad,\quad \ldots\quad ,\quad (z-a_N)^2=r_N^2.
    \end{equation}
    By Proposition \ref{prop:ambient_isotopy_of_complex_spheres}, we may assume $a_N=0$ and $a_j=i\cdot e_j$ (with $e_j$ the $j$th canonical basis vector of $\mathbb{C}^{n+1}$) for all $j\in\{1,\ldots,N-1\}$ by first applying an appropriate ambient isotopy.\footnote{In Section \ref{sec:generators}, we show that for this choice the complex spheres are indeed in general position.} Then the last equation reads $z^2=r_N^2$ and plugging this into the remaining ones yields
    \begin{equation}
        -2a_jz=r_j^2-r_N^2-a_j^2 \qquad\Leftrightarrow\qquad z_j=\frac{i}{2}(r_j^2-r_N^2+1)
    \end{equation}
    for all $j\in\{1,\ldots,N-1\}$. Hence, we have
    \begin{equation}
        \bigcap_{i=1}^NS_\mathbb{C}^n(a_i,r_i)=\{z\in\mathbb{C}^{n+1} \;|\; z^2=r_N^2,\; \forall j\in\{1,\ldots,N-1\}:z_j=\frac{i}{2}(r_j^2-r_N^2+1)\}
    \end{equation}
    and we get the desired homeomorphism by projecting down to the last $n-N+2$ components of $\mathbb{C}^{n+1}$.
\end{proof}
The fact that we do not leave the world of complex spheres by forming intersections turns out to be very useful for our computations. It allows us to apply useful induction arguments in our computation of the homology further below.\\
For our purposes, it is convenient to not only consider finite unions of complex spheres, but unions intersected with finitely many complex spheres. Due to Proposition \ref{prop:intersection_of_complex_spheres}, such intersected unions are still unions of complex spheres (of smaller dimension). We make the following
\begin{defn}
    Let $(a,r)\in(\mathbb{C}^{n+1})^N\times(\mathbb{C}^\times)^N$ and denote $S_1:=S_\mathbb{C}^n(a_1,r_1),\ldots,S_N:=S_\mathbb{C}^n(a_N,r_N)$ be $N$ complex $n$-spheres. For every $I\subset\{1,\ldots,N\}$ and $J\subset\{1,\ldots,N\}\backslash I$, we call
    \begin{equation}
        S_{a,r}^{I,J}:=\bigcap_{i\in I}S_i\cap\bigcup_{j\in J}S_j
    \end{equation}
    an \textit{arrangement of complex $n$-spheres}. Furthermore, we say the arrangement $S_{a,r}^{I,J}$ is in \textit{general position} if the spheres $S_1,\ldots,S_N$ are in general position as submanifolds of $\mathbb{C}^{n+1}$ and if
    \begin{equation}
        \det((a_i-a_1)(a_j-a_1))_{2\leq i,j\leq N}\neq0.
    \end{equation}
\end{defn}
Let $S_1:=S_\mathbb{C}^n(a_1,r_1),\ldots,S_N:=S_\mathbb{C}^n(a_N,r_N)$ be complex $n$-spheres. If the radii $r_1,\ldots,r_N$ and centers $a_1,\ldots,a_N$ are clear from the context or of no particular relevance, we often write $S^{I,J}=S_{a,r}^{I,J}$ for arrangements of these spheres in a slight abuse of notation. First, we make some elementary observations: For every $I\subset\{1,\ldots,N\}$, every $J\subset\{1,\ldots,N\}\backslash I$ and $j\in\{1,\ldots,N\}\backslash I$ we have
\begin{equation}
    S^{\emptyset,J}=\bigcup_{j\in J}S_j,\qquad S^{I,\emptyset}=\emptyset,\qquad S^{I,\{j\}}=\bigcap_{i\in I\cup\{j\}}S_i.
\end{equation}
Furthermore, let $J_1,J_2\subset J$. Then we have
\begin{equation}
    S^{I,J_1\cup J_2}=S^{I,J_1}\cup S^{I,J_2}
\end{equation}
and if $J_1\subset J_2$, we have
\begin{equation}
    S^{I,J_1}\subset S^{I,J_2}.
\end{equation}
The other way around, if $I_1\subset I_2\subset\{1,\ldots,N\}$ and $J\subset\{1,\ldots,N\}\backslash I_2$ then
\begin{equation}
    S^{I_2,J}\subset S^{I_1,J}.
\end{equation}
Also note that each $S_{a,r}^{I,J}$ is an algebraic set. If we denote $p_i(z):=(z-a_i)^2-r_i^2$ for all $1\leq i\leq N$, then $S_{a,r}^{I,J}$ is the common zero locus of the polynomials $p_i$ for $i\in I$ and $\prod_{j\in J}p_j$.
\begin{prop}\label{prop:pathconnected}
    Let $S_1,\ldots,S_N\subset\mathbb{C}^{n+1}$ be complex $n$-spheres in general position and $S^{I,J}$ an arrangement of these spheres. If $|I|<n$ or $J=\emptyset$, then $S^{I,J}$ is a path-connected space and if $n=|I|$ then $S^{I,J}$ consists of $2|J|$ points and thus has $2|J|$ path-connected components.
\end{prop}
\begin{proof}
    First, if $J=\emptyset$ then $S^{I,J}=\emptyset$, which is path-connected. Now suppose that $J\neq\emptyset$ and $|I|<n$. Let $x,y\in S^{I,J}$. Then $x\in S^{I,\{j_1\}}$ and $y\in S^{I,\{j_2\}}$ for some $j_1,j_2\in J$. If $j_1=j_2$, there is nothing to do since $S^{I,\{j_1\}}$ is homeomorphic to a complex $(n-|I|)$-sphere according to Proposition \ref{prop:intersection_of_complex_spheres}, which is a path-connected space. If $j_1\neq j_2$ on the other hand, note that $S^{I,\{j_1\}}\cap S^{I,\{j_2\}}=S^{I\cup\{j_1\},\{j_2\}}\neq\emptyset$. So let $z\in S^{I,\{j_1\}}\cap S^{I,\{j_2\}}$. Then $z\in S^{I,\{j_1\}}$ and hence there exists a path $\gamma_1$ in $S^{I,\{j_1\}}$ from $x$ to $z$. Similarly $z\in S^{I,\{j_2\}}$, so there exists a path $\gamma_2$ from $z$ to $y$. Hence, concatenating $\gamma_1$ and $\gamma_1$ yields a path from $x$ to $y$. The case $n=|I|$ is easy since in this case every $S^{I,\{j\}}$ consists simply of two disjoint points. Since $S^{I,J}$ is assumed to be in general position, the $S^{I,\{j\}}$ must be disjoint. Thus, $S^{I,J}$ is the disjoint union of $2|J|$ points.
\end{proof}

\subsection{Homology of Complex Sphere Arrangements}\label{subsec:homology}
Now we turn to the task of understanding the homology of $\mathbb{C}^{n+1}$ relative to an arrangement of complex spheres $S^{\emptyset,I}$ to at least such a degree that we can name a set of generators and compute the intersection indices of them with $(i\cdot\mathbb{R})^{n+1}$. In this subsection, we begin the task by computing the homology of $\mathbb{C}^{n+1}$ relative to any arrangement of complex spheres $S^{I,J}$ by decomposing it into a direct sum which is trivial in all dimensions but $n+1$, where it consists of summands isomorphic to $\mathbb{Z}$, one for each non-empty subset $K$ of $J$ with $|I\cup K|\leq n+1$. Having constructed such an isomorphism makes the calculations in the subsequent Section \ref{sec:generators} possible, where we determine a set of explicit generators by constructing a CW-decomposition of a subspace of $S^{I,J}$, large enough to contain all the generators.\\
First we need the following simple lemma, which is well-known (see for example \cite{hatcher}) but not often explicitly stated:
\begin{lem}\label{lem:decomposition_homology}
    Let $X$ be a contractible topological space, $A\subset X$ a non-empty subspace. Then the boundary homomorphism $\partial$ induces an isomorphism
    \begin{equation}
        \partial^\ast:H_k(X,A)\overset{\sim}\to\widetilde{H}_{k-1}(A)
    \end{equation}
    for all $k\in\mathbb{N}$.\footnote{The reader is reminded of the fact that for any non-empty space $Y$, we have $\widetilde{H}_k(Y)=0$ for all $k<0$.}
\end{lem}
\begin{proof}
    We consider the long exact sequence
    \begin{equation}\label{eq:lem_decomposition_homology_1}
        \cdots\to\widetilde{H}_k(X)\to\widetilde{H}_k(X,A)\overset{\partial_\ast}\to\widetilde{H}_{k-1}(A)\to\widetilde{H}_{k-1}(X)\to\cdots
    \end{equation}
    of the reduced homology groups of the pair $(X,A)$. Recall that $\partial_\ast$ is the homomorphism induced by the boundary homomorphism $\partial$ on the level of chains. Since $X$ is contractible, we have $\widetilde{H}_k(X)=0$ for all $k\in\mathbb{N}$. Thus, the above sequence \eqref{eq:lem_decomposition_homology_1} contains a sequence
    \begin{equation}
        0\to\widetilde{H}_k(X,A)\overset{\partial_\ast}\to\widetilde{H}_{k-1}(A)\to 0
    \end{equation}
    for each $k\in\mathbb{N}$, which immediately yields $H_k(X,A)=\widetilde{H}_k(X,A)\simeq\widetilde{H}_{k-1}(A)$ and the isomorphism is realized by the boundary $\partial_\ast$.
\end{proof}
The above lemma is useful to translate the computation of the relative homology groups we are interested in to the computation of regular homology groups. For the computation of the latter, we iteratively apply a Mayer-Vietoris sequence argument to decompose the homology into contributions from one complex $(n-|I|)$-sphere $S^{I,\{j\}}$, the union of the remaining $(n-|I|)$-spheres $S^{I,J\backslash\{j\}}$ and their intersection $S^{I,\{j\}}\cap S^{I,J\backslash\{j\}}=S^{I\cup\{j\},J\backslash\{j\}}$, which is homeomorphic to a union of complex spheres of smaller dimension. This opens up the possibility of an inductive argument.
\begin{thm}\label{thm:decomposition_homology}
    Let $S_1,\ldots,S_N\subset\mathbb{C}^{n+1}$ be complex $n$-spheres in general position and $S^{I,J}$ an arrangement of these spheres such that $|I|<n$.\footnote{The marginal cases $|I|=n$ and $|I|=n+1$ are trivial: If $|I|=n$, then $S^{I,J}$ is just $2|J|$ disjoint points. If $|I|=n+1$, then $J=\emptyset$ and thus $S^{I,J}=\mathbb{C}^{n+1}$ or $J\neq\emptyset$ and thus $S^{I,J}=\emptyset$.} Let $m:=|J|$ and write $J=\{j_1,\ldots,j_m\}$. We have a direct sum decomposition
    \begin{equation}\label{eq:decomposition_thm}
        H_{n-|I|+1}(\mathbb{C}^{n+1},S^{I,J})\simeq\widetilde{H}_{n-|I|}(S^{I,J})\;\simeq\bigoplus_{\substack{K\subset J \\ K\neq\emptyset}}\widetilde{H}_{n-|I|-|K|+1}(\bigcap_{i\in I\cup K}S_i)\;\simeq\;\begin{cases}\mathbb{Z}^{2^{|J|}-1} & \text{if }N\leq n+1 \\ \mathbb{Z}^{2^{|J|}-2} & \text{if }N=n+2\end{cases}
    \end{equation}
    for the relative homology group in degree $n-|I|+1$ and
    \begin{equation}
        H_{k+1}(\mathbb{C}^{n+1},S^{I,J})=\widetilde{H}_k(S^{I,J})=0 \qquad \forall k\neq n-|I|
    \end{equation}
    for the remaining degrees.\\
    Furthermore, the first isomorphism in \eqref{eq:decomposition_thm} is given by the boundary $\partial_\ast$ as in Lemma \ref{lem:decomposition_homology} and the second isomorphism is given by a map $g_{I,J}$, which is recursively determined by
    \begin{equation}\label{eq:decomposition_thm_iso}
        g_{I,J}=(f_1,g_{I,J\backslash\{j_1\}}\circ f_2,g_{I\cup\{j_1\},J\backslash\{j_1\}}\circ\delta_{j_1}),\qquad g_{I,\emptyset}=\mathrm{id},
    \end{equation}
    with $\delta_{j_1}$ the Mayer-Vietoris homomorphism with respect to the decomposition $S^{I,J}=S^{I,\{j_1\}}\cup S^{I,J\backslash\{j_1\}}$ and
    \begin{equation}
        f_1:\widetilde{H}_{n-|I|+1}(S^{I,J})\to\widetilde{H}_{n-|I|+1}(S^{I,\{j_1\}}),\qquad f_2:\widetilde{H}_{n-|I|+1}(S^{I,J})\to\widetilde{H}_{n-|I|+1}(S^{I,J\backslash\{j_1\}})
    \end{equation}
    are homomorphisms with the following property: If $k:S^{I,\{j_1\}}\hookrightarrow S^{I,J}$ and $l:S^{I,J\backslash\{j_1\}}\hookrightarrow S^{I,J}$ are the natural inclusions then
    \begin{equation}
        f_1\circ k_\ast=\mathrm{id},\quad f_1\circ l_\ast=0 \qquad\text{and}\qquad f_2\circ k_\ast=0,\quad f_2\circ l_\ast=\mathrm{id}.
    \end{equation}
\end{thm}
Before we go into the proof, it should be remarked that this decomposition is reminiscent of the homology decomposition from Theorem (3) in \cite{app-iso}. In this theorem, the (co)homology of a complement $Y-\bigcup_{i\in I}\Sigma_i$ is decomposed into a direct sum of homology groups of $Y\cap\bigcap_{j\in J}\Sigma_j$ and $J$ runs over subsets of $I$. Here, $Y$ is $\mathbb{C}^n$ (viewed as embedded in $\mathbb{C}\mathbb{P}^n$) intersected with a finite number of compact complex analytic codimension 1 submanifolds $S^i$ of $\mathbb{C}\mathbb{P}^n$ and the $\Sigma_i$ are intersections of $Y$ with compact complex analytic codimension 1 submanifolds $S_i$ of $\mathbb{C}\mathbb{P}^n$. The only necessary assumption is that the $S^i$ and $S_i$ together with the hyperplane at infinity are in general position. It should be noted that our situation, where all manifolds $\Sigma_i$ are complex spheres, is not covered by this theorem. We would have to set $S_i:=S_{\mathbb{C}\mathbb{P}}^n(a_i,r_i)$\footnote{Note that in our notation, $n$ is incremented by 1.} but then the $S_i$ are not in general position at infinity (see \cite{cutkosky1} for more details and a resolution of this problem in the context of singular integrals). This is also reflected by the fact that the following proof depends on the iterated structure of complex sphere arrangements (intersections of complex spheres are complex spheres, see Proposition \ref{prop:intersection_of_complex_spheres}). We remark however that one would obtain a decomposition as in \cite{app-iso} if one could find a different smooth toric compactification such that the $S_i$ are in general position together with all the hyperplanes at infinity \cite{toric_compactification}.
\begin{proof}
    Since $\mathbb{C}^{n+1}$ is contractible, we immediately obtain from Lemma \ref{lem:decomposition_homology} that the boundary $\partial$ descends to an isomorphism
    \begin{equation}
        H_k(\mathbb{C}^{n+1},S^{I,J})\simeq\widetilde{H}_{k-1}(S^{I,J})
    \end{equation}
    for all $k\in\mathbb{N}$. Now we show that $\widetilde{H}_k(S^{I,J})=0$ for all $k\neq n-|I|$ and that
    \begin{equation}
        \widetilde{H}_{n-|I|}(S^{I,J})\simeq\bigoplus_{\substack{K\subset J \\ K\neq\emptyset}}\widetilde{H}_{n-|I|-|K|+1}(\bigcap_{i\in I\cup K}S_i),
    \end{equation}
    where the isomorphism is given by $g_{I,J}$ from equation \eqref{eq:decomposition_thm_iso}.\\
    First, we equip $S^{I,J}$ with the structure of a CW-complex such that each $S^{I',J'}$ with $I\subset I'$ and $J'\subset J$ appears as a subcomplex. We conduct the proof by induction on $|J|$. If $|J|=0$, there is nothing to do: In this case, we have $J=\emptyset$ and thus $S^{I,J}=\emptyset$. This means we trivially have
    \begin{equation}
        \widetilde{H}_{n-|I|}(S^{I,J})=0=\bigoplus_{\substack{K\subset\emptyset \\ K\neq\emptyset}}\widetilde{H}_{n-|I|-|K|+1}(\bigcap_{i\in I\cup K}S_i).
    \end{equation}
    Now suppose we already know the theorem is true up to some some $|J|-1=m-1\geq0$ and want to show that it also holds for $|J|=m$. Let us write $J=\{j_1,\ldots,j_m\}$. The reduced Mayer-Vietoris sequence (see Theorem \ref{thm:mayer-vietoris}) associated to the decomposition $S=S^{I,\{j_1\}}\cup S^{I,J\backslash\{j_1\}}$ reads
    \begin{equation}\label{eq:mayer_vietoris}
        \cdots\to\widetilde{H}_k(S^{I\cup\{j_1\},J\backslash\{j_1\}})\to\widetilde{H}_k(S^{I,\{j_1\}})\oplus\widetilde{H}_k(S^{I,J\backslash\{j_1\}})\to\tilde{H}_k(S^{I,J})\to\widetilde{H}_{k-1}(S^{I\cup\{j_1\},J\backslash\{j_1\}})\to\cdots
    \end{equation}
    We can apply the induction hypothesis to $S^{I\cup\{j_1\},J\backslash\{j_1\}}$, $S^{I,\{j_1\}}$ and $S^{I,J\backslash\{j_1\}}$ so that
    \begin{equation}
        \widetilde{H}_k(S^{I\cup\{j_1\},J\backslash\{j_1\}})=0
    \end{equation}
    for all $k\in\mathbb{N}$ such that $k\neq n-|I\cup\{j_1\}|+1=n-|I|$ and
    \begin{equation}
        \widetilde{H}_k(S^{I,\{j_1\}})=\widetilde{H}_k(S^{I,J\backslash\{j_1\}})=0
    \end{equation}
    for all $k\in\mathbb{N}$ such that $k\neq n-|I|+1$. Thus, the sequence \eqref{eq:mayer_vietoris} contains an exact sequence
    \begin{equation}
        0\to\widetilde{H}_k(S^{I,J})\to0
    \end{equation}
    for all $k\neq n-|I|+1$, which means $\widetilde{H}_k(S^{I,J})=0$ for all $k\neq n-|I|+1$. Additionally, \eqref{eq:mayer_vietoris} contains a short exact sequence
    \begin{equation}
        0\to\widetilde{H}_{n-|I|+1}(S^{I,\{j_1\}})\oplus\widetilde{H}_{n-|I|+1}(S^{I,J\backslash\{j_1\}})\overset{k_\ast-l_\ast}\to\widetilde{H}_{n-|I|+1}(S^{I,J})\overset{\delta_{j_1}}\to\widetilde{H}_{n-|I|}(S^{I\cup\{j_1\},J\backslash\{j_1\}})\to0,
    \end{equation}
    where $k:S^{I,\{j_1\}}\hookrightarrow S^{I,J}$ and $l:S^{I,J\backslash\{j_1\}}\hookrightarrow S^{I,J}$ are the natural inclusions, while $\delta_{j_1}$ is the Mayer-Vietoris homomorphism with respect to the given decomposition. Now, again by the induction hypothesis, the group $\widetilde{H}_{n-|I|}(S^{I\cup\{j_1\},J\backslash\{j_1\}})$ is free and hence this short exact sequence splits. Thus,
    \begin{equation}\label{eq:splitting}
        \widetilde{H}_{n-|I|+1}(S^{I,J})\simeq\widetilde{H}_{n-|I|+1}(S^{I,\{j_1\}})\oplus\widetilde{H}_{n-|I|+1}(S^{I,J\backslash\{j_1\}})\oplus\widetilde{H}_{n-|I|}(S^{I\cup\{j_1\},J\backslash\{j_1\}}),
    \end{equation}
    where the isomorphism is given by a map $(f_1,f_2,\delta_\ast)$ such that
    \begin{equation}
        (f_1,f_2)\circ(k_\ast-l_\ast)=\mathrm{id},
    \end{equation}
    which is equivalent to
    \begin{equation}
        f_1\circ k_\ast=\mathrm{id},\quad f_1\circ l_\ast=0 \qquad\text{and}\qquad f_2\circ k_\ast=0,\quad f_2\circ l_\ast=\mathrm{id}.
    \end{equation}
    Using the induction hypothesis one more time, we have
    \begin{equation}
        \widetilde{H}_{n-|I|+1}(S^{I,J\backslash\{j_1\}})\overset{g_{I,J\backslash\{j_1\}}}\simeq\bigoplus_{\substack{K\subset J\backslash\{j_1\} \\ K\neq\emptyset}}\widetilde{H}_{n-|I|-|K|+1}(\bigcap_{i\in I\cup K}S_i)
    \end{equation}
    and
    \begin{equation}
        \begin{split}
            \widetilde{H}_{n-|I\cup\{j_1\}|+1}(S^{I\cup\{j_1\},J\backslash\{j_1\}})&\overset{g_{I\cup\{j_1\},J\backslash\{j_1\}}}\simeq\bigoplus_{\substack{K\subset J\backslash\{j_1\} \\ K\neq\emptyset}}\widetilde{H}_{n-|I\cup\{j_1\}|-|K|+1}(\bigcap_{i\in I\cup\{j_1\}\cup K}S_i)\\
            &\overset{\phantom{g_{I\cup\{j_1\},J\backslash\{j_1\}}}}=\bigoplus_{\substack{K\subset J \\ j_1\in K,\; K\neq\emptyset,\{j_1\}}}\widetilde{H}_{n-|I|-|K|+1}(\bigcap_{i\in I\cup K}S_i).
        \end{split}
    \end{equation}
    Plugging this into equation \eqref{eq:splitting} and using $S^{I,\{j_1\}}=\bigcap_{i\in I\cup\{j_1\}}S_i$, we obtain
    \begin{equation}
        \widetilde{H}_{n-|I|+1}(S^{I,J})\overset{(f_1,g_{I,J\backslash\{j_1\}}\circ f_2,g_{I\cup\{j_1\},J\backslash\{j_1\}}\circ\delta_{j_1})}\simeq\bigoplus_{\substack{K\subset J \\ K\neq\emptyset}}\widetilde{H}_{n-|I|-|K|+1}(\bigcap_{i\in I\cup K}S_i)
    \end{equation}
    as claimed.\\
    All that remains to do is prove the existence of the last isomorphism in \eqref{eq:decomposition_thm}, which is easy: For any non-empty $K\subset\{1,\ldots,N\}$, the space $\bigcap_{i\in K}S_i$ is homeomorphic to a complex $(n-|K|+1)$-sphere according to Proposition \ref{prop:intersection_of_complex_spheres} and thus we have
    \begin{equation}
        \widetilde{H}_{n-|I|+1}(\bigcap_{i\in K}S_i)\simeq\begin{cases}\mathbb{Z} & \text{if }|K|\leq n+1 \\ 0 & \text{otherwise} \end{cases}.
    \end{equation}
    according to Corollary \ref{cor:homology_complex_sphere}. Thus, $\widetilde{H}_{n-|I|+1}(S^{I,J})$ contains a summand $\mathbb{Z}$ for every non-empty $K\subset J$ with $|I\cup K|\leq n+1$. Since there are $2^{|J|}-1$ such subsets of $J$ if $N\leq n+1$ and $2^{|J|}-2$ if $N=n+2$, we are done.
\end{proof}

\subsubsection{Comparison with the Local Homology}
The above Theorem \ref{thm:decomposition_homology} computes the homology of the pair $(\mathbb{C}^{n+1},S^{I,J})$ up to isomorphism. But a priori, this does not tell us much about the discontinuity of singular integrals since the vanishing cycle, sphere and cell are homology classes in
\begin{equation}
    H_{n+1-N}(U\cap\bigcap_{i\in I}S_i),\quad H_{n+1}(U-\bigcup_{i\in I}S_i) \quad\text{and}\quad H_{n+1}(U,\,\bigcup_{i\in I}S_i)
\end{equation}
respectively, where $I\subset\{1,\ldots,N\}$ is a non-empty subset and $U\subset\mathbb{C}^{n+1}$ is a small neighborhood around a pinch point (for details, the reader is referred to \cite{pham} for the general theory or \cite{cutkosky1} for this construction in the context of our setup of complex spheres). As is well-known, homology does not behave particularly nicely with respect to inclusions. In particular, if $i:U\hookrightarrow\mathbb{C}^{n+1}$ is the natural inclusion, then the induced map $i_\ast$ on the level of homology need neither be injective nor surjective. For arrangements of complex spheres however, there is a direct correspondence between the vanishing classes and the generators of the global homology:\\
Let $U\subset\mathbb{C}^{n+1}$ be an open neighborhood near pinch point where the $S_i$ with $i\in I$ pinch such that the vanishing classes are defined. Then the vanishing cycle is given by the deformation retract of $\bigcap_{i\in I}S_i$ to a real sphere contained in $U$. But the global homology $H_{n+1-N}(\bigcap_{i\in I}S_i)$ is also generated by this real sphere and thus the inclusion
\begin{equation}
    i:U\cap\bigcap_{i\in I}S_i\hookrightarrow\bigcap_{i\in I}S_i
\end{equation}
descends to an isomorphism on the level of homology. Furthermore, it is known that we have isomorphisms \cite{pham}
\begin{equation}
    \partial:H_{n+1}(U,\,\bigcup_{i\in I}S_i)\overset{\sim}\to H_{n+1-N}(U\cap\bigcap_{i\in I}S_i)
\end{equation}
and
\begin{equation}
    \delta:H_{n+1-N}(U\cap\bigcap_{i\in I}S_i)\overset{\sim}\to H_{n+1}(U-\bigcup_{i\in I}S_i),
\end{equation}
given by the iterated boundary $\partial:=\partial_N\circ\cdots\circ\partial_1$ (see equation \eqref{eq:iterated_boundary}) and the iterated Leray coboundary $\delta:=\delta_1\circ\cdots\circ\delta_N$ (see equation \eqref{eq:iterated_coboundary}) respectively.\footnote{This is not to be confused with the Mayer-Vietoris homomorphism, which we also denote by $\delta$.} Denoting $S:=\bigcup_{i\in I}S_i$ and $[N]:=\{1,\ldots,N\}$, we get a commutative diagram
\begin{center}
    \begin{tikzcd}[scale=0.3]
        H_{n+1}(U-S) \arrow[d, "\sim"] & & & \\
        H_{n+1-N}(U\cap\bigcap_{i\in I}S_i) \arrow[r, "\sim"] & H_{n+1-N}(\bigcap_{i\in I}S_i) \arrow[hook, r] & \bigoplus_{\substack{J\subset[N] \\ J\neq\emptyset}}H_{n+1-|J|}(\bigcap_{j\in J}S_j) \arrow{r}{\sim}[swap]{g_{\emptyset,I}^{-1}} & H_{n+1}(\mathbb{C}^{n+1},S) \\
        H_{n+1}(U,S) \arrow[u, "\sim"] & & &
    \end{tikzcd}
\end{center}
where all arrows are isomorphisms except for the inclusion into the direct sum, which is injective. In particular, the vanishing cell $[\mathbf{e}]\in H_{n+1}(U,S)$ (one of the two generators of $H_{n+1}(U,S)\simeq\mathbb{Z}$) can be expressed as a linear combination of the generators of the global homology group $H_{n+1}(\mathbb{C}^{n+1},S)$.

\section{Generators for the Homology}\label{sec:generators}
Theorem \ref{thm:decomposition_homology} tells us that the homology of $\mathbb{C}^{n+1}$ relative to an arrangement of complex spheres $S^{I,J}$ is trivial except in dimension $n-|I|+1$. The group $H_{n-|I|+1}(\mathbb{C}^{n+1},S^{I,J})$ is generated by $2^{|J|}-1$ elements $[\mathbf{e}_K^{I,J}]$ (except in the case $N=n+2$, where there are $2^{|J|}-2$ elements), one for each non-empty set $K\subset J$ with $|I\cup K|\leq n+1$. More specifically, following the isomorphism $g_{I,J}$, it is essentially the direct sum of terms
\begin{equation}
    (\partial_\ast)^{-1}\delta_{k_1}^{-1}\cdots\delta_{k_{|K|}}^{-1}\widetilde{H}_{n-|I|-|K|+1}(\bigcap_{i\in I\cup K}S_i)
\end{equation}
isomorphic to $\mathbb{Z}$ for each $K=\{k_1,\ldots,k_{|K|}\}\subset J$ with $k_1<\cdots<k_{|K|}$. Here, the $\delta_{k_i}$ are the Mayer-Vietoris homomorphisms given by the iterated decomposition of $S^{I,J}$ as described in the proof of Theorem \ref{thm:decomposition_homology}. To compute the intersection indices we are actually interested in however, it is convenient to have concrete representatives of these generators at hand. Since the above theorem already establishes the isomorphism $g_{I,J}$ which realizes the decomposition, it is sufficient to guess the generators of $H_{n-|I|+1}(\mathbb{C}^{n+1},S^{I,J})$, apply $g_{I,J}$ and check that they yield the obvious generators of the summands $\widetilde{H}_{n-|I|-|K|+1}(\bigcap_{i\in I\cup K}S_i)$.\\
A rigorous computation of the generators is quite finicky, technical and not necessarily very enlightening on its own. To better understand the intuition behind the ideas in this section, we start with two introductory examples. The first one is the smallest non-trivial case.
\begin{example}\label{ex:cw_1}
    Let us consider the two complex 1-spheres $S_1:=S_\mathbb{C}^1(a_1,r_1)$ and $S_2:=S_\mathbb{C}^1(a_2,r_2)$ in $\mathbb{C}^2$.\footnote{In the context of Feynman integrals, this corresponds to the famous one-loop bubble graph in two dimensions.} To simplify the picture in this example, we assume that $r_1,r_2\in\mathbb{R}_{>0}$ and set $a_1=0$ as well as $a_2=(i,0)$. This is the smallest non-trivial example and, unfortunately, it is already taking place in 4 real dimensions and consequently is hard to draw. Luckily, all the necessary intuition can already be obtained from a sketch in the 3 dimensional $(\text{Im}\,z_2=0)$-hyperplane. This situation is depicted in Figure \ref{fig:example_two_spheres_1}.\\
    \begin{figure}
        \centering
        \def\distance{0.6}
        \def\radiusa{0.5}
        \def\radiusb{0.4}
        \def\stepsa{9}
        \def\stepsaa{9}
        \def\stepsb{7}
        \def\stepsbb{13}
        \def\opa{0.4}
        \begin{tikzpicture}[scale=5]
            \draw[->] (0,0,0) -- (xyz cylindrical cs:radius=1.6);
            \draw (0,0,0) -- (xyz cylindrical cs:radius=1,angle=180);
            \draw[->] (0,0,0) -- (xyz cylindrical cs:radius=1,angle=90);
            \draw (0,0,0) -- (xyz cylindrical cs:radius=1,angle=-90);
            \draw[->] (0,0,0) -- (xyz cylindrical cs:z=-1);
            \draw (0,0,0) -- (xyz cylindrical cs:z=1);
        
            \draw (-0.11,0.03,-1) node {$\text{Re}\,k_1$};
            \draw (0.15,1,0) node {$\text{Re}\,k_2$};
            \draw (1.6,0.05,0) node {$\text{Im}\,k_1$};
            
            \draw (-0.35,0.15,0) node [text=red] {$S_1\cap\mathbb{R}^2$};
            \draw (1.1,0.25,0) node [text=blue] {$S_2\cap((\mathbb{R}+i)\times\mathbb{R})$};
            
            \begin{scope}[canvas is zy plane at x=0]
                \draw[line width=1pt] (0,0) circle (\radiusa);
            \end{scope}
            \begin{scope}[canvas is zy plane at x=-0.01]
                \draw[red,line width=2pt,fill=orange,opacity=\opa] (0,0) circle (\radiusa);
                \draw[red,line width=2pt] (0,0) circle (\radiusa);
            \end{scope}
            \begin{scope}[canvas is zy plane at x=0.01]
                \clip (-1,-1) rectangle (0,1);
                \draw[green, line width=2pt, fill=lime,opacity=\opa] (0,0) circle (\radiusa);
                \draw[green, line width=2pt] (0,0) circle (\radiusa);
            \end{scope}
            \begin{scope}[canvas is zy plane at x=\distance]
                \draw[line width=1pt] (0,0) circle (\radiusb);
            \end{scope}

            \begin{scope}[canvas is xy plane at z=0]
                \draw [thick, domain=0:1] plot (\x, {sqrt(\radiusa^2+\x^2)});
                \draw [thick, domain=0:1] plot (-\x, {sqrt(\radiusa^2+\x^2)});
                \draw [thick, domain=0:1] plot (\x, -{sqrt(\radiusa^2+\x^2)});
                \draw [thick, domain=0:1] plot (-\x, -{sqrt(\radiusa^2+\x^2)});
                \draw [thick, domain=0:1] plot (\x+\distance, {sqrt(\radiusb^2+\x^2)});
                \draw [thick, domain=0:1] plot (-\x+\distance, {sqrt(\radiusb^2+\x^2)});
                \draw [thick, domain=0:1] plot (\x+\distance, -{sqrt(\radiusb^2+\x^2)});
                \draw [thick, domain=0:1] plot (-\x+\distance, -{sqrt(\radiusb^2+\x^2)});
            \end{scope}
            \begin{scope}[canvas is xy plane at z=0.01]
                \draw [green, line width=2pt, domain=0.01:(\radiusb^2-\radiusa^2+\distance^2)/(2*\distance),name path=A] plot (\x, {sqrt(\radiusa^2+\x^2)});
                \draw [green, line width=2pt, domain=0.01:(\radiusb^2-\radiusa^2+\distance^2)/(2*\distance),name path=B] plot (\x, -{sqrt(\radiusa^2+\x^2)});
                \tikzfillbetween[of=A and B] {lime,opacity=\opa};
                
                \draw [green, line width=2pt, domain=-0.02:(\distance-(\radiusb^2-\radiusa^2+\distance^2)/(2*\distance)),name path=C] plot (-\x+\distance, {sqrt(\radiusb^2+\x^2)});
                \draw [green, line width=2pt, domain=-0.02:(\distance-(\radiusb^2-\radiusa^2+\distance^2)/(2*\distance)),name path=D] plot (-\x+\distance, -{sqrt(\radiusb^2+\x^2)});
                \tikzfillbetween[of=C and D] {lime,opacity=\opa};
            \end{scope}
            \begin{scope}[canvas is zy plane at x=\distance-0.01]
                \draw[blue,line width=2pt,fill=cyan,opacity=\opa] (0,0) circle (\radiusb);
                \draw[blue,line width=2pt] (0,0) circle (\radiusb);
            \end{scope}
            \begin{scope}[canvas is zy plane at x=\distance+0.01]
                \clip (-1,-1) rectangle (0,1);
                \draw[green,line width=2pt, fill=lime,opacity=\opa] (0,0) circle (\radiusb);
                \draw[green,line width=2pt] (0,0) circle (\radiusb);
            \end{scope}
            \draw ({(\radiusb^2-\radiusa^2+\distance^2)/(2*\distance)},{sqrt(((\radiusb^2-\radiusa^2+\distance^2)/(2*\distance))^2+\radiusa^2)},0) node [circle, fill, scale=0.7] {};
            \draw ({(\radiusb^2-\radiusa^2+\distance^2)/(2*\distance)+0.01},{sqrt(((\radiusb^2-\radiusa^2+\distance^2)/(2*\distance))^2+\radiusa^2)+0.1},0) node {$p_+$};
            \draw ({(\radiusb^2-\radiusa^2+\distance^2)/(2*\distance)+0.01},{-sqrt(((\radiusb^2-\radiusa^2+\distance^2)/(2*\distance))^2+\radiusa^2)-0.1},0) node {$p_-$};
            \draw ({(\radiusb^2-\radiusa^2+\distance^2)/(2*\distance)},{-sqrt(((\radiusb^2-\radiusa^2+\distance^2)/(2*\distance))^2+\radiusa^2)},0) node [circle, fill, scale=0.7] {};
            
            \draw (0.01,\radiusa,0) node [circle, fill, green, scale=0.4] {};
            \draw (0.01,-\radiusa,0) node [circle, fill, green, scale=0.4] {};
            \draw (\distance+0.01,\radiusb,0) node [circle, fill, green, scale=0.4] {};
            \draw (\distance+0.01,-\radiusb,0) node [circle, fill, green, scale=0.4] {};
        \end{tikzpicture}
        \caption{The union $S_1\cup S_2$ in the $\text{Im}\,z_2=0$ plane. The generators of $H_1(S_1\cup S_2)$ are marked in red, blue and green and the generators of $H_2(\mathbb{C}^2,S_1\cup S_2)$ in orange, cyan and lime.}
        \label{fig:example_two_spheres_1}
    \end{figure}
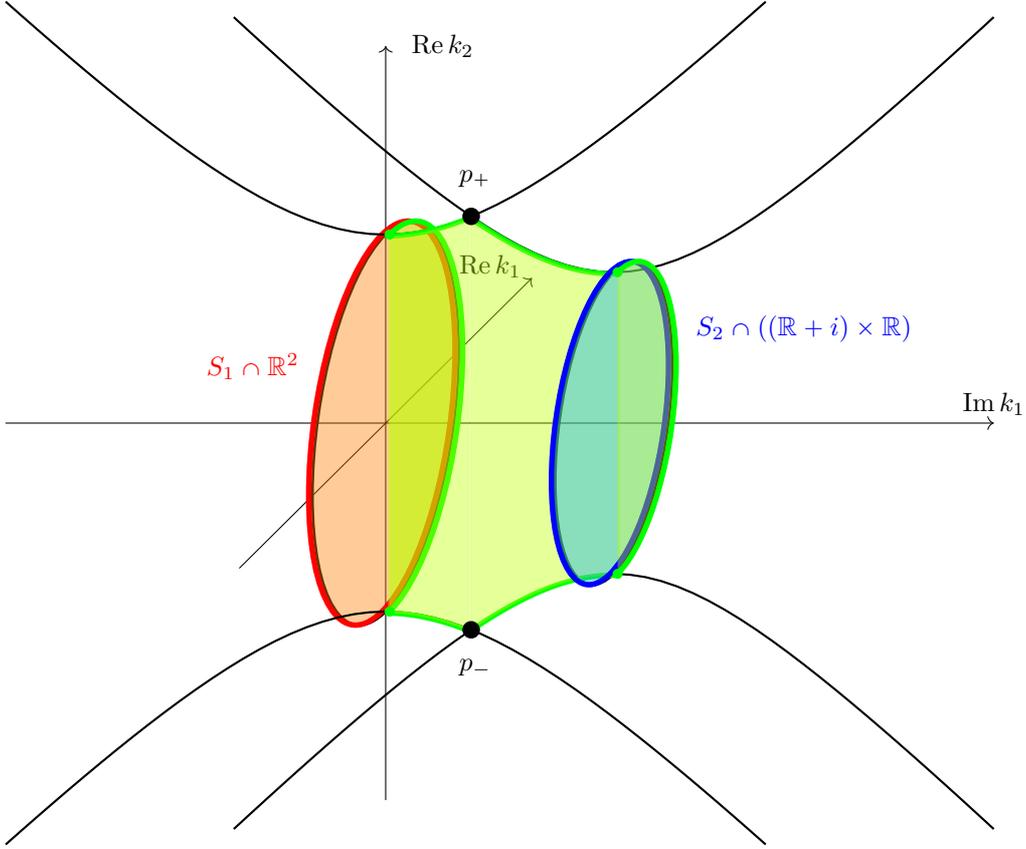
    From the general statement in Theorem \ref{thm:decomposition_homology}, we know that there is a decomposition
    \begin{equation}
        H_2(\mathbb{C}^2,S_1\cup S_2)\simeq\widetilde{H}_1(S_1\cup S_2)\simeq \widetilde{H}_1(S_1)\oplus\widetilde{H}_1(S_2)\oplus\widetilde{H}_0(S_1\cap S_2)
    \end{equation}
    and hence that the homology group $H_2(\mathbb{C}^2,S_1\cup S_2)$ (resp. $\widetilde{H}_1(S_1\cup S_2)$) is generated by three elements $[\mathbf{e}_1]$, $[\mathbf{e}_2]$, $[\mathbf{e}_{12}]$ (resp. $[e_1]$, $[e_2]$, $[e_{12}]$) corresponding to $S_1$, $S_2$ and $S_1\cap S_2$ respectively. In Figure \ref{fig:example_two_spheres_1}, we can easily identify the real spheres to which $S_1$ and $S_2$ deformation retract and hence (representatives of) the generators $e_1$ and $e_2$. These are drawn in red and blue respectively. Concretely, their support can be written as
    \begin{equation}
        |e_1|=S_1\cap\mathbb{R}^2 \qquad\text{and}\qquad |e_2|=S_2\cap((\mathbb{R}+i)\times\mathbb{R}).
    \end{equation}
    Furthermore, we can identify the generators $\mathbf{e}_1$ and $\mathbf{e}_2$ of $H_2(\mathbb{C}^2,S_1\cup S_2)$, drawn in orange and cyan in Figure \ref{fig:example_two_spheres_1}. They are the disks associated to the real spheres $e_1$, $e_2$ and satisfy $\partial\mathbf{e}_i=e_i$ for $i=1,2$. The generator corresponding to the intersection $S_1\cap S_2$ is a little bit more complicated: The subspace $S_1\cap S_2$ itself is a complex 0-sphere, which consists of two points $p_+,p_-$. In Figure \ref{fig:example_two_spheres_1}, these two points are marked as black dots. The reduced homology group $\widetilde{H}_0(S_1\cap S_2)$ is generated by $e_{12}'=p_+-p_-$ (viewed as an appropriate homology class). To be more specific, denote $b:=\frac{1}{2}(r_2^2-r_1^2+1)$. Then we can write
    \begin{equation}
        p_\pm=(b,\pm\sqrt{r_1^2+b^2}).
    \end{equation}
    Theorem \ref{thm:decomposition_homology} tells us that we can obtain the corresponding generator $e_{12}$ by applying the inverse of the Mayer-Vietoris homomorphism corresponding to the decomposition of $S_1\cup S_2$ into $S_1$ and $S_2$. If we consider the green part in Figure \ref{fig:example_two_spheres_1} left of $S_1\cap S_2$ as a chain $u\in C_1(S_1)$ and the green part right of $S_1\cap S_2$ as a chain $v\in C_1(S_2)$, we obtain a cycle $e_{12}=u+v$, which indeed satisfies $\partial u=-\partial v=e_{12}'$ (if we choose the orientation correctly). We can think of $e_{12}$ as comprising four parts, two lying in $S_1$ corresponding to $u$ and two in $S_2$ corresponding to $v$: The support of $e_{12}$ can be written as
    \begin{equation}
        |e_{12}|=(S_1\cap((\mathbb{R}_{\geq0}\times\mathbb{R})\cup(i[0,b]\times\mathbb{R})))\cup(S_2\cap(((\mathbb{R}_{\geq0}+i)\times\mathbb{R})\cup(i[b,1]\times\mathbb{R}))).
    \end{equation}
    The corresponding generator $\mathbf{e}_{12}$ for the relative homology is shown in lime in Figure \ref{fig:example_two_spheres_1}. All of the sets involved in the unions in the support of $e_1$, $e_2$ and $e_{12}$ are homeomorphic to a line segment. Thus it is easy to see that we can give $|e_1|\cup|e_2|\cup|e_{12}|$ the structure of a CW-complex which readily yields the homology we computed pictorially.\\
    For $\mathbf{e}_1$ and $\mathbf{e}_2$, the computation of $\langle(i\cdot\mathbb{R})^2|\mathbf{e}_i\rangle$ is not difficult. They intersect $(i\cdot\mathbb{R})^2$ only in 0 and $a_2$ respectively and it suffices to compare the orientation at these points. For $\mathbf{e}_{12}$, we can slightly deform the cycle away from $(i\cdot\mathbb{R})^2$ as shown in Figure \ref{fig:example_two_spheres_2}. Thus we have $\langle(i\cdot\mathbb{R})^2|\mathbf{e}_{12}\rangle=0$.
    \begin{figure}
        \centering
        \def\distance{0.6}
        \def\radiusa{0.5}
        \def\radiusb{0.4}
        \def\stepsa{9}
        \def\stepsaa{9}
        \def\stepsb{7}
        \def\stepsbb{13}
        \def\opa{0.4}
        \begin{tikzpicture}[scale=5]
            \draw[->] (0,0,0) -- (xyz cylindrical cs:radius=1.6);
            \draw (0,0,0) -- (xyz cylindrical cs:radius=1,angle=180);
            \draw[->] (0,0,0) -- (xyz cylindrical cs:radius=1,angle=90);
            \draw (0,0,0) -- (xyz cylindrical cs:radius=1,angle=-90);
            \draw[->] (0,0,0) -- (xyz cylindrical cs:z=-1);
            \draw (0,0,0) -- (xyz cylindrical cs:z=1);
        
            \draw (-0.11,0.03,-1) node {$\text{Re}\,z_1$};
            \draw (0.15,1,0) node {$\text{Re}\,z_2$};
            \draw (1.6,0.05,0) node {$\text{Im}\,z_1$};
            
            \begin{scope}[canvas is zy plane at x=0]
                \draw[line width=1pt] (0,0) circle (\radiusa);
            \end{scope}
            \begin{scope}[canvas is zy plane at x=0.01]
                \clip (-1,-1) rectangle (0,1);
                \draw[green, line width=3pt, name path = U] (0,0) circle (\radiusa);
                \draw[green, line width=3pt, name path = V] (0,0) circle (0.1);
                \tikzfillbetween[of=U and V] {lime};
            \end{scope}
            \begin{scope}[canvas is zy plane at x=0.05]
                \clip (-1,-1) rectangle (0,1);
                \draw[green] (0,0) circle (0.1);
            \end{scope}
            \begin{scope}[canvas is zy plane at x=0.10]
                \clip (-1,-1) rectangle (0,1);
                \draw[green] (0,0) circle (0.1);
            \end{scope}
            \begin{scope}[canvas is zy plane at x=0.15]
                \clip (-1,-1) rectangle (0,1);
                \draw[green] (0,0) circle (0.1);
            \end{scope}
            \begin{scope}[canvas is zy plane at x=0.20]
                \clip (-1,-1) rectangle (0,1);
                \draw[green] (0,0) circle (0.1);
            \end{scope}
            \begin{scope}[canvas is zy plane at x=0.25]
                \clip (-1,-1) rectangle (0,1);
                \draw[green] (0,0) circle (0.1);
            \end{scope}
            \begin{scope}[canvas is zy plane at x=0.30]
                \clip (-1,-1) rectangle (0,1);
                \draw[green] (0,0) circle (0.1);
            \end{scope}
            \begin{scope}[canvas is zy plane at x=0.35]
                \clip (-1,-1) rectangle (0,1);
                \draw[green] (0,0) circle (0.1);
            \end{scope}
            \begin{scope}[canvas is zy plane at x=0.40]
                \clip (-1,-1) rectangle (0,1);
                \draw[green] (0,0) circle (0.1);
            \end{scope}
            \begin{scope}[canvas is zy plane at x=0.45]
                \clip (-1,-1) rectangle (0,1);
                \draw[green] (0,0) circle (0.1);
            \end{scope}
            \begin{scope}[canvas is zy plane at x=0.50]
                \clip (-1,-1) rectangle (0,1);
                \draw[green] (0,0) circle (0.1);
            \end{scope}
            \begin{scope}[canvas is zy plane at x=0.55]
                \clip (-1,-1) rectangle (0,1);
                \draw[green] (0,0) circle (0.1);
            \end{scope}
            
            \begin{scope}[canvas is zy plane at x=\distance]
                \draw[line width=1pt] (0,0) circle (\radiusb);
            \end{scope}
            \begin{scope}[canvas is xy plane at z=0]
                \draw [thick, domain=0:1] plot (\x, {sqrt(\radiusa^2+\x^2)});
                \draw [thick, domain=0:1] plot (-\x, {sqrt(\radiusa^2+\x^2)});
                \draw [thick, domain=0:1] plot (\x, -{sqrt(\radiusa^2+\x^2)});
                \draw [thick, domain=0:1] plot (-\x, -{sqrt(\radiusa^2+\x^2)});
                \draw [thick, domain=0:1] plot (\x+\distance, {sqrt(\radiusb^2+\x^2)});
                \draw [thick, domain=0:1] plot (-\x+\distance, {sqrt(\radiusb^2+\x^2)});
                \draw [thick, domain=0:1] plot (\x+\distance, -{sqrt(\radiusb^2+\x^2)});
                \draw [thick, domain=0:1] plot (-\x+\distance, -{sqrt(\radiusb^2+\x^2)});
            \end{scope}
            \begin{scope}[canvas is xy plane at z=0.01]
                \draw [green, line width=2pt, domain=0.01:(\radiusb^2-\radiusa^2+\distance^2)/(2*\distance),name path=A] plot (\x, {sqrt(\radiusa^2+\x^2)});
                \draw [green, line width=2pt, domain=0.01:(\radiusb^2-\radiusa^2+\distance^2)/(2*\distance),name path=B] plot (\x, 0.1);
                \draw [green, line width=2pt, domain=0.1:(\radiusb^2-\radiusa^2+\distance^2)/(2*\distance),name path=W, draw=none] plot (\x, 0);
                \tikzfillbetween[of=A and W] {lime,opacity=\opa};
                
                \draw [green, line width=2pt, domain=0.01:(\radiusb^2-\radiusa^2+\distance^2)/(2*\distance),name path=C] plot (\x, -{sqrt(\radiusa^2+\x^2)});
                \draw [green, line width=2pt, domain=0.01:(\radiusb^2-\radiusa^2+\distance^2)/(2*\distance),name path=D] plot (\x, -0.1);
                \draw [green, line width=2pt, domain=0.1:(\radiusb^2-\radiusa^2+\distance^2)/(2*\distance),name path=Z, draw=none] plot (\x, 0);
                \tikzfillbetween[of=C and Z] {lime,opacity=\opa};
                
                \draw [green, line width=2pt, domain=-0.02:(\distance-(\radiusb^2-\radiusa^2+\distance^2)/(2*\distance)),name path=E] plot (-\x+\distance, {sqrt(\radiusb^2+\x^2)});
                \draw [green, line width=2pt, domain=-0.02:(\distance-(\radiusb^2-\radiusa^2+\distance^2)/(2*\distance)),name path=F] plot (-\x+\distance, 0.1);
                \draw [green, line width=0pt, domain=-0.02:(\distance-(\radiusb^2-\radiusa^2+\distance^2)/(2*\distance)),name path=X, draw=none] plot (-\x+\distance, 0);
                \tikzfillbetween[of=E and X] {lime,opacity=\opa};
                
                \draw [green, line width=2pt, domain=-0.02:(\distance-(\radiusb^2-\radiusa^2+\distance^2)/(2*\distance)),name path=G] plot (-\x+\distance, -{sqrt(\radiusb^2+\x^2)});
                \draw [green, line width=2pt, domain=-0.02:(\distance-(\radiusb^2-\radiusa^2+\distance^2)/(2*\distance)),name path=H] plot (-\x+\distance, -0.1);
                \draw [green, line width=0pt, domain=-0.02:(\distance-(\radiusb^2-\radiusa^2+\distance^2)/(2*\distance)),name path=Y, draw=none] plot (-\x+\distance, 0);
                \tikzfillbetween[of=G and Y] {lime,opacity=\opa};
            \end{scope}
            \begin{scope}[canvas is zy plane at x=\distance+0.016]
                \clip (-1,-1) rectangle (0,1);
                \draw[green,line width=2pt, name path=X] (0,0) circle (\radiusb);
                \draw[green,line width=2pt, name path=Y] (0,0) circle (0.1);
                \draw[green,line width=0pt, name path=Z, draw=none] (0,-0.1) -- (0,0.1);
                \tikzfillbetween[of=X and Z] {lime,opacity=\opa};
            \end{scope}
            \draw ({(\radiusb^2-\radiusa^2+\distance^2)/(2*\distance)},{sqrt(((\radiusb^2-\radiusa^2+\distance^2)/(2*\distance))^2+\radiusa^2)},0) node [circle, fill, scale=0.7] {};
            \draw ({(\radiusb^2-\radiusa^2+\distance^2)/(2*\distance)+0.01},{sqrt(((\radiusb^2-\radiusa^2+\distance^2)/(2*\distance))^2+\radiusa^2)+0.1},0) node {$S_1\cap S_2$};
            \draw ({(\radiusb^2-\radiusa^2+\distance^2)/(2*\distance)},{-sqrt(((\radiusb^2-\radiusa^2+\distance^2)/(2*\distance))^2+\radiusa^2)},0) node [circle, fill, scale=0.7] {};
            
            \draw (0.01,\radiusa,0) node [circle, fill, green, scale=0.4] {};
            \draw (0.01,-\radiusa,0) node [circle, fill, green, scale=0.4] {};
            \draw (\distance+0.01,\radiusb,0) node [circle, fill, green, scale=0.4] {};
            \draw (\distance+0.01,-\radiusb,0) node [circle, fill, green, scale=0.4] {};
            \draw (0.01,0.1,0) node [circle, fill, green, scale=0.4] {};
            \draw (0.01,-0.1,0) node [circle, fill, green, scale=0.4] {};
            \draw (\distance+0.01,0.1,0) node [circle, fill, green, scale=0.4] {};
            \draw (\distance+0.01,-0.1,0) node [circle, fill, green, scale=0.4] {};
        \end{tikzpicture}
        \caption{The situation as in Figure \ref{fig:example_two_spheres_1} after a slight deformation to move $\mathbf{e}_{12}$ off of $(i\mathbb{R})^2$.}
        \label{fig:example_two_spheres_2}
    \end{figure}
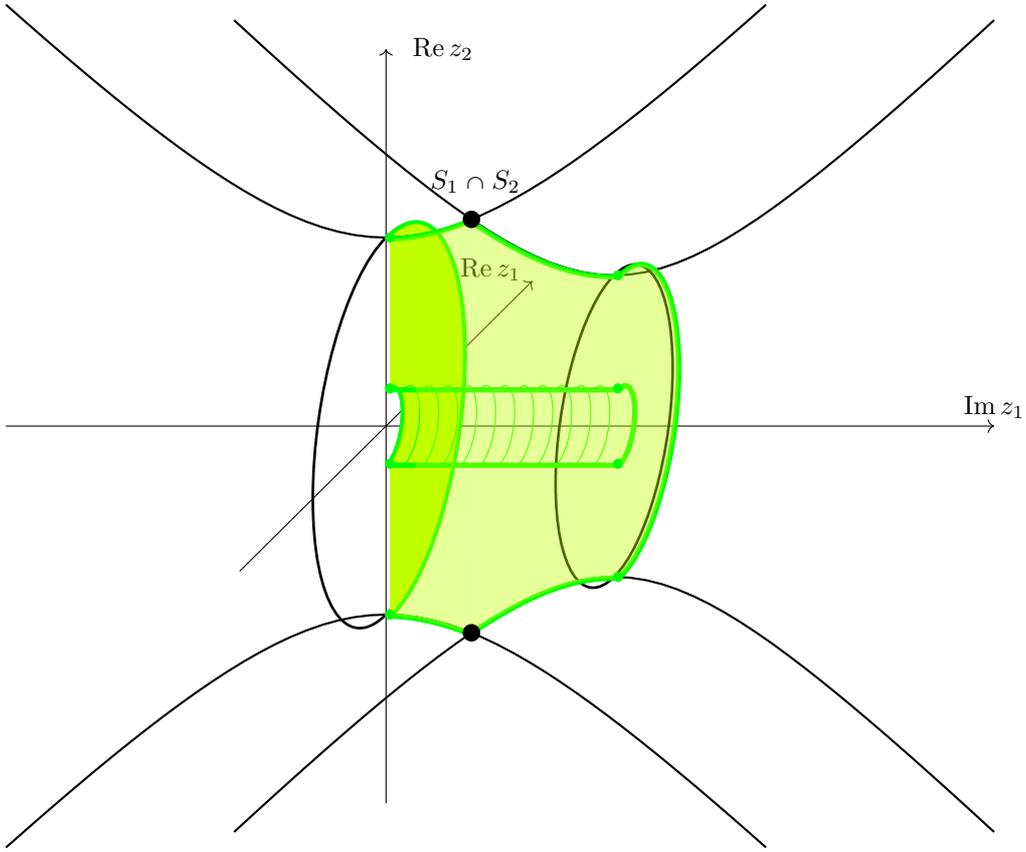
\end{example}
\begin{example}\label{ex:cw_2}
    Let us now consider the case $n=2$. I.e. we consider three complex 2-spheres $S_j:=S_\mathbb{C}^2(a_j,r_j)$ ($j=1,2,3$). To obtain a symmetric situation, we set $r_1=r_2=r_3=1$ and
    \begin{equation}
        a_1=(i,0,0),\qquad a_2=(0,i,0),\qquad a_3=(0,0,i).
    \end{equation}
    With these values, the complex spheres $S_1$, $S_2$ and $S_3$ are in general position. Again, by the general Theorem \ref{thm:decomposition_homology}, we know that the relevant homology groups decompose as
    \begin{equation}\label{eq:ex_2_decomposition}
        \begin{split}
            &H_3(\mathbb{C}^3,\bigcup_{j=1}^3S_j)\simeq\widetilde{H}_2(\bigcup_{j=1}^3S_j)\\
            \simeq\;&\widetilde{H}_2(S_1)\oplus\widetilde{H}_2(S_2)\oplus\widetilde{H}_2(S_3)\oplus\widetilde{H}_1(S_1\cap S_2)\oplus\widetilde{H}_1(S_1\cap S_3)\oplus\widetilde{H}_1(S_2\cap S_3)\oplus\widetilde{H}_0(S_1\cap S_2\cap S_3).
        \end{split}
    \end{equation}
    The individual terms are again easy to understand: For each non-empty $I\subset\{1,2,3\}$, the group $\widetilde{H}_{3-|I|}(\bigcap_{j\in I}S_j)$ is generated by the real $(3-|I|)$-sphere to which $\bigcap_{j\in I}S_j$ deformation retracts. To find the corresponding generators of $\widetilde{H}_2(\bigcup_{j=1}^3S_j)$, we need to trace back the Mayer-Vietoris homomorphisms which lead to the direct sum decomposition.\\
    Let us work our way up from the homology groups in the decomposition \eqref{eq:ex_2_decomposition}. It is not difficult to see that the $0$-dimensional sphere $S_1\cap S_2\cap S_3$ is given by
    \begin{equation}
        (\begin{pmatrix} \frac{i}{3} \\ \frac{i}{3} \\ \frac{i}{3} \end{pmatrix} +\mathbb{R}\cdot\begin{pmatrix} 1 \\ 1 \\ 1 \end{pmatrix})\cap S_j
    \end{equation}
    for any choice of $j\in\{1,2,3\}$. We clearly want this to be included in our complex. For the real spheres associated to $S_1\cap S_2$, we obtain
    \begin{equation}
        (\begin{pmatrix} \frac{i}{2} \\ \frac{i}{2} \\ 0 \end{pmatrix} +\mathbb{R}\cdot\begin{pmatrix} 1 \\ 1 \\ 1 \end{pmatrix}+\mathbb{R}\cdot\begin{pmatrix} 1 \\ 1 \\ 0 \end{pmatrix})\cap S_j
    \end{equation}
    for any $j\in\{1,2\}$. Similar expressions for $S_1\cap S_3$ and $S_2\cap S_3$ can be derived by simply permuting the coordinates. These, we want to include as well. The choice of the basis vectors for the real parts are chosen such that we can glue the cells together nicely. In the same fashion, we can write the real sphere in $S_1$ as
    \begin{equation}
        (\begin{pmatrix} i \\ 0 \\ 0 \end{pmatrix} +\mathbb{R}\cdot\begin{pmatrix} 1 \\ 1 \\ 1 \end{pmatrix}+\mathbb{R}\cdot\begin{pmatrix} 1 \\ 1 \\ 0 \end{pmatrix}+\mathbb{R}\cdot\begin{pmatrix} 1 \\ 0 \\ 1 \end{pmatrix})\cap S_1
    \end{equation}
    and similarly for $S_2$ and $S_3$. To simplify our notation, let us denote the center of the real sphere associated to $\bigcap_{i\in I}S_i$, with $I\subset\{1,2,3\}$ a non-empty subset, by $c_I\in(i\cdot\mathbb{R})^3$. More precisely,
    \begin{equation}
        (c_I)_j:=\begin{cases} \frac{i}{|I|} & \text{if }j\in I \\ 0 & \text{if }j\notin I \end{cases},\qquad j=1,2,3.
    \end{equation}
    Moreover, let us denote
    \begin{equation}
        v:=\begin{pmatrix} 1 \\ 1 \\ 1\end{pmatrix}\qquad\text{and}\qquad \forall j\in\{1,2,3\}\;:\;v_j:=v-e_j.
    \end{equation}
    Then we express the real sphere associated to $\bigcap_{i\in I}S_i$ as
    \begin{equation}
        (c_I+\mathbb{R}\cdot v+\sum_{i\notin I}\mathbb{R}\cdot v_i)\cap S_j
    \end{equation}
    for any choice of $j\in I$.\\
    To connect all these pieces, we first set up some simplices with support contained within $(i\cdot\mathbb{R})^3$. The centers $c_I$ need to be included as 0-simplices. We construct this complex by considering the 2-simplex spanned by $c_{\{1\}}$, $c_{\{2\}}$ and $c_{\{3\}}$. The remaining 4 centers $c_{\{1,2\}}$, $c_{\{1,3\}}$, $c_{\{2,3\}}$ and $c_{\{1,2,3\}}$ lie in this simplex. In fact, performing a barycentric subdivision of this simplex yields a simplicial complex whose 0-simplices are precisely the $c_I$. This complex is depicted in Figure \ref{fig:imaginary simplices_1}. The simplices can be labeled by chains of subsets $\emptyset\neq I_1\subsetneq\cdots\subsetneq I_k\subset\{1,2,3\}$.
    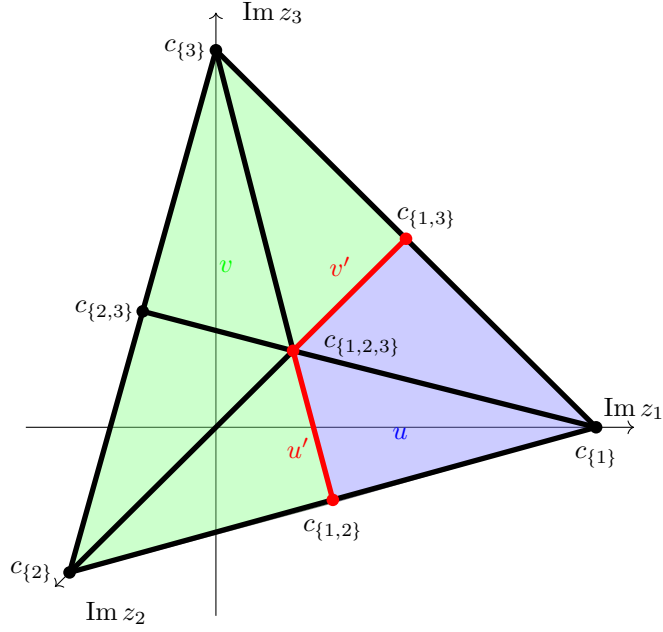
\begin{figure}
        \centering
        \def\distance{0.6}
        \def\radiusa{0.5}
        \def\radiusb{0.4}
        \def\stepsa{9}
        \def\stepsb{7}
        \begin{tikzpicture}[scale=5]
            \draw[->] (0,0,0) -- (xyz cylindrical cs:radius=1.1);
            \draw (0,0,0) -- (xyz cylindrical cs:radius=0.5,angle=180);
            \draw[->] (0,0,0) -- (xyz cylindrical cs:radius=1.1,angle=90);
            \draw (0,0,0) -- (xyz cylindrical cs:radius=0.5,angle=-90);
            \draw (0,0,0) -- (xyz cylindrical cs:z=-1);
            \draw[->] (0,0,0) -- (xyz cylindrical cs:z=1.1);
            
            \begin{scope}[canvas is xy plane at z=0]
                \coordinate (s0) at (1,0);
                \coordinate (r0) at (0,1);
                \coordinate (r1) at (0.5,0.5);
                \coordinate (r2) at (0.2,0.2);
                \coordinate (r3) at (0.31,-0.2);
                \coordinate (r4) at (-0.39,-0.39);
                \filldraw[draw=black, fill=green!50, opacity=0.4] (r0) -- (r1) -- (r2) -- (r3) -- (r4) -- cycle;
                \filldraw[draw=black, fill=blue!50, opacity=0.4] (s0) -- (r1) -- (r2) -- (r3) -- cycle;
            \end{scope}
            
            \draw (1.1,0.05,0) node {$\text{Im}\,z_1$};
            \draw (0.2,-0.03,1.2) node {$\text{Im}\,z_2$};
            \draw (0.15,1.1,0) node {$\text{Im}\,z_3$};

            \draw (1,-0.08,0) node {$c_{\{1\}}$};
            \draw (-0.08,1,0) node {$c_{\{3\}}$};
            \draw (-0.1,0,1) node {$c_{\{2\}}$};
            \draw (0.5,0.5,-0.14) node {$c_{\{1,3\}}$};
            \draw (0.5,-0.08,0.5) node {$c_{\{1,2\}}$};
            \draw (-0.1,0.5,0.5) node {$c_{\{2,3\}}$};
            \draw (0.5,0.33,0.30) node {$c_{\{1,2,3\}}$};

            \draw (0.37,0.1,0.4) node [red] {$u'$};
            \draw (0.5,0.6,0.45) node [red] {$v'$};
            \draw (0.7,0.2,0.56) node [blue] {$u$};
            \draw (0.2,0.6,0.45) node [green] {$v$};
            
            \draw[line width=2pt] (0.5,0.5,0) -- (1,0,0);
            \draw[line width=2pt] (0.5,0.5,0) -- (0,1,0);
            \draw[line width=2pt] (0.5,0,0.5) -- (1,0,0);
            \draw[line width=2pt] (0.5,0,0.5) -- (0,0,1);
            \draw[line width=2pt] (0,0.5,0.5) -- (0,1,0);
            \draw[line width=2pt] (0,0.5,0.5) -- (0,0,1);

            \draw[line width=2pt,red] (0.33,0.33,0.33) -- (0.5,0.5,0);
            \draw[line width=2pt,red] (0.33,0.33,0.33) -- (0.5,0,0.5);
            \draw[line width=2pt] (0.33,0.33,0.33) -- (0,0.5,0.5);
            \draw[line width=2pt] (0.33,0.33,0.33) -- (1,0,0);
            \draw[line width=2pt] (0.33,0.33,0.33) -- (0,1,0);
            \draw[line width=2pt] (0.33,0.33,0.33) -- (0,0,1);
            
            \draw (1,0,0) node [circle, fill, scale=0.5] {};

            \draw (0,1,0) node [circle, fill, scale=0.5] {};

            \draw (0,0,1) node [circle, fill, scale=0.5] {};

            \draw (0.5,0.5,0) node [circle, fill, scale=0.5,red] {};
            \draw (0.5,0,0.5) node [circle, fill, scale=0.5,red] {};
            \draw (0,0.5,0.5) node [circle, fill, scale=0.5] {};

            \draw (0.33,0.33,0.33) node [circle, fill, scale=0.5,red] {};
        \end{tikzpicture}
        \caption{The subcomplex consisting of all simplices with support contained in $(i\mathbb{R})^{n+1}$ for the case $n=2$. The 0-cells are the centers of the complex spheres given by the various intersections of the $S_i$.}
        \label{fig:imaginary simplices_1}
    \end{figure}
    Clearly, this complex alone does not contain any non-trivial cycles and does not even intersect $\bigcup_{i=1}^3S_i$. But we can glue more cells to the complex. To each simplex $(I_1,\ldots,I_k)$, we associate the sets
    \begin{equation}
        V_{J_\leq,J_\geq}:=\mathbb{R}\cdot v+\sum_{j\in J_\leq}\mathbb{R}_{\leq0}\cdot v_j+\sum_{j\in J_\geq}\mathbb{R}_{\geq0}\cdot v_j,\qquad J_\leq\cap I=J_\geq\cap I=\emptyset
    \end{equation}
    and form the cells
    \begin{equation}
        e_{(I_1,\ldots,I_k),J_\leq,J_\geq}:=((I_1,\ldots,I_k)+V_{J_\leq,J_\geq})\cap S_{\min I_1}.
´    \end{equation}
    Of course, we need to show that these sets actually form a CW-complex. We do not do this here but further below we establish this for the general case. For now, we appeal to the geometric intuition of the reader. The complex is sufficiently large to work our way backwards along the up to two Mayer-Vietoris homomorphisms. For any intersection of two complex spheres, the situation looks essentially the same as in Example \ref{ex:cw_1}. Let us work out the maximal intersection $S_1\cap S_2\cap S_3$. Consider the two 2-simplices
    \begin{equation}
        u':=e_{(\{1,2\}),\emptyset,\{3\}}\pm e_{(\{1,2\},\{1,2,3\}),\emptyset,\emptyset} \quad\text{and}\quad v':=e_{(\{1,3\}),\emptyset,\{2\}}\pm e_{(\{1,3\},\{1,2,3\}),\emptyset,\emptyset}
    \end{equation}
    In Figure \ref{fig:imaginary simplices_1}, the projection of these simplices to $(i\cdot\mathbb{R})^3$ is marked in red. Of course, these linear combinations are not to be taken seriously here since we have not discussed the issue of orientation. Orienting everything correctly however (which we do for the general case in the following subsection), $\partial u'=-\partial v'$ is precisely (the real sphere in) $S_1\cap S_2\cap S_3$. Thus $[u'+v']$ generates $\delta_2^{-1}\widetilde{H}_0(S_1\cap S_2\cap S_3)$. In a similar fashion, we can define
    \begin{equation}
        \begin{split}
            u:&=e_{(\{1\}),\emptyset,\{2,3\}}\pm e_{(\{1\},\{1,2\}),\emptyset,\{3\}}\pm e_{(\{1\},\{1,3\}),\emptyset,\{2\}}\\
            &\pm e_{(\{1\},\{1,3\},\{1,2,3\}),\emptyset,\emptyset}\pm e_{(\{1\},\{1,2\},\{1,2,3\}),\emptyset,\emptyset}
        \end{split}
    \end{equation}
    and
    \begin{equation}
        \begin{split}
            v:&=e_{(\{2\}),\emptyset,\{1,3\}}\pm e_{(\{2\},\{1,2\}),\emptyset,\{3\}}\pm e_{(\{2\},\{2,3\}),\emptyset,\{1\}}\\
            &\pm e_{(\{2\},\{1,2\},\{1,2,3\}),\emptyset,\emptyset}\pm e_{(\{2\},\{2,3\},\{1,2,3\}),\emptyset,\emptyset}\\
            &\pm e_{(\{3\}),\emptyset,\{1,2\}}\pm e_{(\{3\},\{1,3\}),\emptyset,\{2\}}\pm e_{(\{3\},\{2,3\}),\emptyset,\{1\}}\\
            &\pm e_{(\{3\},\{1,3\},\{1,2,3\}),\emptyset,\emptyset}\pm e_{(\{3\},\{2,3\},\{1,2,3\}),\emptyset,\emptyset},
        \end{split}
    \end{equation}
    which are already quite unwieldy expressions. Nevertheless, it can be verified, again after carefully choosing orientations, that $\partial u=-\partial v$ is exactly $u'+v'$ so that $[u+v]$ generates $\delta_1^{-1}\delta_2^{-1}\widetilde{H}_0(S_1\cap S_2\cap S_3)$. The projection of $u$ and $v$ to $(i\cdot\mathbb{R})^3$ is marked in blue and green respectively in Figure \ref{fig:imaginary simplices_1}. The generators for the relative homology $H_3(\mathbb{C}^3,S_1\cup S_2\cup S_3)$ can be obtained by basically filling up the space between the generators of $\widetilde{H}_2(S_1\cup S_2\cup S_3)$, similar to the previous Example \ref{ex:cw_1}.
\end{example}
Now we turn to the general construction:
\subsection{Construction of the Complex}
The ideas from the above example can be generalized to an arbitrary number $n+1$ of dimensions and an arbitrary number $N$ of complex spheres as long as $N\leq n+1$. Similarly to Example \ref{ex:cw_2}, it is useful to consider a particularly simple arrangement of complex spheres, symmetric in the indices labeling the spheres (more precisely: a permutation of the indices corresponds to a permutation of the coordinates of $\mathbb{C}^{n+1}$). Without loss of generality, we may consider the case $N=n+1$, which is the maximal number of complex spheres such that the maximal intersection is non-empty. For smaller $N$, we can simply forget about the generators corresponding to intersections of complex spheres which do not appear in the arrangement of $N$ spheres. For larger $N$, we can repeat the construction below for any subset of complex sphere of size $n+1$.\footnote{This is not much of a generalization though since for an arrangement $S^{I,J}$ of $N$ complex $n$-spheres to be in general position, we need $N\leq n+2$.}\\
We set $r_1=\cdots=r_{n+1}:=1$ as well as $a_j:=i\cdot e_j\in\mathbb{C}^{n+1}$ for all $1\leq j\leq n+1$ (with $e_j$ the $j$th canonical basis vector of $\mathbb{C}^{n+1}$). Let us denote $S_j:=S_\mathbb{C}^n(a_j,r_j)$ for all $1\leq j\leq n+1$. Note that for the chosen values of $a$ and $r$, the $S_j$ are in general position. Indeed, let $I=\{i_1,\ldots,i_m\}\subset\{1,\ldots,n+1\}$ be non-empty and suppose there would exist $z\in\bigcap_{j\in I}S_j$ and $\lambda_1,\ldots,\lambda_{n+1}\in\mathbb{C}$ (not all zero) such that $\lambda_j=0$ for all $j\notin I$ and
\begin{equation}
    \sum_{i=1}^{n+1}\lambda_i\cdot(z-a_i)=0.
\end{equation}
Then either $\sum_{j=1}^{n+1}\lambda_j=0$, which leads to the contradiction $\sum_{j=1}^{n+1}\lambda_j\cdot a_j=0$ (since this implies $\lambda_1=\cdots=\lambda_{n+1}=0$), or we may assume $\sum_{j=1}^{n+1}\lambda_j=1$ without loss of generality (by dividing each $\lambda_j$ by $\sum_{j=1}^{n+1}\lambda_j$ if necessary). Then $z=(i\lambda_1,\ldots,i\lambda_{n+1})$. Now let $j\in I$. Then $z\in S_j$ implies $\sum_{k=1}^{n+1}(\lambda_k-\delta_{k,j})^2=-1$ (with $\delta_{j,k}$ the Kronecker-delta) and we may conclude $\lambda_{i_1}=\cdots=\lambda_{i_m}=\lambda\neq0$. From $\sum_{j=1}^{n+1}\lambda_j=|I|\cdot\lambda=1$, we obtain $\lambda=\frac{1}{|I|}\in\mathbb{R}_{>0}$. But then $z\in(i\cdot\mathbb{R})^{n+1}$, i.e. $z$ is purely imaginary, which means $(z-a_j)^2$ is non-positive, in particular not equal to $r_j^2=1$, for all $1\leq j\leq n+1$. This too is a contradiction.\\
How do the various intersections look with this choice of $a$ and $r$? Let $I\subset\{1,\ldots,n+1\}$ be non-empty and choose some $j_0\in I$. Then we have
\begin{equation}\label{eq:intersection_2}
    z\in\bigcap_{j\in I}S_j \qquad\Leftrightarrow\qquad z\in S_{j_0}\;\land\; \forall j\in I\backslash\{j_0\}:z_j=z_{j_0}
\end{equation}
by plugging in the equation for $S_{j_0}$ into the remaining ones. Plugging this back into the equation for $S_{j_0}$, we get
\begin{equation}
    \begin{split}
        &(z_{j_0}-i)^2+\sum_{j\in I\backslash\{j_0\}}z_{j_0}^2+\sum_{j\notin I}z_j^2=1\\
        \Leftrightarrow\qquad& |I|\cdot z_{j_0}^2-2i\cdot z_{j_0}+\sum_{j\notin I}z_j^2=2\\
        \Leftrightarrow\qquad& |I|\cdot(z_{j_0}-\frac{i}{|I|})^2+\sum_{j\notin I}z_j^2=2-\frac{1}{|I|}
    \end{split}
\end{equation}
Hence, the real $(n+1-|I|)$-sphere\footnote{This is a slight abuse of language, since the intersection $\bigcap_{j\in I}S_j$ retracts to a space only homeomorphic to a sphere. Nevertheless, equation \eqref{eq:intersection_2} shows that the homeomorphism can be realized by simply rescaling the coordinates $z_j$ with $j\in I$ and a translation. So it still makes sense to speak of the center of this sphere.} to which $\bigcap_{j\in I}S_j$ deformation retracts is
\begin{equation}\label{eq:real_sphere}
    S_\mathbb{R}(I):=\{z\in S_{j_0} \;|\; \forall j\in I:\text{Im}\,z_j=\frac{1}{|I|},\; \forall j\notin I:\text{Im}\,z_j=0,\; \forall j_1,j_2\in I:\text{Re}\,z_{j_1}=\text{Re}\,z_{j_2}\}.
\end{equation}
In particular, the center of $\bigcap_{j\in I}S_j$ is 
\begin{equation}
    c_I\in\mathbb{C}^{n+1} \qquad\text{with}\qquad (c_I)_j:=\begin{cases} \frac{i}{|I|} & \text{if }j\in I \\ 0 & \text{otherwise}. \end{cases}
\end{equation}
Furthermore, we can write $S_\mathbb{R}(I)=S_{j_0}\cap(c_I+V_I)$ as the intersection of $S_{j_0}$ with the affine $(n+2-|I|)$-space $c_I+V_I$ with
\begin{equation}
    V_I:=\{z\in\mathbb{R}^{n+1} \;|\; \forall j_1,j_2\in I:z_{j_1}=z_{j_2}\}
\end{equation}
Note that we have the following basis for the vector space $V_I$: Let
\begin{equation}
    v:=\begin{pmatrix} 1 \\ \vdots \\ 1 \end{pmatrix}\in\mathbb{R}^{n+1} \qquad\text{and}\qquad v_j:=v-e_j\in\mathbb{R}^{n+1}, \quad\forall j\in\{1,\ldots,n+1\}.
\end{equation}
Then $v,\{v\}_{j\notin I}$ is a basis of $V_I$. Later, we need the following subsets of $V_I$: For all disjoint subsets $J_\leq,J_\geq\subset\{1,\ldots,n+1\}\backslash I$ and all $\tau\in\{-1,0,1\}$, we define
\begin{equation}
    V_{J_\leq,J_\geq}^\tau:=\tau\cdot\mathbb{R}_{\geq0}\cdot v+\sum_{j\in J_\leq}\mathbb{R}_{\leq0}\cdot v_j+\sum_{j\in J_\geq}\mathbb{R}_{\geq0}\cdot v_j.
\end{equation}
Note that the relative interior of $V_{J_\leq,J_\geq}^\tau$ is
\begin{equation}
    \overset{\circ}{V}{}_{J_\leq,J_\geq}^\tau=\tau\cdot\mathbb{R}_{>0}\cdot v+\sum_{j\in J_\leq}\mathbb{R}_{<0}\cdot v_j+\sum_{j\in J_\geq}\mathbb{R}_{>0}\cdot v_j.
\end{equation}
Now, we first establish cells with purely imaginary support, i.e. support contained within $(i\cdot\mathbb{R})^{n+1}$, connecting the centers $c_I$. We clearly want each point $c_I$ itself to be a zero cell $e_I$. Furthermore, for each sequence $I_1\subsetneq\cdots\subsetneq I_k$, we consider the convex hull of $c_{I_1},\ldots,c_{I_k}$ to be a $(k-1)$-simplex $e_{I_1,\ldots,I_k}$. All these simplices together just form a geometric realization of the order complex associated to the power set of $\{1,\ldots,n+1\}$ together with inclusion as pre-order. The boundary of each simplex $e_{I_1,\ldots,I_k}$ can be expressed as
\begin{equation}
    \partial e_{I_1,\ldots,I_k}=\sum_{j=1}^k(-1)^{j-1}\cdot e_{I_1,\ldots,\widehat{I}_j,\ldots,I_k}.
\end{equation}
The complex we obtain this way is just the barycentric subdivision obtained from the $n$-simplex spanned by $e_{\{1\}},\ldots,e_{\{n+1\}}$. We have seen this complex for the case $n=2$ (i.e. for three complex 2-spheres) in Example \ref{ex:cw_1}, which is illustrated in Figure \ref{fig:imaginary simplices_1}. In Figure \ref{fig:imaginary simplices_2}, the case $n=3$ (i.e. for four complex 3-spheres) is depicted. Since the complex in the latter case lives in $(i\cdot\mathbb{R})^4$, we projected to the first three coordinates in this picture to obtain a graphical representation.

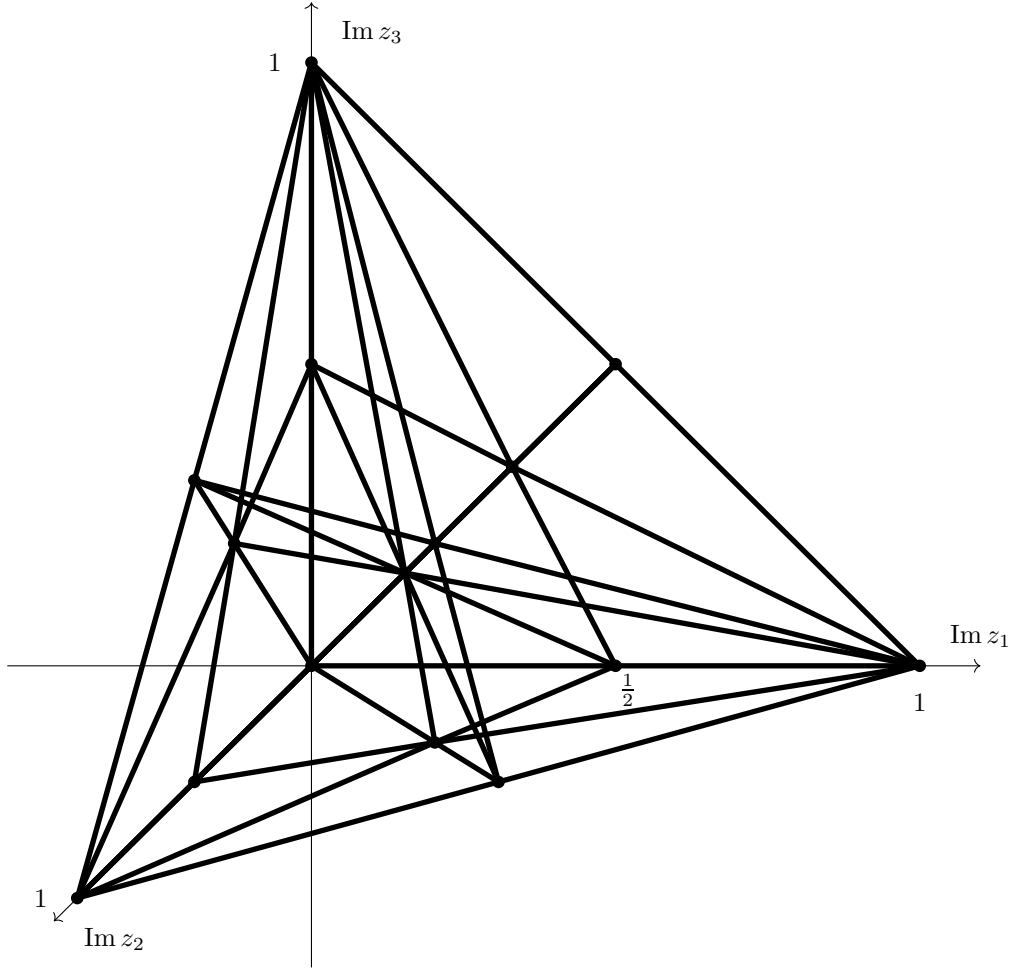
\begin{figure}
    \centering
    \def\distance{0.6}
    \def\radiusa{0.5}
    \def\radiusb{0.4}
    \def\stepsa{9}
    \def\stepsb{7}
    \begin{tikzpicture}[scale=8]
        \draw[->] (0,0,0) -- (xyz cylindrical cs:radius=1.1);
        \draw (0,0,0) -- (xyz cylindrical cs:radius=0.5,angle=180);
        \draw[->] (0,0,0) -- (xyz cylindrical cs:radius=1.1,angle=90);
        \draw (0,0,0) -- (xyz cylindrical cs:radius=0.5,angle=-90);
        \draw (0,0,0) -- (xyz cylindrical cs:z=-1);
        \draw[->] (0,0,0) -- (xyz cylindrical cs:z=1.1);

        \draw (1.1,0.05,0) node {$\text{Im}\,z_1$};
        \draw (1,-0.06,0) node {1};
        \draw (0.52,-0.04,0) node {$\frac{1}{2}$};
        \draw (0.1,-0.03,1.1) node {$\text{Im}\,z_2$};
        \draw (-0.06,0,1) node {1};
        \draw (0.1,1.05,0) node {$\text{Im}\,z_3$};
        \draw (-0.06,1,0) node {1};

        \draw (1,0,0) node [circle, fill, scale=0.5] {};
        \draw (0,1,0) node [circle, fill, scale=0.5] {};
        \draw (0,0,1) node [circle, fill, scale=0.5] {};
        \draw (0,0,0) node [circle, fill, scale=0.5] {};

        \draw (0.5,0,0) node [circle, fill, scale=0.5] {};
        \draw (0,0.5,0) node [circle, fill, scale=0.5] {};
        \draw (0,0,0.5) node [circle, fill, scale=0.5] {};
        \draw (0.5,0.5,0) node [circle, fill, scale=0.5] {};
        \draw (0.5,0,0.5) node [circle, fill, scale=0.5] {};
        \draw (0,0.5,0.5) node [circle, fill, scale=0.5] {};

        \draw (0.33,0.33,0.33) node [circle, fill, scale=0.5] {};
        \draw (0.33,0.33,0) node [circle, fill, scale=0.5] {};
        \draw (0.33,0,0.33) node [circle, fill, scale=0.5] {};
        \draw (0,0.33,0.33) node [circle, fill, scale=0.5] {};
        
        \draw (0.25,0.25,0.25) node [circle, fill, scale=0.5] {};

        \draw[line width=2pt] (0.5,0.5,0) -- (1,0,0);
        \draw[line width=2pt] (0.5,0.5,0) -- (0,1,0);
        \draw[line width=2pt] (0.5,0,0.5) -- (1,0,0);
        \draw[line width=2pt] (0.5,0,0.5) -- (0,0,1);
        \draw[line width=2pt] (0,0.5,0.5) -- (0,1,0);
        \draw[line width=2pt] (0,0.5,0.5) -- (0,0,1);
        \draw[line width=2pt] (0.5,0,0) -- (1,0,0);
        \draw[line width=2pt] (0.5,0,0) -- (0,0,0);
        \draw[line width=2pt] (0,0.5,0) -- (0,1,0);
        \draw[line width=2pt] (0,0.5,0) -- (0,0,0);
        \draw[line width=2pt] (0,0,0.5) -- (0,0,1);
        \draw[line width=2pt] (0,0,0.5) -- (0,0,0);

        \draw[line width=2pt] (0.33,0.33,0.33) -- (0.5,0.5,0);
        \draw[line width=2pt] (0.33,0.33,0.33) -- (0.5,0,0.5);
        \draw[line width=2pt] (0.33,0.33,0.33) -- (0,0.5,0.5);
        \draw[line width=2pt] (0.33,0.33,0.33) -- (1,0,0);
        \draw[line width=2pt] (0.33,0.33,0.33) -- (0,1,0);
        \draw[line width=2pt] (0.33,0.33,0.33) -- (0,0,1);
        \draw[line width=2pt] (0.33,0.33,0) -- (0.5,0.5,0);
        \draw[line width=2pt] (0.33,0.33,0) -- (0.5,0,0);
        \draw[line width=2pt] (0.33,0.33,0) -- (0,0.5,0);
        \draw[line width=2pt] (0.33,0.33,0) -- (1,0,0);
        \draw[line width=2pt] (0.33,0.33,0) -- (0,1,0);
        \draw[line width=2pt] (0.33,0.33,0) -- (0,0,0);
        \draw[line width=2pt] (0.33,0,0.33) -- (0.5,0,0.5);
        \draw[line width=2pt] (0.33,0,0.33) -- (0.5,0,0);
        \draw[line width=2pt] (0.33,0,0.33) -- (0,0,0.5);
        \draw[line width=2pt] (0.33,0,0.33) -- (1,0,0);
        \draw[line width=2pt] (0.33,0,0.33) -- (0,0,1);
        \draw[line width=2pt] (0.33,0,0.33) -- (0,0,0);
        \draw[line width=2pt] (0,0.33,0.33) -- (0,0.5,0.5);
        \draw[line width=2pt] (0,0.33,0.33) -- (0,0.5,0);
        \draw[line width=2pt] (0,0.33,0.33) -- (0,0,0.5);
        \draw[line width=2pt] (0,0.33,0.33) -- (0,1,0);
        \draw[line width=2pt] (0,0.33,0.33) -- (0,0,1);
        \draw[line width=2pt] (0,0.33,0.33) -- (0,0,0);
        
        \draw[line width=2pt] (0.25,0.25,0.25) -- (0.33,0.33,0.33);
        \draw[line width=2pt] (0.25,0.25,0.25) -- (0.33,0.33,0);
        \draw[line width=2pt] (0.25,0.25,0.25) -- (0.33,0,0.33);
        \draw[line width=2pt] (0.25,0.25,0.25) -- (0,0.33,0.33);
        \draw[line width=2pt] (0.25,0.25,0.25) -- (0.5,0.5,0);
        \draw[line width=2pt] (0.25,0.25,0.25) -- (0.5,0,0.5);
        \draw[line width=2pt] (0.25,0.25,0.25) -- (0,0.5,0.5);
        \draw[line width=2pt] (0.25,0.25,0.25) -- (0.5,0,0);
        \draw[line width=2pt] (0.25,0.25,0.25) -- (0,0.5,0);
        \draw[line width=2pt] (0.25,0.25,0.25) -- (0,0,0.5);
        \draw[line width=2pt] (0.25,0.25,0.25) -- (1,0,0);
        \draw[line width=2pt] (0.25,0.25,0.25) -- (0,1,0);
        \draw[line width=2pt] (0.25,0.25,0.25) -- (0,0,1);
        \draw[line width=2pt] (0.25,0.25,0.25) -- (0,0,0);
        
    \end{tikzpicture}
    \caption{The projection to the first three coordinates of the subcomplex consisting of all simplices with support contained in $(i\cdot\mathbb{R})^{n+1}$ for the case $n=3$. The 0-cells are the centers of the complex spheres given by the various intersections of the $S_i$.}
    \label{fig:imaginary simplices_2}
\end{figure}
Now we attach more cells to the complex, similarly to Example \ref{ex:cw_2}. To each cell $e_{I_1,\ldots,I_k}$, we associate a real part in an appropriate $V_I$. We first define the following index set:
\begin{equation}
    \mathcal{I}:=\{(I_1,\ldots,I_k,J_\leq,J_\geq,\prec) \;|\; \substack{\emptyset\neq I_1\subsetneq\cdots\subsetneq I_k\subset\{1,\ldots,n+1\},\;J_\leq, J_\geq\subset\{1,\ldots,n+1\}\backslash I,\\ J_\leq\cap J_\geq=\emptyset,\;\prec\in\{\leq,=\}}\}.
\end{equation}
The idea is that the $I_1,\ldots,I_k$ encode the previously introduced simplices $e_{I_1,\ldots,I_k}$, the sets $J_\leq,J_\geq$ encode which subspace $V_{J_\leq,J_\geq}^\tau$ of $V_I$ the real part belongs to, including the signs (or the vanishing) of the coefficients of the basis vectors and $\prec$ encodes if we consider values of $z$ with $(z-a_j)^2\leq r_j^2$ (for the relative homology) or $(z-a_j)^2=r_j^2$ (for the reduced homology). For every $i=(I_1,\ldots,I_k,J_\leq,J_\geq,\prec)\in\mathcal{I}$, we denote
\begin{equation}
    J_i:=J_\leq\cup J_\geq \qquad\text{and}\qquad |i|:=k.
\end{equation}
Now for each $i=(I_1,\ldots,I_k,J_\leq,J_\geq,\prec)\in\mathcal{I}$, we define
\begin{equation}
    e_i^\tau:=\{z\in\mathbb{C}^{n+1} \;|\; i\cdot\text{Im}\,z\in e_{I_1,\ldots,I_k},\;\text{Re}\,z\in V_{J_\leq,J_\geq}^\tau,\;(z-a_{\min I_1})^2\prec1\}.
\end{equation}
We want to show that the set $\bigcup_{\tau\in\{-1,0,1\}}\bigcup_{i\in\mathcal{I}}e_i^\tau$ has the structure of a CW-complex with cells $e_i^\tau$. First, we show that the relative interior of all these sets are disjoint so that we indeed obtain a partition of the above union into open sets.
\begin{lem}\label{lem:partition}
    Let $i_1,i_2\in\mathcal{I}$ and $\tau_1,\tau_2\in\{-1,0,1\}$ with $i_1\neq i_2$ or $\tau_1\neq\tau_2$. Then the relative interior $\overset{\circ}{e}{}_{i_1}^{\tau_1}$ of $e_{i_1}^{\tau_1}$ and $\overset{\circ}{e}{}_{i_2}^{\tau_2}$ of $e_{i_2}^{\tau_2}$ are disjoint.
\end{lem}
\begin{proof}
    First note that the relative interior of any $e_i^\tau$ for $i=(I_1,\ldots,I_k,J_\leq,J_\geq,\prec)$ is given by
    \begin{equation}
        \overset{\circ}{e}{}_i^\tau=\{z\in\mathbb{C}^{n+1} \;|\; i\cdot\text{Im}\,z\in \overset{\circ}e_{I_1,\ldots,I_k},\;\text{Re}\,z\in\overset{\circ}{V}{}_{J_\leq,J_\geq}^\tau,\;(z-a_{\min I_1})^2<r_{\min I_1}^2\}
    \end{equation}
    if $\prec$ is $\leq$ and
    \begin{equation}
        \overset{\circ}{e}{}_i^\tau=\{z\in\mathbb{C}^{n+1} \;|\; i\cdot\text{Im}\,z\in \overset{\circ}e_{I_1,\ldots,I_k},\;\text{Re}\,z\in\overset{\circ}{V}{}_{J_\leq,J_\geq}^\tau,\;(z-a_{\min I_1})^2=r_{\min I_1}^2\}
    \end{equation}
    if $\prec$ is $=$. Now denote $i_j=(I_{1,j},\ldots,I_{k_j,j},J_{\leq,j},J_{\geq,j},\prec_j)$ for $j=1,2$ and let us assume $z\in\overset{\circ}{e}{}_{i_1}^{\tau_1}\cap\overset{\circ}{e}{}_{i_2}^{\tau_2}$. Then in particular $i\cdot\text{Im}\,z\in\overset{\circ}e_{I_{1,1},\ldots,I_{k_1,1}}\cap\overset{\circ}e_{I_{1,2},\ldots,I_{k_2,2}}$ since the $e_{I_1,\ldots,I_k}$ form a simplicial complex. This implies $k_1=k_2=:k$ and $I_{j,1}=I_{j,2}$ for all $j\in\{1,\ldots,k\}$. Note that this also immediately implies $\prec_1=\prec_2$ since $\min I_{1,1}=\min I_{1,2}$. Now, fix some $j\in I_{k,1}$. Then $\text{Re}\,z$ can be uniquely expressed as
    \begin{equation}
        \text{Re}\,z=\lambda\cdot v+\sum_{\substack{k=1 \\ k\neq j}}^{n+1}\lambda_k\cdot v_k
    \end{equation}
    since $v,\{v_k\}_{k\neq j}$ is a basis of $\mathbb{R}^{n+1}$. Now, $k\in J_{\leq,1}\Leftrightarrow \lambda_k>0\Leftrightarrow k\in J_{\leq,2}$ so that $J_{\leq,1}=J_{\leq,2}$. Similarly, we obtain $J_{\geq,1}=J_{\geq,2}$ and $\tau_1=\tau_2$. We conclude that $(i_1,\tau_1)=(i_2,\tau_2)$, which complete the proof.
\end{proof}
Next, we want to identify for each $i=(I_1,\ldots,I_k,J_\leq,J_\geq,\prec)$ and $\tau\in\{-1,0,1\}$ the $e_i^\tau$ with the product of the simplex $e_{I_1,\ldots,I_k}$ and a $(|J_i|+|\tau|)$-simplex if $\prec$ is $\leq$ (resp. $(|J_i|+|\tau|-1)$-simplex if $\prec$ is $=$).\footnote{Here, we implicitly use the convention that a $k$-simplex is empty if $k<0$.} Recall the standard notation for $p$-norm of a real vector. For every $x\in\mathbb{R}^k$, we denote
\begin{equation}
    \Vert x\Vert_p=\sqrt[p]{\sum_{i=1}^k|x_i|^p}.
\end{equation}
Also recall that the standard $k$-simplex is typically defined as
\begin{equation}\label{eq:standard_simplex}
    \{x\in\mathbb{R}^{k+1} \;|\; \forall1\leq j\leq k+1:x_j\geq0,\;\sum_{i=1}^{k+1}x_i=1\}.
\end{equation}
For us, it is more convenient to employ a slight variation of the standard simplex, namely the realization as a corner of a cube. So we set
\begin{equation}
    \Delta^k:=\{x\in\mathbb{R}^k \;|\; x_i\geq0,\; \Vert x\Vert_1\leq1\}.
\end{equation}
This is homeomorphic to \eqref{eq:standard_simplex}. We have the following
\begin{lem}\label{lem:partial_sphere}
    Let $n\in\mathbb{N}$, let $D^n\subset\mathbb{R}^n$ be the closed unit disk of dimension $n$ around 0, let $\tau\in\{-1,1\}^n$ and let
    \begin{equation}
        X:=(\tau_1\cdot\mathbb{R}_{\geq0})\times\cdots\times(\tau_n\cdot\mathbb{R}_{\geq0}).
    \end{equation}
    The map
    \begin{equation}\label{eq:partial_sphere}
        h:D^n\cap X\to\Delta^n, \qquad x\mapsto\begin{cases}\frac{\Vert x\Vert_2}{\Vert x\Vert_1}\cdot(\tau_1\cdot x_1,\ldots,\tau_n\cdot x_n) & \text{if }x\neq0 \\ 0 & \text{if }x=0\end{cases}
    \end{equation}
    is a homeomorphism with inverse
    \begin{equation}
        h^{-1}:\Delta^n\to D^n\cap X, \qquad x\mapsto\begin{cases}\frac{\Vert x\Vert_1}{\Vert x\Vert_2}\cdot (\tau_1\cdot x_1,\ldots,\tau_n\cdot x_n) & \text{if }x\neq0 \\ 0 & \text{if }x=0.\end{cases}
    \end{equation}
    Moreover, $h$ restricts to a homeomorphism
    \begin{equation}
        h|_{S^{n-1}\cap X}:S^{n-1}\cap X\overset{\sim}\to\iota(\Delta^{n-1}),
    \end{equation}
    where $\iota$ is the inclusion 
    \begin{equation}
        \iota:\Delta^{n-1}\hookrightarrow\Delta^n,\qquad x\mapsto(x,1-\sum_{j=1}^{n-1}x_i)
    \end{equation}
    into the last face of $\Delta^n$.
\end{lem}
\begin{proof}
    First, we show that $h$ and $h^{-1}$ are well-defined in the sense that they take values in the correct codomains. Let $x\in D^n\cap X$. Then clearly every coordinate of $h(x)$ is positive and we have
    \begin{equation}
        \Vert h(x)\Vert_1=\Vert x\Vert_2\leq1.
    \end{equation}
    Thus $h(x)\in\Delta^n$. The other way around, if $x\in\Delta^n$ then again every coordinate of $h^{-1}(x)$ is positive and
    \begin{equation}
        \Vert h^{-1}(x)\Vert_2=\Vert x\Vert_1\leq1.
    \end{equation}
    Now we show that $h$ is a continuous map. It is obvious that $h$ is continuous at every point $x\neq0$ and it suffices to show that $h$ is continuous at 0. All norms in $\mathbb{R}^{n+1}$ are equivalent so that there are constants $m_1,m_2,M_1,M_2\in\mathbb{R}_{>0}$ such that
    \begin{equation}
        m_1\Vert x\Vert_1\leq\Vert x\Vert_2\leq M_1\Vert x\Vert_1 \qquad\text{and}\qquad m_2\Vert x\Vert_2\leq\Vert x\Vert_1\leq M_2\Vert x\Vert_2
    \end{equation}
    for all $x\in D^n\backslash\{0\}$. Thus, $\lim_{x\to0}h(x)=0=h(0)=h(\lim_{x\to0}x)$. Similarly, one shows that $h^{-1}$ is continuous. To show that these maps are inverse to each other, let $x\in D^n\cap X$. In case $x=0$, we have $h^{-1}(h(0))=h^{-1}(0)=0=x$. If on the other hand $x\neq0$, we have
    \begin{equation}
        h^{-1}(h(x))=h^{-1}(\frac{\Vert x\Vert _2}{\Vert x\Vert_1}\cdot x)=\frac{\Vert\frac{\Vert x\Vert_2}{\Vert x\Vert_1}\cdot x\Vert_1}{\Vert\frac{\Vert x\Vert_2}{\Vert x\Vert_1}\cdot x\Vert_2}\cdot\frac{\Vert x\Vert_2}{\Vert x\Vert_1}\cdot x=\frac{\Vert x\Vert_1}{\Vert x\Vert_2}\cdot\frac{\Vert x\Vert_2}{\Vert x\Vert_1}\cdot x=x.
    \end{equation}
    The other way around, let $x\in\Delta^n$. Then if $x=0$ we have $h(h^{-1}(0))=h(0)=0=x$ and if $x\neq0$ we have
    \begin{equation}
        h(h^{-1}(x))=h(\frac{\Vert x\Vert_1}{\Vert x\Vert_2}\cdot x)=\frac{\Vert\frac{\Vert x\Vert_1}{\Vert x\Vert_2}\cdot x\Vert_2}{\Vert\frac{\Vert x\Vert_1}{\Vert x\Vert_2}\cdot x\Vert_1}\cdot\frac{\Vert x\Vert_1}{\Vert x\Vert_2}\cdot x=\frac{\Vert x\Vert_2}{\Vert x\Vert_1}\cdot\frac{\Vert x\Vert_1}{\Vert x\Vert_2}\cdot x=x.
    \end{equation}
    For the second part of the statement, let $x\in S_n\cap X$. Then $x\neq0$ as well as $\Vert x\Vert_2=1$ and we compute
    \begin{equation}
        \Vert h(x)\Vert_1=\Vert\frac{\Vert x\Vert_2}{\Vert x\Vert_1}\cdot x\Vert_1=\Vert x\Vert_2=1.
    \end{equation}
    Thus $\tau(x)\in\iota(\Delta^{n-1})$.
\end{proof}
To keep track of the faces, we need a little bit more notation. Let us consider the various inclusions of the faces in a product of two simplices. Let $m,n\in\mathbb{N}$. For every $1\leq j\leq m$, we have the inclusions
\begin{equation}
    \iota_j^{(1)}:\Delta^{m-1}\times\Delta^n\hookrightarrow\Delta^m\times\Delta^n,\qquad (x,y)\mapsto(x_1,\ldots,x_{j-1},0,x_j,\ldots,x_m,y)
\end{equation}
\begin{equation}
    \iota_{m+1}^{(1)}:\Delta^{m-1}\times\Delta^n\hookrightarrow\Delta^m\times\Delta^n,\qquad (x,y)\mapsto(x_1,\ldots,x_{m-1},1-\sum_{j=1}^{m-1}x_j,y)
\end{equation}
into the $m+1$ faces of the simplex $\Delta^m$. Similarly, we denote for every $1\leq j\leq n$ the inclusions
\begin{equation}
    \iota_j^{(2)}:\Delta^m\times\Delta^{n-1}\hookrightarrow\Delta^m\times\Delta^n,\qquad (x,y)\mapsto(x,y_1,\ldots,y_{j-1},0,y_j,\ldots,y_n)
\end{equation}
\begin{equation}
    \iota_{n+1}^{(2)}:\Delta^m\times\Delta^{n-1}\hookrightarrow\Delta^m\times\Delta^n,\qquad (x,y)\mapsto(x,y_1,\ldots,y_{n-1},1-\sum_{j=1}^{n-1}y_j)
\end{equation}
into the $n+1$ faces of the simplex $\Delta^n$. These inclusions clearly depend on $m$ and $n$. But the relevant values are always obvious from context and we therefore drop this dependence from our notation. Furthermore, for every $i\in\mathcal{I}$, we denote by
\begin{equation}
    n_i:J_i\overset{\sim}\to\{1,\ldots,|J_i|\}
\end{equation}
the unique order-preserving bijection (with respect to the total order induces by the standard order relation on $\mathbb{Z}$).
The following proposition establishes explicit homeomorphisms between the sets $e_i^\tau$ and products of two simplices.
\begin{prop}\label{prop:homeomorphism_1}
    Let
    \begin{equation}
        i=(I_1,\ldots,I_k,J_\leq,J_\geq,\leq),\quad i'=(I_1,\ldots,I_k,J_\leq,J_\geq,=)\in\mathcal{I},
    \end{equation}
    and $\tau\in\{-1,0,1\}$. Denote $m:=|J_i|+|\tau|$. Then there is a homeomorphism
    \begin{equation}
        \phi_i^\tau:e_i^\tau\overset{\sim}\to\Delta^m\times\Delta^k,
    \end{equation}
    which restricts to a homeomorphism
    \begin{equation}
        \phi_{i'}^\tau:=\phi_i^\tau|_{e_{i'}^\tau}:e_{i'}^\tau\overset{\sim}\to\iota_{m+1}^{(1)}(\Delta^{m-1}\times\Delta^k)\simeq\Delta^{m-1}\times\Delta^k.
    \end{equation}
    Furthermore for $i''=(I_1,\ldots,I_k,J_\leq,J_\geq,\prec)\in\mathcal{I}$, we have the following identities:
    \begin{alignat}{2}
        \phi_i^\tau\circ(\phi_{i'}^\tau)^{-1}&=\iota_{m+1}^{(1)}&&,\label{eq:homeomorphism_1_1}\\
        \phi_{i''}^\tau\circ(\phi_{(I_1,\ldots,I_k,J_\leq-\{j\},J_\geq-\{j\},\tau,\prec)}^\tau)^{-1}&=\iota_{n_{i''}(j)}^{(1)} \qquad&&\forall j\in J_i,\label{eq:homeomorphism_1_2}\\
        \phi_{i''}^\tau\circ(\phi_{(I_1,\ldots,I_k,J_\leq,J_\geq,\prec)}^0)^{-1}&=\iota_m^{(1)} \qquad&&\text{if } \tau\neq0,\label{eq:homeomorphism_1_3}\\
        \phi_{i''}^\tau\circ(\phi_{(I_1,\ldots,\widehat{I}_j,\ldots,I_k,J_\leq,J_\geq,\prec)}^\tau)^{-1}&=\iota_j^{(2)} \qquad&&\forall j\in\{1,\ldots,k\},\label{eq:homeomorphism_1_4}
    \end{alignat}
\end{prop}
\begin{proof}
    We start by construction the homeomorphism $\phi_i^\tau$. Let us denote $j_0:=\min I_1$ and
    \begin{equation}
        c:\mathbb{C}^{n+1}\overset{\sim}\to(\mathbb{R}^{n+1})^2,\qquad z\mapsto(\text{Re}\,z,\text{Im}\,z).
    \end{equation}
    Then for every element $(x,y)\in c(e_i^\tau)$, the first entry $x$ can be expressed as
    \begin{equation}
        x=d_{j_m}\cdot v+\sum_{j\in J_i}d_j\cdot v_j
    \end{equation}
    with $d_{j_m}\in\tau\cdot\mathbb{R}_{\geq0}$ as well as $d_j\in\mathbb{R}_{\leq0}$ if $j\in J_\leq$ and $d_j\in\mathbb{R}_{\geq0}$ if $j\in J_\geq$. Write $J_i:=\{j_1,\ldots,j_{|J_i|}\}$ with $j_1<\cdots<j_{|J_i|}$ and denote
    \begin{equation}
        X:=(\tau_1\cdot\mathbb{R}_{\geq0})\times\cdots\times(\tau_{|J_i|}\cdot\mathbb{R}_{\geq0})\times(\tau\cdot\mathbb{R}_{\geq0}),
    \end{equation}
    where $\tau_k=-1$ if $j_k\in J_\leq$, $\tau_k=1$ if $j_k\in J_\geq$ and $\tau_k=0$ otherwise.
    Then there is a change of basis homeomorphism
    \begin{equation}
        g_i^\tau:V_{J_\leq,J_\geq}^\tau\overset{\sim}\to X,\qquad x\mapsto(d_{j_1},\ldots,d_{j_m}).
    \end{equation}
    The defining inequality for $(g_i^\tau,\text{id}_{\mathbb{R}^{n+1}})(c(e_i^\tau))$ reads
    \begin{equation}\label{eq:cells_1}
        \sum_{k=1}^{n+1}(\sum_{k=1}^md_{j_k})^2\leq 1+(y-e_{j_0})^2.
    \end{equation}
    On the left hand side, we have a positive-definite quadratic form in the $d_{j_1},\ldots,d_{j_m}$ and thus there exists a real invertible $m\times m$-matrix $M_i^\tau$, sending $(d_{j_1},\ldots,d_{j_m})$ to $(d_{j_1}',\ldots,d_{j_m}')$ such that equation \eqref{eq:cells_1} reads
    \begin{equation}
        \sum_{k=1}^m(d_{j_k}')^2\leq 1+(y-e_{j_0})^2.
    \end{equation}
    Now we divide by the modified radius
    \begin{equation}
        \tilde{r}_i(y):=\sqrt{1+(y-e_{j_0})^2} \quad\in\quad\mathbb{R}_{>0}
    \end{equation}
    and apply the homeomorphism $h$ from Lemma \ref{lem:partial_sphere}. To simplify the notation, let us write
    \begin{equation}
        \begin{split}
            f_i^\tau:\{(d,y)\in X\times\Delta^k \;|\; d^2\leq\tilde{r}_i^2(y)\}&\to\Delta^m\times\Delta^k,\\
            (d,y)&\mapsto(h(\frac{M_i^\tau\cdot d}{\tilde{r}_i(y)}),y).
        \end{split}
    \end{equation}
    The map $f_i$ is a homeomorphism with continuous inverse
    \begin{equation}
        \begin{split}
            (f_i^\tau)^{-1}:\Delta^m\times\Delta^k&\to\{(d,y)\in X\times\Delta^k \;|\; d^2\leq\tilde{r}_i^2(y)\},\\
            (x,y)&\mapsto((M_i^\tau)^{-1}\cdot(\tilde{r}_i(y)\cdot h^{-1}(x)),y).
        \end{split}
    \end{equation}
    Indeed, for all $(d,y)\in\{(d,y)\in X\times\Delta^k \;|\; d^2\leq\tilde{r}_i^2(y)\}$ we have
    \begin{equation}
        \begin{split}
            ((f_i^\tau)^{-1}\circ f_i^\tau)(d,y)&=(f_i^\tau)^{-1}(h(\frac{M_i^\tau\cdot d}{\tilde{r}_i(y)}),y)\\
            &=((M_i^\tau)^{-1}\cdot\tilde{r}_i(y)\cdot h^{-1}(h(\frac{M_i^\tau\cdot d}{\tilde{r}_i(y)})),y)\\
            &=((M_i^\tau)^{-1}\cdot M_i^\tau\cdot d,y)=(d,y)
        \end{split}
    \end{equation}
    and for all $(x,y)\in\Delta^m\times\Delta^k$ we have
    \begin{equation}
        \begin{split}
            (f_i^\tau\circ(f_i^\tau)^{-1})(x,y)&=f_i^\tau((M_i^\tau)^{-1}\cdot\tilde{r}_i(y)\cdot h^{-1}(x),y)\\
            &=(h(\frac{M_i^\tau\cdot(M_i^\tau)^{-1}\cdot\tilde{r}_i(y)\cdot h^{-1}(x)}{\tilde{r}_i(y)}),y)\\
            &=(h(h^{-1}(x)),y)=(x,y).
        \end{split}
    \end{equation}
    Then we can define the sought homeomorphism by
    \begin{equation}
        \phi_i^\tau:=f_i^\tau\circ(g_i^\tau,\text{id}_{\mathbb{R}^{n+1}})\circ c\quad:\quad e_i^\tau\overset{\sim}\to\Delta^m\times\Delta^k,
    \end{equation}
    which has continuous inverse
    \begin{equation}
        (\phi_i^\tau)^{-1}=c^{-1}\circ(g_i^\tau,\text{id}_{\mathbb{R}^{n+1}})^{-1}\circ(f_i^\tau)^{-1}\quad:\quad\Delta^m\times\Delta^k\overset{\sim}\to e_i^\tau.
    \end{equation}
    Next, we consider the restriction $\phi_{i'}^\tau$ of $\phi_i^\tau$ to $e_{i'}^\tau$. It suffices to show that $\phi_i^\tau$ sends elements in $e_{i'}^\tau$ to $\iota_{m+1}^{(1)}(\Delta^{m-1}\times\Delta^k)$ and the other way around for $(\phi_i^\tau)^{-1}$. So let $z\in e_{i'}^\tau$. Then applying $(g_i^\tau,\text{id}_{\mathbb{R}^{n+1}})\circ c$ sends $z$ to some $(d,y)$. Applying the matrix $M_i^\tau$ to $d$ yields some $(d',y)$ such that $(d')^2=(\tilde{r}_i(y))^2$. After diving $d'$ by $\tilde{r}_i(y)$, the result lies on the unit sphere. By Lemma \ref{lem:partial_sphere}, $h$ sends such an element to $\iota(\Delta^{m-1})$ so that the claim follows after using the fact that $\iota_{m+1}^{(1)}=(\iota,\text{id})$. The other way around, let $(x,y)\in\iota_{m+1}^{(1)}(\Delta^{m-1}\times\Delta^k)$. Then similarly to the first case, applying $h^{-1}$ to $x$ and multiplying by $\tilde{r}_i(y)$ yields an element $(d,y)$ with $d^2=(r_i(y))^2$. By the same logic as above, this gets send to some $z\in e_{i'}^\tau$.\\
    Now we prove the second part of the statement. Equation \eqref{eq:homeomorphism_1_1} follows immediately from the preceding discussion. For equation \eqref{eq:homeomorphism_1_2} to \eqref{eq:homeomorphism_1_4}, we first compute
    \begin{equation}
        \phi_i^\tau\circ(\phi_{i'}^\tau)^{-1}=f_i^\tau\circ(g_i^\tau,\text{id}_{\mathbb{R}^{n+1}})\circ(g_{i'}^\tau,\text{id}_{\mathbb{R}^{n+1}})^{-1}\circ(f_{i'}^\tau)^{-1}.
    \end{equation}
    To show equation equation \eqref{eq:homeomorphism_1_2}, denote $\tilde{i}=(I_1,\ldots,I_k,J_\leq-\{j\},J_\geq-\{j\},\prec)$. The homeomorphism $(f_{\tilde{i}}^\tau)^{-1}$ sends $(x,y)\in\Delta^{m-1}\times\Delta^k$ to some $(d,y)$ with $d^2\prec\tilde{r}_{\tilde{i}}^2(y)=\tilde{r}_i^2(y)$. Then $(g_{\tilde{i}}^\tau)^{-1}$ sends $d$ to
    \begin{equation}
        d_{j_{m-1}}\cdot v+\sum_{j\in J_i-\{j\}}d_j\cdot v_j,
    \end{equation}
    which $g_{i''}^\tau$ sends to $(d_{j_1},\ldots,d_{j_{n_i(j)}-1},0,d_{j_{n_i(j)}},\ldots,d_{j_{m-1}})$ and the claim follows. Similarly, one shows equation \eqref{eq:homeomorphism_1_3}. Equation \eqref{eq:homeomorphism_1_4} follows directly from the structure of the simplical complex with purely imaginary support.
\end{proof}
Using the established homeomorphisms, we can endow the set
\begin{equation}
    Y:=\bigcup_{\tau\in\{-1,0,1\}}\bigcup_{i\in\mathcal{I}}e_i^\tau\quad\subset\quad\mathbb{C}^{n+1}
\end{equation}
with the structure of a CW-complex of dimension $n+1$. To unify the notation, we employ the Kronecker-Delta symbol $\delta_{a,b}$ which is 1 if $a=b$ and 0 otherwise.
\begin{cor}\label{cor:CW-complex}
    The set $Y$ has the structure of a regular CW-complex with cells $e_i^\tau$. Furthermore, the CW-structure is such that for all $i=(I_1,\ldots,I_k,J_\leq,J_\geq,\prec)\in\mathcal{I}$ and all $\tau\in\{-1,0,1\}$, we have the following identity for the boundary $\partial e_i^\tau$ of $e_i^\tau$:
    \begin{equation}\label{eq:boundary_c_a}
        \begin{split}
            \partial e_i^\tau&=\delta_{|\tau|,1}\cdot(-1)^{|J_i|}\cdot e_{(I_1,\ldots,I_k,J_\leq,J_\geq,\prec)}^0\\
            &+\delta_{\prec,\leq}\cdot(-1)^{|J_i|+|\tau|}\cdot e_{(I_1,\ldots,I_k,J_\leq,J_\geq,=)}^\tau+\sum_{j\in J_i}(-1)^{n_i(j)-1}\cdot e_{(I_1,\ldots,I_k,J_\leq-\{j\},J_\geq-\{j\},\leq)}^\tau\\
            &+(-1)^{|J_i|+|\tau|+\delta_{\prec,\leq}-1}\cdot\sum_{j=1}^k(-1)^{j-1}\cdot e_{(I_1,\ldots,\widehat{I}_j,\ldots,I_k,J_\leq,J_\geq,\leq)}^\tau.
        \end{split}
    \end{equation}
\end{cor}
\begin{proof}
    By the discussion above, in particular Proposition \ref{prop:homeomorphism_1}, we can view each $e_i^\tau$ as a cell and equations \eqref{eq:homeomorphism_1_1} to \eqref{eq:homeomorphism_1_4} show that the gluing maps are in fact homeomorphisms. By Lemma \ref{lem:partition}, the $e_i^\tau$ form a partition of $Y$ into open cells. Thus, $Y$ has the structure of a regular CW-complex with cells $e_i^\tau$ and it remains to compute the boundary of each cell $e_i^\tau$.\\
    Let $i=(I_1,\ldots,I_k,J_\leq,J_\geq,\prec)\in\mathcal{I}$, $\tau\in\{-1,0,1\}$ and denote $m:=|J_i|+|\tau|$. We use the fact that the homomorphisms induced by continuous maps on the level of chains commute with the boundary:
    \begin{equation}
        \partial e_i^\tau=\partial((\phi_i^\tau)^{-1}\circ\phi_i^\tau)(e_i^\tau)=(\phi_i^\tau)^{-1}\partial(\phi_i^\tau(e_i^\tau)).
    \end{equation}
    Here, we denoted the homomorphisms induces by $\phi_i^\tau$ and $(\phi_i^\tau)^{-1}$ on the level of chains by the same symbol to not clutter the notation further. If $\prec$ is $\leq$, we calculate
    \begin{equation}
        \partial(\phi_i^\tau(e_i^\tau))=\partial(\Delta^m\times\Delta^k)=\sum_{j=1}^{m+1}(-1)^{j-1}\iota_j^{(1)}(\Delta^{m-1}\times\Delta^k)+(-1)^m\cdot\sum_{j=1}^{k+1}(-1)^{j-1}\iota_j^{(2)}(\Delta^m\times\Delta^{k-1}).
    \end{equation}
    Similarly, if $\prec$ is $=$, we calculate
    \begin{equation}
        \partial(\phi_i^\tau(e_i^\tau))=\partial(\Delta^{m-1}\times\Delta^k)=\sum_{j=1}^m(-1)^{j-1}\iota_j^{(1)}(\Delta^{m-2}\times\Delta^k)+(-1)^{m-1}\cdot\sum_{j=1}^{k+1}(-1)^{j-1}\iota_j^{(2)}(\Delta^{m-1}\times\Delta^{k-1}).
    \end{equation}
    As for $\phi_i^\tau$ and $(\phi_i^\tau)^{-1}$, we denote the homomorphisms induced by $\iota_k^{(1)}$ and $\iota_k^{(2)}$ by the same symbol. Applying $(\phi_i^\tau)^{-1}$ to both sides yields
    \begin{equation}
        \begin{split}
            \partial e_i^\tau&=\delta_{|\tau|,1}\cdot(-1)^{m-1}\cdot e_i^0+(-1)^m\cdot e_{(I_1,\ldots,I_k,J_\leq,J_\geq,=)}^\tau+\sum_{j\in J_i}(-1)^{n_i(j)-1}\cdot e_{(I_1,\ldots,I_k,J_\leq-\{j\},J_\geq-\{j\},\leq)}\\
            &+(-1)^m\cdot\sum_{j=1}^k(-1)^{j-1}\cdot e_{(I_1,\ldots,\widehat{I}_j,\ldots,I_k,J_\leq,J_\geq,\leq)},
        \end{split}
    \end{equation}
    resp.
    \begin{equation}
        \begin{split}
            \partial e_i^\tau&=\delta_{|\tau|,1}\cdot(-1)^{m-1}\cdot e_i^0+\sum_{j\in J_i}(-1)^{n_i(j)-1}\cdot e_{(I_1,\ldots,I_k,J_\leq-\{j\},J_\geq-\{j\},\leq)}\\
            &+(-1)^{m-1}\cdot\sum_{j=1}^k(-1)^{j-1}\cdot e_{(I_1,\ldots,\widehat{I}_j,\ldots,I_k,J_\leq,J_\geq,\leq)},
        \end{split}
    \end{equation}
    according to Proposition \ref{prop:homeomorphism_1}. Now combining these equations yields the desired equation \eqref{eq:boundary_c_a}.
\end{proof}
Now, the combinatorial part of this work begins. In the above corollary, the notation is somewhat unwieldy. To simplify the notation a bit, we make the following definitions, in particular to express the boundary of cells in $Y$ in a more convenient manner. Let $i=(I_1,\ldots,I_k,J_\leq,J_\geq,\prec)\in\mathcal{I}$. If $\prec$ is $\leq$, we denote
\begin{equation}
    \mathcal{D}^=(i):=(I_1,\ldots,I_k,J_\leq,J_\geq,=).
\end{equation}
Furthermore, for every $1\leq j\leq k$, we write
\begin{equation}
    \mathcal{D}_j^I(i):=(I_1,\ldots,\widehat{I}_j,\ldots,I_k,J_\leq,J_\geq,\prec)
\end{equation}
and for every $j\in J_i$, we write
\begin{equation}
    \mathcal{D}_j^J(i):=(I_1,\ldots,I_k,J_\leq-\{j\},J_\geq-\{j\},\prec).
\end{equation}
Then result of Corollary \ref{cor:CW-complex} can be expressed as follows: For every $i=(I_1,\ldots,I_k,J_\leq,J_\geq,\prec)\in\mathcal{I}$ and $\tau\in\{-1,0,1\}$, we have
\begin{equation}\label{eq:boundary_simplex}
    \begin{split}
        \partial e_i^\tau&=\delta_{|\tau|,1}\cdot(-1)^{|J_i|}\cdot e_i^0+\delta_{\prec,\leq}\cdot(-1)^{|J_i|+|\tau|}\cdot e_{\mathcal{D}^=(i)}^\tau+\sum_{j\in J_i}(-1)^{n_i(j)-1}\cdot e_{\mathcal{D}_j^J(i)}^\tau\\
            &+(-1)^{|J_i|+|\tau|+\delta_{\prec,\leq}-1}\cdot\sum_{j=1}^k(-1)^{j-1}\cdot e_{\mathcal{D}_j^I(i)}^\tau.
    \end{split}
\end{equation}
To make our life a little easier, we group some of the cells together. It turns out that the cells $e_i^0$ do not play a relevant role for our generators and that the cells $e_i^1$ and $e_i^{-1}$ always appear together. Therefore, we set
\begin{equation}
    e_i:=e_i^1-e_i^{-1}
\end{equation}
for every $i\in\mathcal{I}$. Then the $(\tau=0)$-faces cancel in the boundary of $e_i$:
\begin{lem}
    Let $i=(I_1,\ldots,I_k,J_\leq,J_\geq,\prec)\in\mathcal{I}$. Then
    \begin{equation}
        \partial e_i=\delta_{\prec,\leq}\cdot(-1)^{|J_i|+1}\cdot e_{\mathcal{D}^=(i)}+\sum_{j\in J_i}(-1)^{n_i(j)-1}\cdot e_{\mathcal{D}_j^J(i)}+(-1)^{|J_i|+\delta_{\prec,\leq}}\cdot\sum_{j=1}^k(-1)^{j-1}\cdot e_{\mathcal{D}_j^I(i)}.
    \end{equation}
\end{lem}
\begin{proof}
    Using equation \eqref{eq:boundary_simplex}, this is obtained from a straightforward computation.
\end{proof}
In a next step, we want to group even more cells together. For this purpose, we consider the extended index set
\begin{equation}
    \tilde{\mathcal{I}}:=\{(I_1,\ldots,I_k,J_\leq,J_\geq,\prec) \;|\; \substack{\emptyset\neq I_1\subsetneq\cdots\subsetneq I_k\subset\{1,\ldots,n+1\},\;J_\leq, J_\geq\subset\{1,\ldots,n+1\}\backslash I, \\ \prec\in\{\leq,=\}}\},
\end{equation}
where $J_\leq$ and $J_\geq$ are allowed to overlap. Clearly, $\mathcal{I}\subset\tilde{\mathcal{I}}$. We can still use our notation $\mathcal{D}^=$, $\mathcal{D}_j^J$ and $\mathcal{D}_j^I$ for this extended set of indices in the obvious manner. For every $i=(I_1,\ldots,I_k,J_\leq,J_\geq,\prec)\in\tilde{\mathcal{I}}$, we now define
\begin{equation}
    e_i:=\sum_{\substack{A\cup B=J_\leq\cap J_\geq \\ A\cap B=\emptyset}}(-1)^{|A|-1}\cdot e_{(I_1,\ldots,I_k,(J_\leq\backslash(J_\leq\cap J_\geq))\cup A,(J_\geq\backslash(J_\leq\cap J_\geq))\cup B,\prec)}.
\end{equation}
In other words, for overlapping $J_\leq$ and $J_\geq$, we sum over all partitions, distributing the intersection $J_\leq\cap J_\geq$ among $J_\leq$ and $J_\geq$ (with well-chosen signs). The content of the following lemma is that the boundary terms coming from $J_\leq\cap J_\geq$ cancel.
\begin{lem}\label{lem:combined_J}
    Let $i=(I_1,\ldots,I_k,J_\leq,J_\geq,\prec)\in\tilde{\mathcal{I}}$. Then we have
    \begin{equation}
        \partial e_i=\delta_{\prec,\leq}\cdot(-1)^{|J_i|+1}\cdot e_{\mathcal{D}^=(i)}+\sum_{j\in J_i\backslash (J_\leq\cap J_\geq)}(-1)^{n_i(j)-1}\cdot e_{\mathcal{D}_j^J(i)}+(-1)^{|J_i|+\delta_{\prec,\leq}}\cdot\sum_{j=1}^k(-1)^{j-1}\cdot e_{\mathcal{D}_j^I(i)}.
    \end{equation}
\end{lem}
\begin{proof}
    Let us denote $K:=J_\leq\cap J_\geq$ as well as $\tilde{J}_\leq:=J_\leq\backslash K$, $\tilde{J}_\geq:=J_\geq\backslash K$ and $\tilde{J}_i=J_i\backslash K$. Using equation \eqref{eq:boundary_simplex}, we compute
    \begin{equation}
        \begin{split}
            \partial e_i&=\sum_{\substack{A\cup B=K \\ A\cap B=\emptyset}}(-1)^{|A|-1}\cdot\partial e_{(I_1,\ldots,I_k,\tilde{J}_\geq\cup A,\tilde{J}_\geq\cup B,\prec)}\\
            &=\sum_{\substack{A\cup B=K \\ A\cap B=\emptyset}}(-1)^{|A|-1}\cdot\Bigl(\delta_{\prec,\leq}\cdot(-1)^{|J_i|+1}\cdot e_{(I_1,\ldots,I_k,\tilde{J}_\leq\cup A,\tilde{J}_\geq\cup B,\tau,=)}\\
            &+\sum_{j\in J_i}(-1)^{n_i(j)-1}\cdot e_{(I_1,\ldots,I_k,\tilde{J}_\leq\cup A-\{j\},\tilde{J}_\geq\cup B-\{j\},\prec)}+(-1)^{|J_i|+\delta_{\prec,\leq}}\cdot\sum_{j=1}^k(-1)^{j-1}\cdot e_{(I_1,\ldots,\widehat{I}_j,\ldots,I_k,\tilde{J}_\leq\cup A,\tilde{J}_\geq\cup B,\prec)}\Bigr).
        \end{split}
    \end{equation}
    For the first term, we immediately obtain
    \begin{equation}
        \begin{split}
            &\sum_{\substack{A\cup B=K \\ A\cap B=\emptyset}}(-1)^{|A|-1}\cdot\delta_{\prec,\leq}\cdot(-1)^{|J_i|+1}\cdot e_{(I_1,\ldots,I_k,\tilde{J}_\leq\cup A,\tilde{J}_\geq\cup B,\tau,=)}\\
            =&\delta_{\prec,\leq}\cdot(-1)^{|J_i|+1}\cdot\sum_{\substack{A\cup B=K \\ A\cap B=\emptyset}}(-1)^{|A|-1}e_{(I_1,\ldots,I_k,\tilde{J}_\leq\cup A,\tilde{J}_\geq\cup B,\tau,=)}=\delta_{\prec,\leq}\cdot(-1)^{|J_i|+1}\cdot e_{\mathcal{D}^=(i)}.
        \end{split}
    \end{equation}
    Similarly, for the last term, we obtain
    \begin{equation}
        \begin{split}
            &\sum_{\substack{A\cup B=K \\ A\cap B=\emptyset}}(-1)^{|A|-1}\cdot(-1)^{|J_i|+\delta_{\prec,\leq}}\cdot\sum_{j=1}^k(-1)^{j-1}\cdot e_{(I_1,\ldots,\widehat{I}_j,\ldots,I_k,\tilde{J}_\leq\cup A,\tilde{J}_\geq\cup B,\prec)}\\
            =&(-1)^{|J_i|+\delta_{\prec,\leq}}\cdot\sum_{j=1}^k(-1)^{j-1}\cdot\sum_{\substack{A\cup B=K \\ A\cap B=\emptyset}}(-1)^{|A|-1}\cdot e_{(I_1,\ldots,\widehat{I}_j,\ldots,I_k,\tilde{J}_\leq\cup A,\tilde{J}_\geq\cup B,\prec)}\\
            =&(-1)^{|J_i|+\delta_{\prec,\leq}}\cdot\sum_{j=1}^k(-1)^{j-1}\cdot e_{\mathcal{D}_j^I(i)}.
        \end{split}
    \end{equation}
    To deal with the middle term, let $j\in K$. Then for every $A\cup B=K$ with $A\cap B=\emptyset$ and $j\in A$, there are $A':=A\backslash\{j\}$ and $B':=B\cup\{j\}$ which satisfy $A'\cup B'=K$ and $A'\cap B'=\emptyset$. This means that we have $A-\{j\}=A'=A'-\{j\}$ and $B'-\{j\}=B=B-\{j\}$ so that
    \begin{equation}
        e_{(I_1,\ldots,I_k,\tilde{J}_\leq\cup A-\{j\},\tilde{J}_\geq\cup B-\{j\},\prec)}=e_{(I_1,\ldots,I_k,\tilde{J}_\leq\cup A'-\{j\},\tilde{J}_\geq\cup B'-\{j\},\prec)}.
    \end{equation}
    Furthermore, $|A|=|A'|+1$ which means $(-1)^{|A|-1}=-(-1)^{|A'|-1}$. Hence, the terms with $j\in K$ cancel pairwise and the middle term becomes
    \begin{equation}
        \begin{split}
            &\sum_{j\in\tilde{J}_i}(-1)^{n_i(j)-1}\cdot\sum_{\substack{A\cup B=K \\ A\cap B=\emptyset}}(-1)^{|A|-1}\cdot e_{(I_1,\ldots,I_k,\tilde{J}_\leq\cup A-\{j\},\tilde{J}_\geq\cup B-\{j\},\prec)}\\
            =&\sum_{j\in\tilde{J}_i}(-1)^{n_i(j)-1}\cdot e_{\mathcal{D}_j^J(i)}.
        \end{split}
    \end{equation}
    This completes the proof.
\end{proof}
In a final step, we form even larger groups of cells. For each pair of subsets $K_1,K_2\subset\{1,\ldots,n+1\}$, we define the index set
\begin{equation}
    \mathcal{I}_{K_1,K_2,\prec}:=\{i=(K_1,I_2,\ldots,I_k,\{1,\ldots,n+1\}\backslash K_2,\{1,\ldots,n+1\}\backslash I_k,\prec)\in\tilde{\mathcal{I}} \;|\; \substack{I_k\subset K_2, \\ \forall2\leq j\leq k:|I_{j-1}|=|I_j|-1}\}.
\end{equation}
Note that $\mathcal{I}_{K_1,K_2,\prec}$ can only be non-empty if $K_1\subset K_2$ (for $\prec\in\{\leq,=\}$) and if $K_1,K_2$ are non-empty. Clearly, an element $i=(I_1,\ldots,I_k,J_\leq,J_\geq,\prec)$ in $\mathcal{I}_{K_1,K_2,\prec}$ only depends on $I_1,\ldots,I_k$ (actually only on $I_2,\ldots,I_k$). Therefore, we simply write $i=(I_1,\ldots,I_k)$ in this case in a slight abuse of notation. For every $i=(I_1,\ldots,I_k)\in\mathcal{I}_{K_1,K_2,\prec}$, we have
\begin{equation}
    J_i=\{1,\ldots,n+1\}\backslash I_k \qquad\text{and}\qquad J_\leq\cap J_\geq=\{1,\ldots,n+1\}\backslash K_2.
\end{equation}
The plan is to form linear combinations of all cells $e_i$ with $i\in\mathcal{I}_{K_1,K_2,\prec}$. The tricky part is to figure out the correct sign for each cell so that we get a nice relation for the boundary of these linear combinations.
\begin{lem}\label{lem:sign}
    Let $K_1\subset K_2\subset\{1,\ldots,n+1\}$ and $\prec\in\{\leq,=\}$. There exists a map
    \begin{equation}
        \mathrm{sgn}:\mathcal{I}_{K_1,K_2,\prec}\to\{-1,1\}
    \end{equation}
    such that the following identities hold: Let $i=(I_1,\ldots,I_k)\in\mathcal{I}_{K_1,K_2,\prec}$. Then
    \begin{enumerate}
        \item If $\prec$ is $\leq$ and $i':=(I_1,\ldots,I_k,J_\leq,J_\geq,=)$, then $\mathrm{sgn}(i)=(-1)^{|J_i|+1}\cdot\mathrm{sgn}(i')$.
        \item For $2\leq j\leq k-1$, let $\{j_1\}:=I_j\backslash I_{j-1}$ and $\{j_2\}:= I_{j+1}\backslash I_j$. Then for
        \begin{equation}
            i':=(I_1,\ldots,I_{j-1},(I_j\cup\{j_2\})\backslash\{j_1\},I_{j+1},\ldots,I_k),
        \end{equation}
        we have $\mathrm{sgn}(i)=-\mathrm{sgn}(i')$.
        \item Let $j\in J_i$ and let $i':=(I_1,\ldots,I_k,I_k\cup\{j\})$. Then
        \begin{equation}
            \mathrm{sgn}(i')=(-1)^{|J_{i'}|+\delta_{\prec,\leq}+k-n_i(j)}\cdot\mathrm{sgn}(i).
        \end{equation}
        \item Suppose $k>1$. Let $\{j\}:=I_2\backslash I_1$, let $i'=(I_2,\ldots,I_k)\in\mathcal{I}_{K_1\cup\{j\},K_2}$ and let $N(j)$ be the number of elements in $\{1,\ldots,n+1\}\backslash K_1$ smaller than $j$. Then
        \begin{equation}
            \mathrm{sgn}(i)=(-1)^{|J_i|+N(j)-1}\cdot\mathrm{sgn}(i').
        \end{equation}
    \end{enumerate}
\end{lem}
\begin{proof}
    First, we assign each $i\in\mathcal{I}_{K_1,K_2,\prec}$ a permutation $\sigma_i\in S_{n+1}$. For $i=(I_1,\ldots,I_k)$, denote
    \begin{equation}
        \{1,\ldots,n+1\}\backslash(I_k\backslash K_1)=\{i_1,\ldots,i_{n-k+2}\}
    \end{equation}
    with $i_1<\cdots i_{n-k+2}$. For each $l\in\{1,\ldots,k-1\}$, the set $I_{l+1}\backslash I_l$ contains exactly one element which we denote by $j_l$. Then we set\footnote{The reader is reminded of the standards notation for elements of the permutation group: Let $\{i_1,\ldots,i_n\}=\{1,\ldots,n\}$. Then $\begin{pmatrix} 1 & \cdots & n \\ i_1 & \cdots & i_n \end{pmatrix}\in S_n$ denotes the bijection $\{1,\ldots,n\}\simeq\{1,\ldots,n\}$ which sends $j$ to $i_j$ for all $j\in\{1,\ldots,n\}$.}
    \begin{equation}
        \sigma_i:=
        \begin{pmatrix}
            1 & \cdots & n-k+2 & n-k+3 & \cdots & n+1 \\
            i_1 & \cdots & i_{n-k+2} & j_1 & \cdots & j_{k-1}
        \end{pmatrix}.
    \end{equation}
    Now we define
    \begin{equation}
        \mathrm{sgn}(i):=(-1)^{\delta_{\prec,\leq}(|J_i|+1)}\cdot(-1)^{(k-1)|K_1|}\cdot\mathrm{sgn}\,\sigma_i
    \end{equation}
    and verify that this satisfies (1) to (4).\\
    The first item follows immediately, since if $i$ and $i'$ only differ by the value of $\prec$, the signs $\mathrm{sgn}(i)$ of $i$ and $\mathrm{sgn}(i')$ of $i'$ differ only in the first factor, which is one in $\mathrm{sgn}(i')$, while it is $(-1)^{|J_i|+1}$ for $\mathrm{sgn}(i)$.\\
    To show (2), first note that $\mathrm{sgn}(i)$ and $\mathrm{sgn}(i')$ differ only by the sign of the permutations $\sigma_i$ and $\sigma_{i'}$. Denoting $m:=n-k+2$, we can express $\sigma_{i'}$ as
    \begin{equation}
        \sigma_{i'}:=
        \begin{pmatrix}
            1 & \cdots & m & m+1 & \cdots & m+l-1 & m+l & m+l+1 & m+l+2 & \cdots & n+1 \\
            i_1 & \cdots & i_m & j_1 & \cdots & j_{l-1} & j_{l+1} & j_l & j_{l+2} & \cdots & j_{k-1}
        \end{pmatrix}.
    \end{equation}
    Clearly, $\sigma_i$ and $\sigma_{i'}$ differ by one transposition (exchanging the $(m+l)$th and the $(m+l+1)$th position) so that
    \begin{equation}
        \mathrm{sgn}\,\sigma_i=-\mathrm{sgn}\,\sigma_{i'} \qquad\Rightarrow\qquad \mathrm{sgn}(i)=-\mathrm{sgn}(i').
    \end{equation}
    Now for (3), let $l:=n_i(j)$ which means $j=i_l$. We investigate the difference in all three factors of $\mathrm{sgn}(i)$ and $\mathrm{sgn}(i')$ separately. Clearly $|J_{i'}|=|J_i|-1$, so that \begin{equation}
        (-1)^{\delta_{\prec,\leq}(|J_{i'}|+1)}=(-1)^{\delta_{\prec,\leq}\cdot|J_i|}=-(-1)^{\delta_{\prec,\leq}}\cdot(-1)^{\delta_{\prec,\leq}\cdot(|J_i|-1)}.
    \end{equation}
    In words, the first factor of $\mathrm{sgn}(i)$ and $\mathrm{sgn}(i')$ is the same if $\prec$ is $=$ and opposite if $\prec$ is $\leq$. The second factor of $\mathrm{sgn}(i')$ reads
    \begin{equation}
        (-1)^{k|K_1|}=(-1)^{|K_1|}\cdot(-1)^{(k-1)|K_1|},
    \end{equation}
    so that it differs from the second factor of $\mathrm{sgn}(i)$ by a factor of $(-1)^{|K_1|}$. For the last factor, we can express the permutation $\sigma_{i'}$ as
    \begin{equation}
        \sigma_{i'}=
        \begin{pmatrix}
            1 & \cdots & l-1 & l & \cdots & n-k+1 & n-k+2 & \cdots & n & n+1 \\
            i_1 & \cdots & i_{l-1} & i_{l+1} & \cdots & i_{n-k+2} & j_1 & \cdots & j_{k-1} & i_l
        \end{pmatrix}.
    \end{equation}
    It is easy to see that $\sigma_{i'}$ differs from $\sigma_i$ by $k-1+(n-k+2-l)=n-l+1=n-n_i(j)+1$ transpositions. We use that
    \begin{equation}
        |J_{i'}|=n+1-|I_k\cup\{j\}|=n+1-(k+|K_1|)\quad\Leftrightarrow\quad n=|J_{i'}|+k+|K_1|-1,
    \end{equation}
    so that
    \begin{equation}
        \mathrm{sgn}\,\sigma_i=(-1)^{n-n_i(j)+1}\cdot\mathrm{sgn}\,\sigma_{i'}=(-1)^{|J_{i'}|+k+|K_1|-n_i(j)}\cdot\mathrm{sgn}\,\sigma_{i'}.
    \end{equation}
    Putting all the above results together, we obtain
    \begin{equation}
        \mathrm{sgn}(i')=(-1)^{\delta_{\prec,\leq}}\cdot(-1)^{|K_1|}\cdot(-1)^{|J_{i'}|+k+|K_1|-n_i(j)}\cdot\mathrm{sgn}(i)=(-1)^{|J_{i'}|+\delta_{\prec,\leq}+k-n_i(j)}\cdot\mathrm{sgn}(i)
    \end{equation}
    as claimed.\\
    Similarly to (3), we show (4). First note that $J_i=J_{i'}$. In particular, the first factor of the sign is identical for $i$ and $i'$. The second factor of $\mathrm{sgn}(i')$ reads
    \begin{equation}
        (-1)^{(k-2)|K_1|}=(-1)^{|K_1|}\cdot(-1)^{(k-1)|K_1|}
    \end{equation}
    so that it differs from the second factor of $\mathrm{sgn}(i)$ by $(-1)^{|K_1|}$. Now the permutation $\sigma_{i'}$ associated to $i'$ reads
    \begin{equation}
        \sigma_{i'}=
        \begin{pmatrix}
            1 & \cdots & N(j) & N(j)+1 & N(j)+2 & \cdots & n-k+2 & n-k+3 & \cdots & n+1 \\
            i_1 & \cdots & i_{N(j)} & j & i_{N(j)+1} & \cdots & i_{n-k+1} & j_2 & \cdots & j_{n+1}
        \end{pmatrix}.
    \end{equation}
    It is easy to see that $\sigma_{i'}$ differs from $\sigma_i$ by $n-k+2-(N(j)+1)=|J_i|+|K_1|-N(j)-1$ transpositions. Thus
    \begin{equation}
        \mathrm{sgn}(i')=(-1)^{|K_1|}\cdot(-1)^{|J_i|+|K_1|-N(j)-1}\cdot\mathrm{sgn}(i)=(-1)^{|J_i|-N(j)-1}\cdot\mathrm{sgn}(i).
    \end{equation}
\end{proof}
The map $\mathrm{sgn}$ from Lemma \ref{lem:sign} depends on $K_1$, $K_2$ and $\prec$. We drop this dependence from our notation since it is always clear from context which $K_1$, $K_2$ and $\prec$ are meant. With our signs in place, let us consider the linear combinations
\begin{equation}
    e_{K_1,K_2,\prec}:=\sum_{i\in\mathcal{I}_{K_1,K_2,\prec}}\mathrm{sgn}(i)\cdot e_i.
\end{equation}
Note that we can view $e_{K_1,K_2,\leq}$ as a chain in $\mathbb{C}^{n+1}$ relative to $S^{K_1,K_2}$ and $e_{K_1,K_2,=}$ as a chain in $S^{K_1,K_2}$. The signs $\mathrm{sgn}(i)$ are designed such that we obtain the following nice relation for the boundary of the $e_{K_1,K_2,\prec}$:
\begin{lem}\label{lem:cube_lemma}
    Let $K_1,K_2\subset\{1,\ldots,n+1\}$ and denote for all $j\in K_2\backslash K_1$ by $N(j)$ the number of elements in $\{1,\ldots,n+1\}\backslash K_1$ smaller than $j$. Then we have
    \begin{equation}
        \partial e_{K_1,K_2,\prec}=\delta_{\prec,\leq}\cdot e_{K_1,K_2,=}+\sum_{j\in K_2\backslash K_1}(-1)^{N(j)+\delta_{\prec,\leq}-1}\cdot e_{K_1\cup\{j\},K_2,\prec},
    \end{equation}
    where $N(j)$ is the number of elements in $\{1,\ldots,n+1\}\backslash K_1$ smaller than $j$.
\end{lem}
\begin{proof}
    Using Lemma \ref{lem:combined_J}, we compute
    \begin{equation}\label{eq:cube_lemma_1}
        \begin{split}
            \partial e_{K_1,K_2,\prec}&=\sum_{i\in\mathcal{I}_{K_1,K_2,\prec}}\mathrm{sgn}(i)\cdot\partial e_i\\
            &=\sum_{i\in\mathcal{I}_{K_1,K_2,\prec}}\mathrm{sgn}(i)\cdot\Bigl((-1)^{|J_i|+1}\cdot\delta_{\prec,\leq}\cdot e_{\mathcal{D}^=(i)}\\
            &+\sum_{j\in K_2\backslash I_k}(-1)^{n_i(j)-1}\cdot e_{\mathcal{D}_j^J(i)}+(-1)^{|J_i|+\delta_{\prec,\leq}}\cdot\sum_{j=1}^k(-1)^{j-1}\cdot e_{\mathcal{D}_j^I(i)}\Bigr).
        \end{split}
    \end{equation}
    If $\prec$ is $=$, then the first term vanishes due to the $\delta_{\prec,\leq}$ and is thus trivially identical to $\delta_{\prec,\leq}\cdot e_{K_1,K_2,=}$. If $\prec$ is $\leq$ on the other hand, we can use (1) of Lemma \ref{lem:sign} and write the first term as
    \begin{equation}
        \sum_{i\in\mathcal{I}_{K_1,K_2,\prec}}\mathrm{sgn}(i)\cdot(-1)^{|J_i|+1}\cdot e_{\mathcal{D}^=(i)}=\sum_{i\in\mathcal{I}_{K_1,K_2,=}}\mathrm{sgn}(i)\cdot e_i=e_{K_1,K_2,=}.
    \end{equation}
    Now we attend to the remaining terms. The outer sum $\sum_{i\in\mathcal{I}_{K_1,K_2,\prec}}$ of \eqref{eq:cube_lemma_1} can be written as
    \begin{equation}
        \sum_{i\in\mathcal{I}_{K_1,K_2,\prec}}=\sum_{k=1}^{|K_2|-|K_1|+1}\sum_{\substack{i\in\mathcal{I}_{K_1,K_2,\prec} \\ |i|=k}}.
    \end{equation}
    We show that for every $1\leq k\leq|K_2|-|K_1|+1$ and $2\leq j\leq k-1$, we have
    \begin{equation}
        \sum_{\substack{i\in\mathcal{I}_{K_1,K_2,\prec} \\ |i|=k}}\mathrm{sgn}(i)\cdot e_{\mathcal{D}_j^I(i)}=0,
    \end{equation}
    since every term arises exactly twice with opposite signs (aside from the marginal case $k=1$ in which the expression trivially vanishes). Indeed, for every fixed $I_1,\ldots,I_{j-1},I_{j+1},\ldots,I_k$ the set $I_{j+1}\backslash I_{j-1}$ contains exactly two elements. Call these elements $j_1$ and $j_2$. Then there are exactly two terms with $e_{(I_1,\ldots,\widehat{I}_j,\ldots,I_k)}$ in the sum, which appear for
    \begin{equation}
        i=(I_1,\ldots,I_{j-1},I_{j-1}\cup\{j_1\},I_{j+1},\ldots,I_k)\qquad\text{and}\qquad i'=(I_1,\ldots,I_{j-1},I_{j-1}\cup\{j_2\},I_{j+1},\ldots,I_k).
    \end{equation}
    According to Lemma \ref{lem:sign}, (2), we have $\mathrm{sgn}(i)=-\mathrm{sgn}(i')$. In conclusion, only the outer terms of the sum, namely those for $j=1$ and $j=k$, remain and we have
    \begin{equation}
        \begin{split}
            &\sum_{i\in\mathcal{I}_{K_1,K_2,\prec}}\mathrm{sgn}(i)\cdot(-1)^{|J_i|+\delta_{\prec,\leq}}\cdot\sum_{j=1}^k(-1)^{j-1}\cdot e_{\mathcal{D}_j^I(i)}\\
            =&\sum_{k=2}^{|K_2|-|K_1|+1}\sum_{\substack{i\in\mathcal{I}_{K_1,K_2,\prec} \\ |i|=k}}\mathrm{sgn}(i)\cdot(-1)^{|J_i|+\delta_{\prec,\leq}}\cdot\Bigl(e_{\mathcal{D}_1^I(i)}+(-1)^{k-1}\cdot e_{\mathcal{D}_k^I(i)}\Bigr).
        \end{split}
    \end{equation}
    Here, the sum starts at $k=2$ since the terms vanish for $k=1$. Thus, we can write
    \begin{equation}
        \begin{split}
            \partial e_{K_1,K_2,\prec}&=\delta_{\prec,\leq}\cdot e_{K_1,K_2,\prec}+\sum_{k=1}^{|K_2|-|K_1|}\Bigl(\sum_{\substack{i\in\mathcal{I}_{K_1,K_2,\prec} \\ |i|=k}}\mathrm{sgn}(i)\cdot\sum_{j\in K_2\backslash I_k}(-1)^{n_i(j)-1}e_{\mathcal{D}_j^J(i)}\\
            &+\sum_{\substack{i\in\mathcal{I}_{K_1,K_2,\prec} \\ |i|=k+1}}(-1)^{|J_i|+\delta_{\prec,\leq}+k}\cdot\mathrm{sgn}(i)\cdot e_{\mathcal{D}_{k+1}^I(I_1,\ldots,I_{k+1})}\Bigr)\\
            &+\sum_{k=2}^{|K_2|-|K_1|+1}\sum_{\substack{i\in\mathcal{I}_{K_1,K_2,\prec} \\ |i|=k}}\mathrm{sgn}(i)\cdot(-1)^{|J_i|+\delta_{\prec,\leq}}\cdot e_{\mathcal{D}_1^I(I_1,\ldots,I_k)}.
        \end{split}
    \end{equation}
    We claim that the two sums in the big brackets cancel each other. Let $1\leq k\leq|K_2|-|K_1|$ and consider the two sets
    \begin{equation}
        A:=\{(i,j)\in\mathcal{I}_{K_1,K_2,=}\times K_2\backslash I_k \;|\; |i|=k\}
    \end{equation}
    and
    \begin{equation}
        B:=\{i\in\mathcal{I}_{K_1,K_2,=} \;|\; |i|=k+1\}.
    \end{equation}
    We define a map
    \begin{equation}
        f:A\to B,\qquad (i,j)\mapsto(I_1,\ldots,I_k,I_k\cup\{j\}).
    \end{equation}
    This is a bijection with inverse
    \begin{equation}
        f^{-1}:B\to A,\qquad i\mapsto((I_1,\ldots,I_k),j(i)),
    \end{equation}
    where $j(i)$ is the unique element in $I_{k+1}\backslash I_k$. Note that 
    \begin{equation}
        \mathcal{D}_{k+1}^I(f(i,j))=\mathcal{D}_j^J(i)
    \end{equation}

    $f((I_1,\ldots,I_k),j)=\mathcal{D}_j^J(I_1,\ldots,I_k,I_k\cup\{j\})$. By Lemma \ref{lem:sign}, we have
    \begin{equation}
        \mathrm{sgn}(f(i,j))=(-1)^{|J_i|+\delta_{\prec,\leq}+k-n_i(j)}\cdot\mathrm{sgn}(i).
    \end{equation}
    Thus, we obtain
    \begin{equation}
        \begin{split}
            &\sum_{(i,j)\in A}\mathrm{sgn}(i)\cdot(-1)^{n_i(j)-1}e_{\mathcal{D}_j^J(i)}+\sum_{i\in B}(-1)^{|J_i|+\delta_{\prec,\leq}+k}\cdot\mathrm{sgn}(i)\cdot e_{\mathcal{D}_{k+1}^I(I_1,\ldots,I_{k+1})}\\
            =&\sum_{(i,j)\in A}\mathrm{sgn}(i)\cdot((-1)^{n_i(j)-1}+(-1)^{n_i(j)})\cdot e_{\mathcal{D}_j^J(i)}=0
        \end{split}
    \end{equation}
    and the remaining terms read
    \begin{equation}
        \begin{split}
            \partial e_{K_1,K_2,\prec}&=\sum_{k=2}^{|K_2|-|K_1|+1}\sum_{\substack{i\in\mathcal{I}_{K_1,K_2,\prec} \\ |i|=k}}\mathrm{sgn}(i)\cdot(-1)^{|J_i|+\delta_{\prec,\leq}}\cdot e_{\mathcal{D}_1^I(I_1,\ldots,I_k)} \\
            &=\sum_{\substack{i\in\mathcal{I}_{K_1,K_2,=} \\ |i|\geq2}}\mathrm{sgn}(i)\cdot(-1)^{|J_i|+\delta_{\prec,\leq}}\cdot e_{\mathcal{D}_1^I(I_1,\ldots,I_k)}.
        \end{split}
    \end{equation}
    Similarly to the cancellation argument above, denote
    \begin{equation}
        A:=\{i\in\mathcal{I}_{K_1,K_2,=} \;|\; |i|\geq2\}
    \end{equation}
    and
    \begin{equation}
        B:=\bigcup_{j\in K_2\backslash K_1}\mathcal{I}_{K_1\cup\{j\},K_2,=}.
    \end{equation}
    Then
    \begin{equation}
        f:B\to A,\qquad i\mapsto(K_1,I_1,\ldots,I_k)
    \end{equation}
    is a bijection with inverse
    \begin{equation}
        f^{-1}:A\to B,\qquad i\mapsto (I_2,\ldots,I_k).
    \end{equation}
    Note that
    \begin{equation}
        \mathcal{D}_1^I(f(i))=i.
    \end{equation}
    By Lemma \ref{lem:sign}, (4), we have
    \begin{equation}
        \mathrm{sgn}(f(i))=(-1)^{|J_i|+N(j)-1}\cdot\,\mathrm{sgn}(i).
    \end{equation}
    Hence,
    \begin{equation}
        \begin{split}
            \partial e_{K_1,K_2,\prec}&=\sum_{j\in K_2\backslash K_1}\sum_{i\in\mathcal{I}_{K_1\cup\{j\},K_2,\prec}}\mathrm{sgn}(i)\cdot(-1)^{N(j)+\delta_{\prec,\leq}-1}\cdot e_i\\
            &=\sum_{j\in K_2\backslash K_1}(-1)^{N(j)+\delta_{\prec,\leq}-1}\cdot e_{K_1\cup\{j\},K_2,\prec}.
        \end{split}
    \end{equation}
\end{proof}
We are finally in a position to compute the generators of $H_{n+1}(\mathbb{C}^{n+1},\bigcup_{i=1}^{n+1}S_i)$. First, we express the obvious generators of the groups $\widetilde{H}_{n+1-|I|}(\bigcap_{i\in I}S_i)$ in terms of our CW-decomposition:
\begin{lem}\label{lem:generators}
    Let $K\subset\{1,\ldots,n+1\}$ be a non-empty subset. Then the group $\widetilde{H}_{n+1-|K|}(\bigcap_{k\in K}S_k)$ is generated by the homology class of $e_{K,K,=}$.
\end{lem}
\begin{proof}
    First, note that $\mathcal{I}_{K,K,=}$ contains precisely one element, namely
    \begin{equation}
        i:=(K,\{1,\ldots,n+1\}\backslash K,\{1,\ldots,n+1\}\backslash K,=).
    \end{equation}
    Thus
    \begin{equation}
        e_{K,K,=}=e_i=\sum_{\substack{A\cup B=\{1,\ldots,n+1\}\backslash K \\ A\cap B=\emptyset}}(-1)^{|A|-1}\cdot e_{(K,A,B,=)}.
    \end{equation}
    Let $j_0:=\min K$. Then the support of $e_{K,K,=}$ is
    \begin{equation}
        \{z\in\mathbb{C}^{n+1} \;|\; i\cdot\text{Im}\,z=c_I,\;\text{Re}\,z\in V_I,\; (z-a_{j_0})^2=r_{j_0}^2\}=(c_I+V_I)\cap S_{j_0},
    \end{equation}
    which is exactly the real sphere to which $\bigcap_{i\in K}S_i$ deformation retracts (see equation \eqref{eq:real_sphere} and below). Thus, it suffices to check that the signs are chosen correctly to yield a cycle in $\bigcap_{i\in K}S_i$. But this follows directly from Lemma \ref{lem:cube_lemma}.
\end{proof}
Now all necessary generators can be obtained from an easy inductive argument. Again, we need to figure out the right signs. Let $I=\{i_1,\ldots,i_m\}\subset\{1,\ldots,n+1\}$ with $|I|=m$. For all $0\leq j\leq m-1$ and $m-j+1\leq k\leq m$, let
\begin{equation}
    \tau_k^{m,j}:=(-1)^{i_k+\sum_{l=m-j}^mi_l-jm+\frac{j}{2}(j+1)}.
\end{equation}
In particular, we have $\tau_m^{m,0}=1$ and
\begin{equation}
    \tau_{m-j}^{m,j}=(-1)^{\sum_{l=m-j+1}^mi_l-jm+\frac{j}{2}(j+1)}.
\end{equation}
The $\tau_k^{m,j}$ clearly depend on the set $I$ but we drop this dependence from our notation as long as no confusion can arise. The $\tau_k^{m,j}$ satisfy the following relations:
\begin{lem}\label{lem:sign_2}
    For all $0\leq j\leq m-1$ and all $m-j+1\leq k,l\leq m$, we have the following identities:
    \begin{enumerate}
        \item $(-1)^{i_k-m+j}\cdot\tau^{m,j}_{m-j}=\tau_k^{m,j-1}$.
        \item $(-1)^{i_{m-j}-m+j}\cdot\tau^{m,j}_k=\tau^{m,j-1}_k$.
        \item $(-1)^{i_k}\cdot\tau_l^{m,j}=(-1)^{i_l}\cdot\tau_k^{m,j}$.
    \end{enumerate}
\end{lem}
\begin{proof}
    \begin{enumerate}
        \item The exponent on the left-hand side reads
        \begin{equation}
            i_k-m+j+\sum_{l=m-j+1}^mi_l-jm+\frac{j}{2}(j+1)=i_k+\sum_{l=m-j+1}^mi_l-(j+1)m+2j+\frac{j-1}{2}j.
        \end{equation}
        Modulo 2, this is
        \begin{equation}\label{eq:sign_2_1}
            i_k+\sum_{l=m-j+1}^mi_l-(j-1)m+\frac{j-1}{2}j
        \end{equation}
        which is precisely the exponent of $\tau_k^{m,j-1}$.
        \item The exponent on the left-hand side reads
        \begin{equation}
            i_{m-j}-m+j+i_k+\sum_{l=m-j}^mi_l-jm+\frac{j}{2}(j+1)=i_k+\sum_{l=m-j+1}^mi_l-(j+1)m+\frac{j-1}{2}j+2i_{m-j}+2j,
        \end{equation}
        which again is equal to the exponent \eqref{eq:sign_2_1} of $\tau_k^{m,j-1}$ modulo 2.
        \item Straightforward.
    \end{enumerate}
\end{proof}

\begin{thm}\label{thm:generators}
    Let $I=\{i_1,\ldots,i_m\}\subset\{1,\ldots,n+2\}$ be non-empty with $i_1<\cdots<i_m$ and let $0\leq j\leq m-1$. Then the group
    \begin{equation}
        \delta_{i_1}^{-1}\cdots\delta_{i_m}^{-1}\widetilde{H}_{n+1-m}(\bigcap_{i\in I}S_i)
    \end{equation}
    is generated by the homology class of
    \begin{equation}
        e_{I,m-j}:=\sum_{k=m-j}^m\tau_k^{m,j}\cdot e_{\{i_1,\ldots,i_{m-j-1},i_k\},I,=}.
    \end{equation}
\end{thm}
\begin{proof}
    We conduct the proof by induction on $j$. For $j=0$, we have
    \begin{equation}
        e_{I,m}=\tau_m^{m,0}\cdot e_{I,I,=}=e_{I,I,=},
    \end{equation}
    which generates the group $\widetilde{H}_{n+1-m}(\bigcap_{i\in I}S_i)$ according to Lemma \ref{lem:generators}. For the induction step, let $j>0$. We write $e_{I,m-j}=u_{I,m-j}+v_{I,m-j}$ with
    \begin{equation}
        u_{I,m-j}:=\tau_{m-j}^{m,j}\cdot e_{\{i_1,\ldots,i_{m-j}\},I,=} \;\in\; C_{n+1-m+j}(S^{\{i_1,\ldots,i_{m-j-1}\},\{i_{m-j}\}})
    \end{equation}
    and
    \begin{equation}
        v_{I,m-j}:=\sum_{k=m-j+1}^m\tau_k^{m,j}\cdot e_{\{i_1,\ldots,i_{m-j-1},i_k\},I,=} \;\in\; C_{n+1-m+j}(S^{\{i_1,\ldots,i_{m-j-1}\},\{i_{m-j+1},\ldots,i_m\}}).
    \end{equation}
    We first compute the boundary of $u_{I,m-j}$. We have
    \begin{equation}
        \begin{split}
            \partial u_{I,m-j}&=\tau_{m-j}^{m,j}\cdot\partial e_{\{i_1,\ldots,i_{m-j}\},I,=}\\
            &=\tau_{m-j}^{m,j}\cdot\sum_{k\in\{i_{m-j+1},\ldots,i_m\}}(-1)^{N(k)-1}\cdot e_{\{i_1,\ldots,i_{m-j},k\},I,=}\\
            &=\sum_{k=m-j+1}^m(-1)^{i_k-m+j}\cdot\tau_{m-j}^{m,j}\cdot e_{\{i_1,\ldots,i_{m-j},i_k\},I,=}\\
            &=\sum_{k=m-j+1}^m\tau_k^{m,j-1}\cdot e_{\{i_1,\ldots,i_{m-j},i_k\},I,=}=e_{m-(j-1)},
        \end{split}
    \end{equation}
    where we used Lemma \ref{lem:cube_lemma} in the second step, the fact that $N(i_k)=i_k-m+j-1$\footnote{$N(i_k)$ is the number of elements in $\{1,\ldots,n+1\}\backslash\{i_1,\ldots,i_{m-j}\}$ smaller than $i_k$. Since $i_1,\ldots,i_{m-j}<i_k$, this is precisely $i_k-1$ minus the number $m-j$ of elements in $\{i_1,\ldots,i_{m-j}\}$.} in the third step and Lemma \ref{lem:sign_2}, (1) in the last step. The boundary of $v_{I,m-j}$ can be computed as
    \begin{equation}
        \begin{split}
            \partial v_{I,m-j}&=\sum_{k=m-j+1}^m\tau_k^{m,j}\cdot\partial e_{\{i_1,\ldots,i_{m-j-1},i_k\},I,=}\\
            &=\sum_{k=m-j+1}^m\tau_k^{m,j}\cdot\sum_{\substack{l=m-j \\ l\neq k}}^m(-1)^{N(i_l)-1}\cdot e_{\{i_1,\ldots,i_{m-j-1},i_k,i_l\},I,=}\\
            &=\sum_{k=m-j+1}^m\tau_k^{m,j}\cdot(-1)^{i_{m-j}-m+j+1}\cdot e_{\{i_1,\ldots,i_{m-j},i_k\},I,=}\\
            &+\sum_{\substack{k,l=m-j+1 \\ l>k}}^m(-1)^{i_l-m+j}\cdot\tau_k^{m,j}\cdot e_{\{i_1,\ldots,i_{m-j-1},i_k,i_l\},I,=}\\
            &-\sum_{\substack{k,l=m-j+1 \\ l<k}}^m(-1)^{i_l-m+j}\cdot\tau_k^{m,j}\cdot e_{\{i_1,\ldots,i_{m-j-1},i_k,i_l\},I,=},
        \end{split}
    \end{equation}
    where again we used Lemma \ref{lem:cube_lemma} in the second step. The two sums with $l<k$ and $l>k$ are identical and hence cancel each other:
    \begin{equation}
        \begin{split}
            &\sum_{\substack{k,l=m-j+1 \\ l>k}}^m(-1)^{i_l-m+j}\cdot\tau_k^{m,j}\cdot e_{\{i_1,\ldots,i_{m-j-1},i_k,i_l\},I,=}\\
            \overset{\phantom{\text{Lemma }\ref{lem:sign_2}\text{ (3)}}}=&\sum_{\substack{k,l=m-j+1 \\ l<k}}^m(-1)^{i_k-m+j}\cdot\tau_l^{m,j}\cdot e_{\{i_1,\ldots,i_{m-j-1},i_k,i_l\},I,=}\\
            \overset{\text{Lemma }\ref{lem:sign_2}\text{ (3)}}=&\sum_{\substack{k,l=m-j+1 \\ l<k}}^m(-1)^{i_l-m+j}\cdot\tau_k^{m,j}\cdot e_{\{i_1,\ldots,i_{m-j-1},i_k,i_l\},I,=}
        \end{split}
    \end{equation}
    Thus, we obtain
    \begin{equation}
        \begin{split}
            \partial v_{I,m-j}&=-\sum_{k=m-j+1}^m\tau_k^{m,j}\cdot(-1)^{i_{m-j}-m+j}\cdot e_{\{i_1,\ldots,i_{m-j},i_k\},I,=}\\
            &=-\sum_{k=m-j+1}^m\tau_k^{m,j-1}\cdot e_{\{i_1,\ldots,i_{m-j},i_k\},I,=}=-e_{I,m-(j-1)}=-\partial u_{I,m-j}
        \end{split}
    \end{equation}
    by Lemma \ref{lem:sign_2}, (2). This shows on the one hand that $\partial e_{I,m-j}=\partial u_{I,m-j}+\partial v_{I,m-j}=0$ and on the other hand that the Mayer-Vietoris homomorphism $\delta_{I,m-j}$ sends $e_{I,m-j}$ to $e_{I,m-(j-1)}$. So $e_{I,m-j}$ does indeed define a homology class which generates the group of interest.
\end{proof}
The boundary homomorphism $\partial_\ast$ is easy to understand and as an immediate consequence, we get the following
\begin{cor}\label{cor:generators}
    In the same situation as in Theorem \ref{thm:generators}, the group
    \begin{equation}
        (\partial_\ast)^{-1}\delta_{i_1}^{-1}\cdots\delta_{i_m}^{-1}\widetilde{H}_{n+1-m}(\bigcap_{i\in I}S_i)
    \end{equation}
    is generated by the homology class of
    \begin{equation}
        \mathbf{e}_{I,m-j}:=\sum_{k=m-j}^m\tau_{m-j}^{m,j}\cdot e_{\{i_1,\ldots,i_{m-j-1},i_k\},I,\leq}.
    \end{equation}
\end{cor}
\begin{proof}
    Similarly to the proof of Theorem \ref{thm:generators}, we write $\mathbf{e}_{I,m-j}=\mathbf{u}_{I,m-j}+\mathbf{v}_{I,m-j}$ with
    \begin{equation}
        \mathbf{u}_{I,m-j}:=\tau_{m-j}^{m,j}\cdot e_{\{i_1,\ldots,i_{m-j}\},I,\leq}\;\in\; C_{n+1-m+j}(\mathbb{C}^{n+1},S^{\{i_1,\ldots,i_{m-j-1}\},\{i_{m-j}\}})
    \end{equation}
    and
    \begin{equation}
        \mathbf{v}_{I,m-j}:=\sum_{k=m-j+1}^m\tau_k^{m,j}\cdot e_{\{i_1,\ldots,i_{m-j-1},i_k\},I,\leq}\;\in\; C_{n+1-m+j}(\mathbb{C}^{n+1},S^{\{i_1,\ldots,i_{m-j-1}\},\{i_{m-j+1,\ldots,i_m}\}}).
    \end{equation}
    The same computations as before show that the only part of the boundary $\partial\mathbf{e}_{I,m-j}$ of $\mathbf{e}_{I,m-j}$ that does not vanish is
    \begin{equation}
        \partial\mathbf{e}_{I,m-j}=\sum_{k=m-j}^m\tau_k^{m,j}\cdot e_{\{i_1,\ldots,i_{m-j-1},i_k\},I,=}=e_{I,m-j}.
    \end{equation}
    Thus, the induced map $\partial_\ast$ sends $\mathbf{e}_{I,m-j}$ to the generator $e_{I,m-j}$ of $\delta_{i_1}^{-1}\cdots\delta_{i_m}^{-1}\widetilde{H}_n(\bigcap_{i\in I}S_i)$, so that $\mathbf{e}_{I,m-j}$ indeed generates the group in question.
\end{proof}
For every non-empty $I\subset\{1,\ldots,n+1\}$, we denote $e_I:=(-1)^{\min I-1}\cdot e_{I,1}$ and $\mathbf{e}_I:=(-1)^{\min I-1}\cdot \mathbf{e}_{I,1}$. If $I=\{i_1,\ldots,i_m\}$ with $i_1<\cdots<i_m$, these read
\begin{equation}
    e_I:=(-1)^{i_1-1}\cdot\sum_{k=1}^m\tau_k^{m,m-1}\cdot e_{\{i_k\},I,=} \quad\text{and}\quad \mathbf{e}_I:=(-1)^{i_1-1}\cdot\sum_{k=1}^m\tau_k^{m,m-1}\cdot e_{\{i_k\},I,\leq}.
\end{equation}
Together with our direct sum decomposition for the homology groups from Section \ref{sec:complex_spheres}, we conclude the following:
\begin{thm}
    The homology groups
    \begin{equation}
        \widetilde{H}_n(\bigcup_{i=1}^{n+1}S_i) \qquad\text{and}\qquad H_{n+1}(\mathbb{C}^{n+1},\bigcup_{i=1}^{n+1}S_i)
    \end{equation}
    are generated by $\{[e_I]\}_{\substack{I\subset\{1,\ldots,n+1\} \\ I\neq\emptyset}}$ and $\{[\mathbf{e}_I]\}_{\substack{I\subset\{1,\ldots,n+1\} \\ I\neq\emptyset}}$ respectively.
\end{thm}
\begin{proof}
    According to Theorem \ref{thm:decomposition_homology}, the relevant homology groups decompose into a direct sum of terms, whose generators are given precisely as claimed according to Theorem \ref{thm:generators} and Corollary \ref{cor:generators}.
\end{proof}
This concludes the task of finding explicit representatives for the generators. With these concrete representatives for the generators of $H_{n+1}(\mathbb{C}^{n+1},\bigcup_{i=1}^{n+1}S_i)$ in place, it is now possible to compute the intersection index of all generators with $(i\cdot\mathbb{R})^{n+1}$. But this task becomes even easier after slightly altering the representatives $\mathbf{e}_I$ of the generators such that they intersect $(i\cdot\mathbb{R})^{n+1}$ in exactly one point or not at all (see Example \ref{ex:cw_1}).
\begin{lem}\label{lem:deformation_of_generators}
    Let $I\subset\{1,\ldots,n+1\}$ with $|I|\geq 2$. Then $[\mathbf{e}_I]$ has a representative $\mathbf{e}_I'$ with support disjoint from $(i\cdot\mathbb{R})^{n+1}$.
\end{lem}
\begin{proof}
    First note that $\bigcup_{i=1}^{n+1}S_i$ is disjoint from $(i\cdot\mathbb{R})^{n+1}$. More concretely, the distance between the set $\bigcup_{i=1}^{n+1}S_i$ and $(i\cdot\mathbb{R})^{n+1}$ is $1$. Choose an open neighborhood $U$ of
    \begin{equation}
        X:=(i\cdot\mathbb{R})^{n+1}\cap\bigcup_{\tau\in\{-1,0,1\}}\bigcup_{i\in\mathcal{I}}e_i^\tau
    \end{equation}
    disjoint from $\bigcup_{i=1}^{n+1}S_i$ and let $K\subset U$ be a compact neighborhood of $X$. Choose some $i_0\in I$ and consider the basis $v,\{v_j\}_{j\neq i_0}$ of $\mathbb{R}^{n+1}$. Let
    \begin{equation}
        \tilde{g}:\mathbb{R}^{n+1}\to\mathbb{R}^{n+1},\qquad d\cdot v+\sum_{\substack{j=1 \\ j\neq i_0}}^{n+1}d_j\cdot v_j\mapsto (d,d_1,\ldots,\widehat{d}_{i_0},\ldots,d_{n+1})
    \end{equation}
    be the change of basis from $v,\{v_j\}_{j\neq i_0}$ to the canonical basis and extend it to $\mathbb{C}^{n+1}$ by
    \begin{equation}
        g:\mathbb{C}^{n+1}\to\mathbb{C}^{n+1},\qquad z=x+i\cdot y\mapsto\tilde{g}(x)+i\cdot y.
    \end{equation}
    Now let $V$ be the constant unit vector field in the direction
    \begin{equation}
        (1,\ldots,1)+i\cdot(0,\ldots,0).
    \end{equation}
    Let $\rho:\mathbb{C}^{n+1}\to\mathbb{R}_{\geq0}$ be a bump-function with support in $g(U)$ such that $\rho|_{g(K)}=1$. Then integrating the Lipschitzian vector field $\rho\cdot V$ yields an ambient isotopy $\varphi':\mathbb{C}^{n+1}\times[0,1]\to\mathbb{C}^{n+1}$ of $\bigcup_{i=1}^{n+1}g(S_i)$ in $\mathbb{C}^{n+1}$. We use $\varphi'$ to define the ambient isotopy $\varphi:=g^{-1}\circ\varphi'$ of $\bigcup_{i=1}^{n+1}S_i$ in $\mathbb{C}^{n+1}$. Since $|I|\geq2$, there must be at least one $j\in\{1,\ldots,n+1\}\backslash\{i_0\}$ such that $(\tilde{g}(\text{Re}\,z))_j\geq0$ for all $z\in\mathbf{e}_I$. But then $\varphi(\cdot,1)$ takes all such $z$ to $z'$ with $(\tilde{g}(\text{Re}\,z'))_j>0$ and in particular $z'\notin(i\cdot\mathbb{R})^{n+1}$. Thus $\mathbf{e}_I':=\varphi(\mathbf{e}_I,1)$ fulfills the requirement of the lemma.
\end{proof}
\begin{thm}
    We have the following results for the intersection indices:\footnote{Here, it is understood implicitly that $(i\cdot\mathbb{R})^{n+1}$ inherits its orientation from the standard orientation of $\mathbb{C}^{n+1}$.}
    \begin{equation}
        \langle(i\cdot\mathbb{R})^{n+1}|\mathbf{e}_I\rangle=\begin{cases} 1 & \mathrm{if }\;|I|=1 \\ 0 & \mathrm{if }\;|I|\geq2\end{cases}
    \end{equation}
\end{thm}
\begin{proof}
    From Lemma \ref{lem:deformation_of_generators}, we immediately obtain $\langle(i\cdot\mathbb{R})^{n+1}|\mathbf{e}_I\rangle=0$ for all $I\subset \{1,\ldots,n+1\}$ with $|I|\geq2$. If $|I|=1$ on the other hand, let us write $I=\{k\}$. Then the support of $\mathbf{e}_I$ is contained in $a_k+\mathbb{R}^{n+1}¸$ and contains the point $a_k$. Hence $\mathbf{e}_I$ and $(i\cdot\mathbb{R})^{n+1}$ only intersect at $a_k$ and according to Proposition \ref{prop:intersection_index_disjoint_support}, it suffices to check that the orientation of $\mathbf{e}_I$ and $(i\cdot\mathbb{R})^{n+1}$ match at this point. We have $\tau_1^{1,0}=1$ and thus
    \begin{equation}
        \mathbf{e}_I=(-1)^{k-1}\cdot e_{\{k\},\{k\},\leq}.
    \end{equation}
    We identify the tangent space of $\mathbb{C}^{n+1}$ (as a real $(2n+2)$-manifold) at any point with $\mathbb{C}^{n+1}$ itself and orient the space such that
    \begin{equation}\label{eq:complex_basis}
        \begin{pmatrix} 1 \\ 0 \\ \vdots \\ 0 \end{pmatrix},\quad\ldots\quad,\quad \begin{pmatrix} 0 \\ \vdots \\ 0 \\ 1 \end{pmatrix},\quad \begin{pmatrix} i \\ 0 \\ \vdots \\ 0 \end{pmatrix},\quad\ldots\quad,\quad \begin{pmatrix} 0 \\ \vdots \\ 0 \\ i \end{pmatrix}
    \end{equation}
    forms a positive basis at any point of the tangent space $\mathbb{C}^{n+1}$ (viewed as a real vector space). We oriented $e_{\{k\},\{k\},\leq}$ such that
    \begin{equation}
        v_1,\;\ldots,\;\widehat{v_k},\;\ldots,v_{n+1},\;v
    \end{equation}
    is a positive basis of the tangent space of $e_{\{k\},\{k\},\leq}$ at 0. The change of basis matrix $M_k$ from the canonical basis to this basis has determinant
    \begin{equation}
        \det
        \begin{pmatrix}
            0 & 1 & \cdots & \widehat{1} & \cdots & 1 & 1 \\
            1 & 0 & \cdots & \widehat{1} & \cdots & 1 & 1 \\
            1 & 1 & \cdots & \widehat{1} & \cdots & 1 & 1 \\
            \vdots & \vdots & \ddots & \vdots & \ddots & \vdots & \vdots \\
            1 & 1 & \cdots & \widehat{0} & \cdots & 1 & 1 \\
            \vdots & \vdots & \ddots & \vdots & \ddots & \vdots & \vdots \\
            1 & 1 & \cdots & \widehat{1} & \cdots & 0 & 1
        \end{pmatrix}=(-1)^{k-1},
    \end{equation}
    where the hat indicates that the $k$th column is removed from the matrix. For $(i\cdot\mathbb{R})^{n+1}$, the vectors
    \begin{equation}
        \begin{pmatrix} i \\ 0 \\ \vdots \\ 0 \end{pmatrix},\quad\ldots\quad,\quad \begin{pmatrix} 0 \\ \vdots \\ 0 \\ i \end{pmatrix}
    \end{equation}
    form a positive basis of the tangent space at 0. The resulting combined basis is taken to the basis \eqref{eq:complex_basis} by the matrix
    \begin{equation}
        \begin{pmatrix} M_k^{-1} & 0 \\ 0 & 1_{n+1}\end{pmatrix},
    \end{equation}
    which has determinant $\det M_k^{-1}\cdot\det 1_{n+1}=(\det M_k)^{-1}=(-1)^{k-1}$. Hence the intersection index of $\mathbf{e}_I$ with $(i\cdot\mathbb{R})^{n+1}$ is
    \begin{equation}
        \langle(i\cdot\mathbb{R})^{n+1}|\mathbf{e}_I\rangle=(-1)^{k-1}\cdot\langle (i\cdot\mathbb{R})^{n+1}|e_{\{k\},\{k\},\leq}\rangle=((-1)^{k-1})^2=1.
    \end{equation}
\end{proof}
This concludes the main part of this article. As a small extra, we also compute the intersection index of the vanishing spheres that can occur with the generators $\mathbf{e}_I$. These indices agree with the known intersection indices of the vanishing spheres with the vanishing cells. To this end, we need the following
\begin{lem}\label{lem:iterated_boundary}
    Let $I=\{i_1,\ldots,i_m\}\subset\{1,\ldots,n+1\}$ and denote by $\partial_{i_j}$ the homomorphism that takes the boundary in $S_{i_j}$ as in equation \eqref{eq:iterated_boundary}. Then
    \begin{equation}
        (\partial_{i_m}\circ\cdots\circ\partial_{i_1})(\mathbf{e}_I)=e_{I,|I|}=e_{I,I,=}.
    \end{equation}
\end{lem}
\begin{proof}
    The second equality is just the definition of $e_{I,|I|}$ and it suffices to show the first one. From Corollary \ref{cor:generators}, we know that
    \begin{equation}
        \partial\mathbf{e}_I=e_I=\sum_{k=1}^m\tau_k^{m,1}\cdot e_{\{i_k\},I,=}.
    \end{equation}
    The part in $S_{i_1}$ is thus $\tau_1^{m,m-1}\cdot e_{\{i_1\},I,=}=u_{I,m-1}$. The boundary of $u_{I,m-1}$ is
    \begin{equation}
        \partial u_{I,m-1}=e_{I,m-2}=\sum_{k=2}^m\tau_k^{m,m-2}\cdot e_{\{i_1,i_k\},I,=}
    \end{equation}
    and the part lying in $S_{i_2}$ is
    \begin{equation}
        \tau_2^{m,m-2}\cdot e_{\{i_1,i_2\},I,=}=u_{I,m-2}.
    \end{equation}
    By induction, we obtain
    \begin{equation}
        (\partial_{i_m}\circ\cdots\circ\partial_{i_1})(\mathbf{e}_I)=u_{I,m}=e_{I,m},
    \end{equation}
    where we used $\tau_1^{m,0}=1$.
\end{proof}
Now let us consider the union of complex $n$-spheres $S_\mathbb{C}^n(a_1,r_1),\ldots,S_\mathbb{C}^n(a_{n+1},r_{n+1})$ and suppose that a subset of these spheres, indexed by $I=\{i_1,\ldots,i_m\}\subset\{1,\ldots,n+1\}$ with $i_1<\cdots<i_m$, exhibits a simple pinch at $(a_p',r_p')\in(\mathbb{C}^{n+1})^{|I|}\times(\mathbb{C}^\times)^{|I|}$. Let $(a_p,r_p)\in(\mathbb{C}^{n+1})^{|I|}\times(\mathbb{C}^\times)^{|I|}$ be a point near $(a_p',r_p')$ such that the vanishing classes are defined at $(a_p,r_p)$. Denote the vanishing cell (resp. sphere and cycle) by $\mathbf{e}$ (resp. $\tilde{e}$ and $e$). Let $\sigma:\mathbb{C}^{n+1}\times[0,1]\to\mathbb{C}^{n+1}$ be an ambient isotopy of $\bigcup_{i=1}^{n+1}S_\mathbb{C}^n(a,r)$ in $\mathbb{C}^{n+1}$ induced by a path $\gamma$ from $(a,r)$ to $(a_p,r_p)$ and write $\sigma_1:=\sigma(\cdot,1)$. Then we have the following result:
\begin{prop}
    We have
    \begin{equation}
        \langle\tilde{e}|\mathbf{e}\rangle=\langle\tilde{e}|\sigma_1(\mathbf{e}_I)\rangle=\begin{cases} 2\cdot(-1)^\frac{(n+1)(n+2)}{2} & \text{if }I=J\text{ and }n+1-|I|\text{ even} \\ 0 & \text{otherwise}.\end{cases}
    \end{equation}
\end{prop}
\begin{proof}
    The result for $\langle\tilde{e}|\mathbf{e}\rangle$ is a general result which can be found in \cite{pham} for example. It follows from the same simple computation we use to establish the result for $\langle\tilde{e}|\sigma_1(\mathbf{e}_I)\rangle$. Let $e$ be the vanishing cycle for the given pinch. Using equation \eqref{eq:intersection_dualty_2}, we compute\footnote{The reader is reminded that in this context, the $\delta_j$ are the Leray coboundary as in equation \ref{eq:iterated_coboundary}, not the Mayer-Vietoris homomorphisms.}
    \begin{equation}\label{eq:vanishing_sphere_index}
        \langle\tilde{e}|\sigma(\mathbf{e}_I,1)\rangle=\langle(\delta_{i_1}\circ\cdots\circ\delta_{i_m})(e)|\sigma_1(\mathbf{e}_I)\rangle=(-1)^{m\cdot(k+\frac{m+1}{2})}\cdot\langle e|\sigma_1((\partial_{i_m}\circ\cdots\circ\partial_{i_1})(\mathbf{e}_I))\rangle.
    \end{equation}
    Now, applying the iterated boundary operator $\partial_m\circ\cdots\circ\partial_1$ to $\mathbf{e}_I$ yields $e_{I,|I|}$ according to Lemma \ref{lem:iterated_boundary}. If $I\neq J$, then $\sigma_1(e_{I,|I|})$ must be disjoint from $e$, otherwise it is equal to $e$ and we can compute the intersection index by Proposition \ref{prop:intersection_index_sphere}: If $k$ is odd, then the expression \eqref{eq:vanishing_sphere_index} vanishes. If $k$ is even, we can drop $m\cdot k$ from the exponent of -1 and get
    \begin{equation}
        \langle\tilde{e}|\sigma(\mathbf{e}_I,1)\rangle=2\cdot(-1)^{\frac{m(m+1)}{2}+\frac{k}{2}}=2\cdot(-1)^{\frac{(m+k)(m+k+1)}{2}}=2\cdot(-1)^\frac{(n+1)(n+2)}{2}.
    \end{equation}
\end{proof}
To finish this article, we give an application of the above results to the analytical structure of a simple Feynman integral.
\begin{example}
    Let $D\in\mathbb{N}$ and consider the integral\footnote{This is basically the integral \eqref{eq:feynman_integral} from the Introduction \ref{sec:introduction} in the \enquote{finite} chart. We adopted a slightly different notation and applied an obvious change of variables to conform to the physics literature.}
    \begin{equation}\label{eq:bubble}
        \int_{\mathbb{R}^D}\frac{d^Dk}{(k^2+m_1^2)((p-k)^2+m_2^2)}
    \end{equation}
    with $m_1,m_2\in\mathbb{R}_{>0}$ and $p\in\mathbb{R}^D$. It converges absolutely for $D<4$. For this example, we consider only the dependence of the so called external momentum $p$ and leave the so called masses $m_1$ and $m_2$ fixed. Its Landau surface, i.e. the set of points to which the function defined by the integral \eqref{eq:bubble} can not be analytically continued to, is known to be   
    \begin{equation}
        L=\{p\in\mathbb{C}^D \;|\; -p^2=(m_1\pm m_2)^2 \text{ or }p=0\}.
    \end{equation}
    Any point $p\in L$ with $p\neq0$ corresponds to a simple pinch, involving both surfaces
    \begin{equation}
        S_1=\{k\in\mathbb{C}^D \;|\; k^2+m_1^2=0\} \quad\text{and}\quad S_2=\{k\in\mathbb{C}^D \;|\; (p-k)^2+m_2^2=0\}.
    \end{equation}
    We do not perform the computation here, but it can be shown that the vanishing cell $\mathbf{e}_+$ for a point $p_+\in L$ with $-p_+^2=(m_1+m_2)^2$ is $\pm(\mathbf{e}_{\{1\}}-\mathbf{e}_{\{1,2\}})$ (transported by an appropriate ambient isotopy). The overall sign is up to choice and we set it to $+$ for concreteness. Therefore, the intersection index is $\langle(i\cdot\mathbb{R})^D|\mathbf{e}_+\rangle=1$. Similarly, the vanishing cell $\mathbf{e}_-$ for a point $p_-\in L$ with $-p_-^2=(m_1-m_2)^2$ is either $\pm\mathbf{e}_{\{1,2\}}$ or $\pm(\mathbf{e}_{\{1\}}-\mathbf{e}_{\{2\}}+\mathbf{e}_{\{1,2\}})$, so that $\langle(i\cdot\mathbb{R})^D|\mathbf{e}_-\rangle=0$. Let us denote the corresponding vanishing spheres by $\tilde{e}_\pm$. It is known (see \cite{app-iso}) that tracing a small circle around $p_\pm$ sends the Borel-Moore homology class of $(i\cdot\mathbb{R})^D$ to
    \begin{equation}
        (i\cdot\mathbb{R})^D+(-1)^\frac{(D+1)(D+2)}{2}\cdot\langle(i\cdot\mathbb{R})^D|\mathbf{e}_\pm\rangle\cdot\tilde{e}_\pm.
    \end{equation}
    Thus, tracing a small circle around $p_-$ leaves the integral \eqref{eq:bubble} unchanged while tracing a small circle around $p_+$ results in a change of \cite{cutkosky1}
    \begin{equation}
        \begin{split}
            &(-1)^\frac{(D+1)(D+2)}{2}\cdot\int_{\tilde{e}_+}\frac{dk}{(k^2+m_1^2)((k-p)^2+m_2^2)}\\
            =&(-1)^{\frac{(D+1)(D+2)}{2}}\cdot\int_{\mathbb{R}^D}\delta(k^2+m_1^2)\cdot\delta((p-k)^2+m_2^2)d^Dk\\
            =&\frac{(-1)^{\frac{(D+1)(D+2)}{2}}\cdot\pi^\frac{D+3}{2}}{2^{D-4}\cdot\Gamma(\frac{D-1}{2})}\cdot\frac{(\lambda(-p^2,m_1^2,m_2^2))^\frac{D-3}{2}}{(-p^2)^\frac{D-2}{2}},
        \end{split}
    \end{equation}
    where we used the Residue Theorem due to Leray in the first step. Moreover, we can understand what happens beyond the principal branch: Suppose we already went around $p_+$ once and now we want to trace another circle around $p_\pm$. Then $(i\cdot\mathbb{R})^D+(-1)^{\frac{(D+1)(D+2)}{2}}\cdot\tilde{e}_+$ becomes
    \begin{equation}
        (i\cdot\mathbb{R})^D+(-1)^{\frac{(D+1)(D+2)}{2}}\cdot\tilde{e}_++(-1)^\frac{(D+1)(D+2)}{2}\cdot\langle(i\cdot\mathbb{R})^D|\mathbf{e}_\pm\rangle\cdot\tilde{e}_\pm+\langle\tilde{e}_+|\mathbf{e}_\pm\rangle\cdot\tilde{e}_\pm
    \end{equation}
    For $p_+$, we obtain
    \begin{equation}
        \begin{cases}   (i\cdot\mathbb{R})^D & \text{if }D\text{ is even} \\
        (i\cdot\mathbb{R})^D+4\cdot(-1)^{\frac{(D+1)(D+2)}{2}}\cdot\tilde{e}_+ & \text{if }D\text{ is odd}.
        \end{cases}
    \end{equation}
    So the function has a square-root behaviour if $D$ is even and a logarithmic behaviour otherwise. For $p_-$, we obtain
    \begin{equation}
        \begin{cases}   (i\cdot\mathbb{R})^D+(-1)^{\frac{(D+1)(D+2)}{2}}\cdot\tilde{e}_++2\cdot(-1)^{\frac{D(D+1)}{2}}\cdot\tilde{e}_-) & \text{if }D\text{ is even} \\
        (i\cdot\mathbb{R})^D+(-1)^{\frac{(D+1)(D+2)}{2}}\cdot\tilde{e}_+ & \text{if }D\text{ is odd}.
        \end{cases}
    \end{equation}
\end{example}

\section{Conclusion and Outlook}\label{sec:conclusion}
By constructing a decomposition of the relative homology groups $H_k(\mathbb{C}^{n+1},\bigcup_{i=1}^{n+1}S_i)$ into a direct sum, it is possible to find explicit generators for these groups in terms of a CW-complex. Having these generators at our disposal, it is possible to compute the relevant intersection indices. This is an important step on the way to understand Cutkosky's Theorem and concludes a big chunk of the algebro-topological part of the problem. We have seen in the concluding example how the results of this article can be applied to understand the analytic structure of a Feynman integral. To solve the massive one-loop case completely, the two remaining steps are the following: First, the vanishing cell must be expressed as a linear combination of the found generators $\mathbf{e}_I$. Second, the setup from \cite{cutkosky1} must be generalized so that it covers even dimensions. Both of these points will be addressed in the follow-up to the mentioned article.

\section*{Acknowledgement}
I would like to express my gratitude towards all the people in my life who provide me with continuous support. In particular, I want to thank Dirk Kreimer and Marko Berghoff for guidance and valuable discussion, Chris Wendl, my family and my dear friends Christian Hammermeister and Nora Tabea Sibert.

\end{document}